\documentclass[smallextended,11pt,psfig,a4,reqno]{amsart}
\usepackage{mleftright}

\usepackage{mathrsfs}
\usepackage{geometry}
\usepackage{graphics}
\usepackage[pdftex]{hyperref}
\usepackage[dvips]{graphicx}
\usepackage{pgf}
\usepackage{tikz}

\tikzstyle{vertex}=[circle, draw, inner sep=0pt, minimum size=4pt]

\usepackage{amssymb}
\usepackage[normalem]{ulem}
\usepackage{xcolor,soul}

\usetikzlibrary{arrows,shapes,decorations,automata,backgrounds,petri,positioning}

\tikzset{main node/.style={circle,fill=blue!20,draw,minimum size=1cm,inner sep=0pt},
            }

\newtheorem{theorem}{Theorem}[section]
\newtheorem{lemma}[theorem]{Lemma}
\newtheorem{pro}[theorem]{Proposition}
\newtheorem*{conj*}{Conjecture}

\newtheorem{remark}[theorem]{Remark}

\theoremstyle{definition}

\newtheorem{example}[theorem]{Example}

\theoremstyle{remark}

\numberwithin{equation}{section}

\begin{document}
\pagestyle{plain}

\def\beq{\begin{equation}}
\def\eeq{\end{equation}}
\def\eps{\epsilon}
\def\laa{\langle}
\def\raa{\rangle}
\def\qed{\begin{flushright} $\square$ \end{flushright}}
\def\qee{\begin{flushright} $\Diamond$ \end{flushright}}
\def\ov{\overline}
\def\bma{\begin{bmatrix}}
\def\ema{\end{bmatrix}}

\def\ora{\overrightarrow}

\def\bma{\begin{bmatrix}}
\def\ema{\end{bmatrix}}
\def\bex{\begin{example}}
\def\eex{\end{example}}
\def\beq{\begin{equation}}
\def\eeq{\end{equation}}
\def\eps{\epsilon}
\def\laa{\langle}
\def\raa{\rangle}
\def\qed{\begin{flushright} $\square$ \end{flushright}}
\def\qee{\begin{flushright} $\Diamond$ \end{flushright}}
\def\ov{\overline}

\author[ag]{F. A. Gr\"unbaum}
\address{Department of Mathematics, University of California, Berkeley, CA 94720, USA}
\email{grunbaum@math.berkeley.edu}

\author[cfelipe]{C. F. Lardizabal}
\address{Instituto de Matem\'atica e Estat\'istica, Universidade Federal do Rio Grande do Sul, Porto Alegre, RS  91509-900 Brazil.}
\email{cfelipe@mat.ufrgs.br}

\author[lv]{L. Vel\'azquez}
\address{Departamento de Matem\'atica Aplicada \& IUMA, Universidad de Zaragoza, Mar\'ia de Luna 3,
50018 Zaragoza, Spain.}
\email{velazque@unizar.es}

\date{\today}

\title{Quantum Markov chains: recurrence, Schur functions and splitting rules}

\begin{abstract}
In this work we study the recurrence problem for quantum Markov chains, {\color{black}{ which are}} quantum versions of classical Markov chains introduced by S.~Gudder and described in terms of completely positive maps. A notion of monitored recurrence for quantum Markov chains is examined in association with Schur functions, which codify information on the first return to some given state or subspace. Such objects possess important factorization and decomposition properties which allow us to obtain probabilistic results based solely on those parts of the graph where the dynamics takes place, the so-called splitting rules. These rules also yield an alternative to the folding trick to transform a doubly infinite system into a semi-infinite one which doubles the number of internal degrees of freedom. The generalization of Schur functions --so-called FR-functions-- to the general context of closed operators in Banach spaces is the key for the present applications to open quantum systems. An important class of examples included in this setting are the open quantum random walks, as described by S.~Attal et al., but we will state results in terms of general completely positive trace preserving maps. We also take the opportunity to discuss basic results on recurrence of finite dimensional iterated quantum channels and quantum versions of Kac's Lemma, in close association with recent results on the subject.
\end{abstract}

\setcounter{MaxMatrixCols}{15}

\maketitle

\tableofcontents

\newpage

\section{Introduction}

The focus of this work is on the study of recurrence --i.e. return properties-- for quantum Markov chains \cite{gudder}, {\color{black} which are} quantum versions of classical Markov chains described in terms of completely positive (CP) and trace preserving (TP) maps \cite{benatti,paulsen}. CPTP maps, also called quantum channels, are central in the theory of open quantum systems \cite{wolf}. It turns out that CP maps are not only just an ingredient of quantum Markov chains, but quantum Markov chains can be viewed as the result of iterating CPTP maps (see Subsect.~\ref{ssec:pre}). Due to this reason, we will present the recurrence problem for general iterated CPTP maps, since they constitute a simpler and at the same time a wider setting to study return properties.

In quantum theory there are in principle several nonequivalent notions of recurrence and in this work we follow a recent trend of considering a so-called monitored recurrence for discrete-time quantum systems \cite{ambainis,bourg,cgl,gvfr,werner,brun,ls2015,gawron,sinkovicz,sinkovicz2}: the return properties are defined through a process in which the system is monitored after each time step via a projective measurement. This measurement only checks whether the system returns to a required state --or subspace of states-- or not, so that quantum coherence is only partially destroyed and non-classical effects arise \cite{werner,bourg}.

As in P\'olya's classical recurrence theory, generating functions have proven to be an invaluable tool in quantum recurrence. The generating function approach links recurrence to well known objects in classical analysis known as Schur functions \cite{werner,bourg}. Schur functions are the analytic maps of the open unit disk into its closure and, just as their matrix and operator valued versions, they are central elements in harmonic analysis providing fruitful connections with other areas of mathematics and its applications. 

An extension of the notion of Schur function known as first-return function (FR-function for short) has revealed to be the appropriate formalism unifying the generating function approach to recurrence in classical Markov chains, unitary quantum walks and iterated quantum channels \cite{gvfr}. FR-functions are operator valued functions {\color{black} which can be written as} 
\beq \label{eq:FR}
f(z)=PT(1-zQT)^{-1}P,
\qquad Q=I-P,
\eeq
where $T$ is a linear operator on a Banach space and $P$ is a projection onto a closed subspace. When $T$ is an evolution operator and the complementary projection $Q$ is interpreted as an operator conditioning on the event ``no return", then the FR-function $f$ codifies the information about the return properties of the subspace where $P$ projects \cite{werner,bourg,gvfr} (see also Sect.~\ref{sec:REC}). 

The first link between recurrence and Schur functions appeared in \cite{werner}, where the FR-functions encoding the return properties of a state subject to a unitary discrete evolution were identified as Schur functions. This result was extended in \cite{bourg} to the return properties of a subspace, in which case the corresponding FR-functions become matrix valued Schur functions. The identification of the FR-functions for classical Markov chains as Schur functions only came later \cite{gvfr}, and required the extension of the previous results on unitaries in Hilbert spaces to operators in Banach spaces defined by stochastic matrices. 

The generalization of FR-functions to the setting of operators in Banach spaces  also paved the way to the discovery of the relation between FR-functions and recurrence in iterated quantum channels \cite{gvfr}. However, no interpretation as Schur functions is known so far. As a first original result of this paper we will prove that the FR-functions related to recurrence in iterated quantum channels and quantum Markov chains are again Schur functions built out of operators in Banach spaces, defined in this case by CPTP maps. The proof of this result relies on a non-trivial property of CPTP maps (see Theorem~\ref{thm:SCHUR}). 

FR-functions satisfy splitting properties which extend useful factorizations of Schur functions of interest for the study of recurrence in unitary quantum walks (UQWs) \cite{cgvww}. They provide recurrence splitting rules which reduce the return properties of a large system to those of some subsystems \cite{gvfr}. Apart from their role as a divide and conquer technique, these splitting properties have striking consequences, such as the invariance of the return probability under certain local perturbations \cite{gvfr} (see also the comments in \cite[Sect.~6]{werner} and \cite[Sect.~4]{bourg}). The possibility of exploiting these splitting methods for the analysis of recurrence in iterated quantum channels was pointed out in \cite{gvfr}, where such splitting techniques were merely illustrated by an example which deals with a special instance of iterated quantum channels, known as open quantum walks (OQWs){, introduced for the first time in \cite{attal}. Although a particular case of quantum Markov chains, OQWs were the main motivation behind several results in recent years regarding the dynamical and statistical properties of dissipative quantum walks on graphs, see \cite{spsurvey} for a recent survey on the subject. Regarding hitting probabilities and recurrence in the setting of OQWs, see \cite{bardet,cgl,ls2015,gawron}.}

This work provides a detailed study of the consequences of the splitting properties for FR-functions regarding recurrence in quantum Markov chains. Two kinds of FR-function splittings, related to factorizations and decompositions into sums of the underlying operator, yield two types of splitting rules for quantum Markov chains and, thus, for the particular case of classical Markov chains. As a consequence of these splitting rules we will see that, similarly to UQWs, the return probability is also invariant under certain local perturbations.

Other novel contributions of this paper deal with generalizations of known results on recurrence in classical Markov chains and UQWs. Different quantum generalizations of Kac's lemma for classical Markov chains are commented in Subsect.~\ref{ssec:kac}. Besides, Theorem~\ref{thm:rec-fi} constitutes the quantum version of a well known result for classical Markov chains, namely, that finiteness and irreducibility imply the positive recurrence of every state, i.e. every state returns to itself with probability one and in a finite expected return time. A similar result holds for finite-dimensional unitary evolutions, which in addition present the striking particularity of exhibiting integer valued expected times for the return to any state \cite{werner}. These results has been generalized to the return to a subspace in \cite{bourg}, while \cite{sinkovicz} proves that they hold for the return to a state in iterated quantum channels whenever they are unital. We extend this property of unital quantum channels to the return to a subspace in Theorem~\ref{thm:rec-ui}, following the ideas developed in \cite{bourg} for UQWs.

The structure of the present paper is as follows: the rest of this introduction is devoted to a summary on CP maps, quantum Markov chains and OQWs, which clarifies the relationships among them. In Section~\ref{sec:REC} we introduce the notion of monitored recurrence for iterated CPTP maps, formulating it in terms of FR-functions and Schur functions. For quantum channels acting on finite dimensional Hilbert spaces we prove that irreducibility implies positive recurrence of every subspace, while unitality adds to this the quantization of the average over the subspace of the mean return time --it is a multiple of the inverse of the dimension of the subspace. Also, different quantum versions of Kac's Lemma are discussed. The recurrence splitting rules for quantum Markov chains arising from the splitting properties of FR-functions are analyzed in Section~\ref{sec:SPLIT}, exploiting them to conclude the independence of the return probability with respect to the details of the evolution in certain subsystems. These splitting rules are originated by similar splittings of the underlying evolution operator. Section~\ref{sec:CHAR} provides practical ways to identify these recurrence splitting rules by characterizing the corresponding operator splittings. Along these sections, the ideas and results are illustrated with finite-dimensional examples. The practical application of these methods to infinite-dimensional examples, such as quantum Markov chains on the line or the half-line, require additional techniques which are introduced and illustrated in Section~\ref{sec:1D}.

\subsection{Preliminaries: CP maps, TOMs and OQWs}
\label{ssec:pre}
 
Let $\mathcal{H}$ be a separable Hilbert space with inner product $\langle\,\cdot\,|\,\cdot\,\rangle$, whose closed subspaces will be referred to  as subspaces for short. The superscript ${}^*$ will denote the adjoint operator. The Banach algebra $\mathcal{B}(\mathcal{H})$ of bounded linear operators on $\mathcal{H}$ is the topological dual of its ideal $\mathcal{I}(\mathcal{H})$ of trace-class operators with trace norm
$$
\|\rho\|_1=\operatorname{Tr}(|\rho|),
\qquad 
|\rho|=\sqrt{\rho^*\rho},
$$
through the duality \cite[Lec. 6]{attal_lec}
\beq \label{eq:dual}
\langle \rho,X \rangle = \operatorname{Tr}(\rho X),
\qquad
\rho\in\mathcal{I}(\mathcal{H}),
\qquad 
X\in\mathcal{B}(\mathcal{H}).
\eeq
If $\dim\mathcal{H}=k<\infty$, then $\mathcal{B}(\mathcal{H})=\mathcal{I}(\mathcal{H})$ is identified with the set of square matrices of order $k$, denoted $M_k(\mathbb{C})$. The duality \eqref{eq:dual} yields a useful characterization of the positivity of an operator $\rho\in\mathcal{I}(\mathcal{H})$,
\beq \label{eq:pos-dual}
\rho\in\mathcal{I}(\mathcal{H}): 
\quad
\rho\ge0 \; \Leftrightarrow \; \operatorname{Tr}(\rho X)\ge0,
\quad
\forall X\in\mathcal{B}(\mathcal{H}), 
\quad 
X\ge0,
\eeq
and similarly for the positivity of $X\in\mathcal{B}(\mathcal{H})$.
Given a linear map $\Phi$ on $\mathcal{I}(\mathcal{H})$, its dual $\Phi^*$ on $\mathcal{B}(\mathcal{H})$ is defined by
$$
\langle \Phi(\rho),X \rangle = \langle \rho,\Phi^*(X) \rangle,
\qquad
\rho\in\mathcal{I}(\mathcal{H}),
\qquad 
X\in\mathcal{B}(\mathcal{H}).
$$ 
We refer the reader to \cite{attal_lec,benatti,bhatia,nielsen,wolf} for the definition of completely positive (CP) maps. For our purposes it suffices to recall that a CP map $\Phi$ on $\mathcal{I}(\mathcal{H})$ and its dual $\Phi^*$ on $\mathcal{B}(\mathcal{H})$ can be written {\color{black}{ in Kraus form \cite{kraus}}}
\begin{equation} \label{krausform}
\Phi(\rho)=\sum_i B_i\rho B_i^*, \qquad\qquad \Phi^*(X)=\sum_i B_i^* X B_i,
\end{equation}
where the set $\{B_i\}\subset\mathcal{B}(\mathcal{H})$ is countable. When this set is countably infinite $\sum_i B_i^*B_i$ converges strongly in $\mathcal{B}(\mathcal{H})$, so that the first sum in \eqref{krausform} converges with respect to the trace norm and the second one also converges strongly, both sums being unconditional \cite{kraus}. We refer to $\{B_i\}$ as a set of {\bf Kraus operators} for $\Phi$, or as an {\bf unravelling} of $\Phi$. Sometimes we will write in condensed form $\Phi = \sum_i B_i \cdot B_i^*$.

We will {\color{black} often} restrict the action of $\Phi$ to the subset of {\bf states}, represented by the {\bf density operators},
\beq 
\mathcal{D}(\mathcal{H}) =
\{\rho\in \mathcal{I}(\mathcal{H}): \rho\geq 0,\; \operatorname{Tr}(\rho)=1\},
\eeq
where $\rho\geq 0$ means that the operator is positive semidefinite (positive, for short). 

This restriction determines the whole linear map $\Phi$ on $\mathcal{I}(\mathcal{H})$ due to the positive decomposition of a trace-class operator, which comes from combining the standard decomposition {\color{black} in terms of} the self-adjoint real and imaginary parts with the decomposition of a self-adjoint trace-class operator $\rho=\rho_+-\rho_-$ {\color{black} in terms of} positive ones $\rho_\pm=\frac{1}{2}(|\rho|\pm \rho)$. 
A state is {\bf pure} if it is a projection $\rho_\psi:=|\psi\rangle\langle\psi|$ for some unit vector $\psi\in\mathcal{H}$, so eventually we will identify the set $\mathcal{D}_p(\mathcal{H})$ of pure states with the set of unit vectors of $\mathcal{H}$ up to phases. A state represented by a strictly positive density $\rho>0$ is called {\bf faithful}. A quantum channel is a CP map $\Phi$ on $\mathcal{I}(\mathcal{H})$ which is {\bf trace-preserving} (CPTP), i.e. $\operatorname{Tr}(\Phi(\rho))=\operatorname{Tr}(\rho)$ for all $\rho\in\mathcal{I}(\mathcal{H})$, a condition which, in view of \eqref{eq:pos-dual}, is equivalent to $\sum_i B_i^*B_i=I$. We say that $\Phi$ is {\bf unital} if $\Phi(I)=I$, i.e. $\sum_i B_iB_i^*=I$, but we should bear in mind that $\Phi(I)$ could make no sense if $\dim\mathcal{H}=\infty$. 

For any CP map $\Phi$ on $\mathcal{I}(\mathcal{H})$, the norm $\|\Phi\|=\sup_{\rho\in\mathcal{I}(\mathcal{H})\setminus\{0\}}\|\Phi(\rho)\|_1/\|\rho\|_1$ may be obtained as \cite{kraus} 
\beq \label{eq:CPnorm}
\|\Phi\|=\sup_{\rho\in\mathcal{D}(\mathcal{H})}\operatorname{Tr}(\Phi(\rho)),
\eeq 
the supremum being a maximum if $\dim\mathcal{H}<\infty$. Hence, $\|\Phi\|=1$ if $\Phi$ is CPTP. Norm convergence of CP maps results too restrictive, thus the convergence of a sequence of CP maps $\Phi_n$ will be understood in the strong sense with respect to the trace norm, i.e. as the convergence of $\Phi_n(\rho)$ for any $\rho\in\mathcal{I}(\mathcal{H})$, which only needs to be checked for $\rho\ge0$ due to the positive decomposition of trace-class operators. In consequence, the convergence of a series $\sum_n\Phi_n$ of CP maps is characterized by the convergence of $\sum_n\operatorname{Tr}(\Phi_n(\rho))$ for any $\rho\ge0$. The latter is a series of non-negative terms, thus convergent series of CP maps are unconditionally convergent and infinite sums of CP maps with respect to different indices may be exchanged. We will implicitly use these results in what follows.

The vector representation $vec(A)$ of $A\in M_k(\mathbb{C})$, given by stacking together its rows, will be a useful tool. For instance,
\begin{equation}
A = \begin{bmatrix} a_{11} & a_{12} \\ a_{21} & a_{22} \end{bmatrix}
\quad\equiv\quad
vec(A):=\begin{bmatrix} a_{11} \\ a_{12} \\ a_{21} \\ a_{22}\end{bmatrix}.
\end{equation}
The $vec$ mapping satisfies $vec(AXB^T)=(A\otimes B)\,vec(X)$ {\cite{hj2}} for any square matrices $A, B, X$, with $\otimes$ denoting the Kronecker product. In particular, $vec(BXB^*)=vec(BX\ov{B}^T)=(B\otimes \ov{B})\,vec(X)$,
from which we can obtain the \textbf{matrix representation} $\widehat\Phi$ for a CP map ({\ref{krausform}})  when the underlying Hilbert space $\mathcal{H}$ is finite-dimensional:
\begin{equation}\label{matrep}
\widehat\Phi = \sum_{i} \lceil B_{i} \rceil, 
\qquad \lceil B \rceil := B \otimes \ov{B}.
\end{equation}
Here the operators $B_i$ are identified with some matrix representation. We have that $\lceil B \rceil^* = \lceil B^*\rceil $, where $B^*$ denotes the hermitian transpose of a matrix $B$.

\medskip

Following S. Gudder \cite{gudder}, a {\bf quantum Markov chain} on a graph with a finite or infinite set $V$ of vertices and a Hilbert space $\mathcal{H}$ of internal degrees of freedom may be described in terms of CP maps. The state of the system is described by a column vector
$$
\rho = \begin{bmatrix} \rho_0 \\ \rho_1 \\ \rho_2 \\ \vdots \end{bmatrix},
\qquad \rho_i\in\mathcal{I}(\mathcal{H}), 
\qquad \rho_i\ge0, 
\qquad \sum_{i\in V}\operatorname{Tr}(\rho_i)=1.
$$
After one time step, the system evolves to the state $\mathcal{E}(\rho)$ given by $\mathcal{E}(\rho)_i=\sum_{j\in V}\mathcal{E}_i^j(\rho_j)$, where 
$$
\mathcal{E} =
\begin{bmatrix}
\\[-10pt]
\mathcal{E}_0^0 & \mathcal{E}_0^1 & \mathcal{E}_0^2 & \dots
\\[2pt]
\mathcal{E}_1^0 & \mathcal{E}_1^1 & \mathcal{E}_1^2 & \dots
\\[2pt]
\mathcal{E}_2^0 & \mathcal{E}_2^1 & \mathcal{E}_2^2 & \dots
\\
\dots & \dots & \dots & \dots 
\end{bmatrix}
$$ 
is a {\bf Transition Operator Matrix} (TOM) \cite{gudder}, i.e. such that $\mathcal{E}_i^j$ are CP maps on $\mathcal{I}(\mathcal{H})$ and the column sums $\sum_{i\in V} \mathcal{E}_i^j$ are trace preserving (again, the summations are assumed to converge in the strong operator topology). 

Using an auxiliary Hilbert space $\mathcal{S}$ with an orthonormal basis $\{|i\rangle\}_{i\in V}$ in one-to-one correspondence with the set $V$ of vertices, we can see $\mathcal{E}$ equivalently as the positive operator
\beq \label{eq:TOM}
\mathcal{E}(\rho)=\sum_{i\in V} 
\bigg(\sum_{j\in V}\mathcal{E}_i^j(\rho_j)\bigg) \otimes |i\rangle\langle i|,
\eeq
acting on the Banach subspace of $\mathcal{I}(\mathcal{H}\otimes\mathcal{S})$ given by
$$
\mathcal{I}_\mathcal{S}(\mathcal{H}) :=
\{\rho\in\mathcal{I}(\mathcal{H}\otimes\mathcal{S}) : 
\textstyle \rho=\sum_{i\in V}\rho_i\otimes|i\rangle\langle i|\} =
\bigoplus_{i\in V} \mathcal{I}(\mathcal{H})\otimes|i\rangle\langle i|.
$$
For every $\rho\in\mathcal{I}_\mathcal{S}(\mathcal{H})$ we have that $\operatorname{Tr}(\rho)=\sum_{i\in V}\operatorname{Tr}(\rho_i)$ and the trace preserving condition for the column sums becomes $\operatorname{Tr}(\mathcal{E}(\rho))=\operatorname{Tr}(\rho)$. Then, the set of states is
$$
\mathcal{D}_\mathcal{S}(\mathcal{H}) = 
\{\rho\in\mathcal{I}_\mathcal{S}(\mathcal{H}):\rho\ge0,\;\operatorname{Tr}\rho=1\},
$$
and we will refer to any of its elements as a {\bf TOM density}. For convenience, we will usually refer to the subspace $\mathcal{H}\otimes|i\rangle \subset \mathcal{H}\otimes\mathcal{S}$ as site $|i\rangle$. The dynamics obtained by iterating $\mathcal{E}$ will be called a {\color{black} \bf{quantum Markov chain on $\mathcal{H}\otimes\mathcal{S}$}}.

Any TOM \eqref{eq:TOM} can be alternatively seen as a CPTP map on the whole space $\mathcal{I}(\mathcal{H}\otimes\mathcal{S})$ of the form
\beq \label{eq:TOM2}
\Phi = \sum_{i,j\in V} \mathcal{E}_i^j \otimes 
|i\rangle \langle j| \cdot |j\rangle \langle i|.
\eeq
Given any element $\rho=\sum_{i,j\in V}\rho_{ij}\otimes|i\rangle\langle j|\in\mathcal{I}(\mathcal{H}\otimes\mathcal{S})$, 
$$
\Phi(\rho) = \mathcal{E}(\tilde\rho),
\qquad
\tilde\rho = \sum_{i\in V} \rho_{ii}\otimes|i\rangle\langle i|, 
$$
i.e., $\Phi=\mathcal{E}\mathcal{T}$, where $\mathcal{T}\colon\mathcal{I}(\mathcal{H}\otimes\mathcal{S})\to\mathcal{I}_\mathcal{S}(\mathcal{H})$ is the operator $\mathcal{T}(\rho)=\tilde\rho$ which extracts from every element of $\mathcal{I}(\mathcal{H}\otimes\mathcal{S})$ its block diagonal part. This shows that the iterations of $\Phi$ are naturally restricted to $\mathcal{I}_\mathcal{S}(\mathcal{H})$ because $\Phi(\mathcal{I}(\mathcal{H}\otimes\mathcal{S}))\subset\mathcal{I}_\mathcal{S}(\mathcal{H})$, so that the walker moves in $\mathcal{I}_\mathcal{S}(\mathcal{H})$ after the first step. Since $\Phi$ acts as $\mathcal{E}$ on $\mathcal{I}_\mathcal{S}(\mathcal{H})$, these iterations reproduce the quantum Markov chain driven by \eqref{eq:TOM}. Nevertheless, working with the CPTP map $\Phi$ on $\mathcal{I}(\mathcal{H}\otimes\mathcal{S})$ has some advantages. On one hand this allows us to directly extend results on CPTP maps to TOMs. On the other hand, this facilitates the combination of TOMs with projective measurements, something crucial for the monitored notion of recurrence that we will discuss later on. 

If the system is in a state $\rho\in\mathcal{D}_\mathcal{S}(\mathcal{H})$, Born's rule states that for a minimal measurement keeping track only of whether the system is in a subspace $\mathcal{K}\subset\mathcal{H}\otimes\mathcal{S}$ or not, the probability of a positive outcome is $\operatorname{Tr}(P\rho P)$, where $P$ is the orthogonal projection of $\mathcal{H}\otimes\mathcal{S}$ onto $\mathcal{K}$. After such a positive outcome the system collapses to the state $P\rho P/\operatorname{Tr}(P\rho P)\in\mathcal{D}(\mathcal{H}\otimes\mathcal{S})$, but this state does not necessarily fall on $\mathcal{I}_\mathcal{S}(\mathcal{H})$, making impossible to iterate a quantum Markov chain governed by a TOM \eqref{eq:TOM} after a measurement. This drawback disappears {\color{black} if we consider measurements that simply check the situation at every site}, which amounts to dealing only with what we will call {\bf admissible subspaces} of pure states, given by
\beq \label{eq:adm}
\mathcal{K} = \bigoplus_{i\in V}\mathcal{H}_i\otimes|i\rangle,
\qquad
\mathcal{H}_i\subset\mathcal{H}.
\eeq 
Then, $P=\oplus_{i\in V}P_i\otimes|i\rangle\langle i|$, where $P_i$ is the orthogonal projection of $\mathcal{H}$ onto the subspace $\mathcal{H}_i$. In this case, $P(\sum_{i\in V}\rho_i\otimes|i\rangle\langle i|)P=\sum_{i\in V}P_i\rho_iP_i\otimes|i\rangle\langle i|$, so that the measurement of the observable $P$ keeps $\mathcal{D}_\mathcal{S}(\mathcal{H})$ invariant, permitting subsequent iterations of the TOM. Nevertheless, looking at a quantum Markov chain as the iterations of a CPTP map \eqref{eq:TOM2} one is free to combine it with projective measurements on any subspace $\mathcal{K}\subset\mathcal{H}\otimes\mathcal{S}$ because there is no restriction on the densities $\rho\in\mathcal{H}\otimes\mathcal{S}$ where such a CPTP map may act.
 
If $\dim\mathcal{H}=1$, TOMs become column stochastic matrices, so that classical Markov chains are particular cases of  quantum Markov chains, which for $\dim\mathcal{H}=2$ are called (dissipative) quantum walks of one qubit. We can generate TOMs, for instance, starting from quantum channels: if $\Phi_i$ are quantum channels on $\mathcal{H}$, then so is any convex combination $\sum_i p_i\Phi_i$, $p_i\ge0$, $\sum_i p_i=1$. Therefore, if $\Phi_{ij}$ are quantum channels on $\mathcal{H}$ for $i,j\in V$, we have that $[p_{ij}\Phi_{ij}]_{i,j\in V}$ is a TOM on a graph with a set $V$ of vertices, provided that $p_{ij}\ge0$ and $\sum_{i\in V} p_{ij}=1$.

\medskip

For another class of examples, consider the case for which 
\beq\label{def_oqw}
\mathcal{E}_i^j(\rho)=B_i^j\rho {B_i^j}^*,
\qquad
B_i^j\in \mathcal{B}(\mathcal{H}),
\qquad
\sum_{k \in V} {B_k^j}^*B_k^j=I,
\qquad
\forall i,j \in V.
\eeq
The summation above must be understood in the strong sense, and the corresponding identity is the trace preserving condition for the columns of the TOM $\mathcal{E}$. We will say that $B_i^j$ is the {\bf effect} matrix of transitioning from vertex $j$ to vertex $i$. TOMs for which $\mathcal{E}_i^j$ can be written in the form (\ref{def_oqw}) are called {\bf Open Quantum Random Walks} (OQWs), following the terminology established by S. Attal et al. \cite{attal}. In this setting, a TOM density will be alternatively called an {\bf OQW density}.
Explicitly, OQWs are TOMs of the form
\beq\label{eq:OQW}
\mathcal{E}(\rho) =
\sum_{i\in V}\Big(\sum_{j\in V} 
B_i^j\rho_j {B_i^j}^*\Big)\otimes |i\rangle\langle i|,
\eeq 
and, as any TOM, they may be alternatively seen as CPTP maps on $\mathcal{I}(\mathcal{H}\otimes\mathcal{S})$. The Kraus decomposition of the CPTP map corresponding to an OQW \eqref{eq:OQW} is given by \cite{attal}
\beq\label{eq:OQW-CPTP}
\Phi(\rho) = \sum_{i,j\in V} M_j^i \rho {M_j^i}^*,
\qquad
M_i^j = B_i^j \otimes |i\rangle\langle j|,
\qquad
\rho\in\mathcal{I}(\mathcal{H}\otimes\mathcal{S}).
\eeq 

The vector and matrix representation of states and CP maps may be easily adapted to OQWs. In fact, since any element of $\mathcal{I}_\mathcal{S}(\mathcal{H})$ is block diagonal, when $\dim\mathcal{H}<\infty$, it may be represented by combining the vector representations of the finite diagonal blocks,
$$
\rho=\sum_{i\in V} \rho_i\otimes|i\rangle\langle i| 
\quad\equiv\quad
\overrightarrow{\rho}:=\begin{bmatrix} vec(\rho_1) \\ vec(\rho_2) \\ \vdots \end{bmatrix}.
$$
Then, the OQW \eqref{eq:OQW} admits the {\bf block matrix representation}
\beq\label{eq:mrOQW}
\overrightarrow{\mathcal{E}(\rho)} = 
\widehat{\mathcal{E}}\,\overrightarrow{\rho}, 
\qquad 
\widehat{\mathcal{E}} = 
\begin{bmatrix} 
\lceil B_1^1 \rceil & \lceil B_1^2 \rceil & \cdots
\\[2pt] 
\lceil B_2^1 \rceil & \lceil B_2^2 \rceil & \cdots 
\\
\vdots & \vdots 
\end{bmatrix}.
\eeq

OQWs are not just one more instance of TOMs, but the key example, since the Kraus decomposition of the CP maps $\mathcal{E}_i^j$ of a TOM $\mathcal{E}=[\mathcal{E}_i^j]_{i,j\in V}$ expresses them as sums of terms like \eqref{def_oqw}. This means that many {\color{black} constructions} developed originally for OQWs may be easily translated to general TOMs. For instance, from \eqref{eq:OQW-CPTP} we find that the Kraus decomposition $\mathcal{E}_i^j(\rho)=\sum_kB_{i,k}^j\rho B_{i,k}^{j*}$ of the CP maps $\mathcal{E}_i^j$ yields the one of the CPTP representation \eqref{eq:TOM2} of $\mathcal{E}$,
$$
\Phi(\rho) = \sum_{i,j\in V} \sum_k M_{i,k}^j \rho M_{i,k}^{j*},
\qquad
M_{i,k}^j = B_{i,k}^j \otimes |i\rangle\langle j|.
$$
Also, if $\dim\mathcal{H}<\infty$, the block matrix representation \eqref{eq:mrOQW} extends to TOMs as $\widehat{\mathcal{E}}:=[\widehat{\mathcal{E}}_i^j]_{i,j\in V}$ in terms of the matrix representations $\widehat{\mathcal{E}}_i^j=\sum_k\lceil B_{i,k}^j \rceil$ of the CP maps $\mathcal{E}_i^j$.

\medskip

Finally, we recall the important notion of irreducibility of OQWs, following \cite{carbone2,carbone1}.
We say that a positive map $\Phi$ is {\bf irreducible} if the only orthogonal projections $P$ reducing $\Phi$, i.e. such that $\Phi(P \mathcal{I}(\mathcal{H})P)\subset P \mathcal{I}(\mathcal{H})P$, are $P=0$ and $P=I$. 

\medskip

Let $\Phi$ be a CPTP map, so that it can be written as $\Phi(X)=\sum_{i\in V} B_i\rho B_i^*$. Let $\mathbb{C}[B]$ denote the algebra (not the $*$-algebra) generated by the $B_i$, $i\in V$. Then irreducibility is equivalent to any of the following \cite[Lemma 3.7]{carbone1}:
\begin{enumerate}
\item For any $\phi\in\mathcal{H}\setminus\{0\}$, the set $\mathbb{C}[B]\phi$ is dense in $\mathcal{
H}$.
\item For any $\phi,\psi\in\mathcal{H}\setminus\{0\}$, there are $i_1,\dots,i_k\in V$ such that $\langle\psi|B_{i_k}\!\cdots B_{i_1}\phi\rangle\neq 0$.
\end{enumerate}

\medskip

Below we state the basic result on irreducibility which is used in this work.

\begin{theorem}\label{cp_irr} Let $\Phi$ be a CP map on $\mathcal{I}(\mathcal{H})$.
\begin{itemize} 
\item[a)] \cite[Thm. 3.14]{carbone1} If $\Phi$ is irreducible and has an invariant state, then it is unique and faithful. 
\item[b)] \cite[ Remark 3.6]{carbone2} If $\Phi$ admits a unique invariant state and such state is faithful, then $\Phi$ is irreducible.
\end{itemize}
\end{theorem}

When $\dim\mathcal{H}<\infty$, every positive trace preserving map on $\mathcal{I}(\mathcal{H})$ has an invariant state, thus in the finite-dimensional case the irreducibility of a CPTP map is characterized by the existence of a faithful unique invariant state. If $\dim\mathcal{H}=\infty$, the presence of an invariant state for a CPTP map is guaranteed whenever a --not necessarily positive-- fixed point in $\mathcal{I}(\mathcal{H})$ exists, as follows from \cite[Thm. 4.1]{schrader}.

The notion of irreducibility, as well as the above results, are not only applicable to quantum channels, but also to TOMs once they are understood as CPTP maps. 



\section{Recurrence for quantum Markov chains and Schur functions}
\label{sec:REC}

Among the different notions of recurrence appearing in the quantum literature, we will refer here to a recent one based on a {\bf monitoring} process, {\color{black} developed for unitary quantum walks \cite{ambainis,bourg,werner,brun}}, and later extended to open quantum walks \cite{bardet,cgl,gvfr,ls2015}. Consider a discrete time evolution given by iterating a quantum channel $\Phi$ on a Hilbert space $\mathcal{H}$. Given a subspace $\mathcal{H}_0\subset\mathcal{H}$,  we will identify $\mathcal{I}(\mathcal{H}_0)$ with the subspace constituted by those operators $\rho\in\mathcal{I}(\mathcal{H})$ with $\operatorname{ran}\rho\subset\mathcal{H}_0$ and $\ker\rho\supset\mathcal{H}_0^\bot$. We are interested in the probability of eventually finding the system in $\mathcal{H}_0$ when starting the evolution at a state $\rho\in\mathcal{D}(\mathcal{H}_0)$, i.e. a state $\rho\in\mathcal{D}(\mathcal{H})$ whose support $\operatorname{supp}\rho=(\ker\rho)^\bot=\overline{\operatorname{ran}}\,\rho$ lies on $\mathcal{H}_0$. In the spirit of P\'olya's recurrence theory for classical random walks, we check after every step whether the system is found in $\mathcal{H}_0$ or not. According to Born's rule, this amounts to perform in between any two steps a projective measurement corresponding to one of the projections 
\beq\label{eq:PQ}
\mathbb{P}(\rho):=P\rho P,
\qquad
\mathbb{Q}(\rho):=Q\rho Q,
\qquad
\rho\in\mathcal{I}(\mathcal{H}),
\eeq
where $P$ is the orthogonal projection of $\mathcal{H}$ onto $\mathcal{H}_0$ and $Q=I-P$. The operators $\mathbb{P}=P\cdot P$ and $\mathbb{Q}=Q \cdot Q$ are the natural projections of $\mathcal{I}(\mathcal{H})$ onto $\mathcal{I}(\mathcal{H}_0)$ and $\mathcal{I}(\mathcal{H}_0^\perp)$, respectively, related to the following orthogonal decomposition with respect to the inner product $\langle X,Y \rangle=\operatorname{Tr}(X^*Y)$,
$$
\mathcal{I}(\mathcal{H}) = 
\mathcal{I}(\mathcal{H}_0) \oplus
\mathcal{I}(\mathcal{H}_0,\mathcal{H}_0^\bot) \oplus
\mathcal{I}(\mathcal{H}_0^\bot),
$$
where $\mathcal{I}(\mathcal{H}_0,\mathcal{H}_0^\bot)$ stands for the subspace of operators of $\mathcal{I}(\mathcal{H})$ mapping $\mathcal{H}_0$ into $\mathcal{H}_0^\bot$ and vice versa.

The projections $\mathbb{P}$ and $\mathbb{Q}$ condition on the events ``return to $\mathcal{H}_0$'' and ``no return to $\mathcal{H}_0$'', respectively. {\color{black} By starting at a state $\rho_0=\rho\in\mathcal{D}(\mathcal{H}_0)$ and monitoring after each step we obtain a sequence of states $\rho_n\in\mathcal{D}(\mathcal{H})$. A negative outcome for the measurement of the return to $\mathcal{H}_0$ during the first $n-1$ steps } results in the states
$$ 
\rho_k = 
\frac{\mathbb{Q}\Phi(\rho_{k-1})}{\operatorname{Tr}(\mathbb{Q}\Phi(\rho_{k-1}))},
\qquad k=1,2,\dots,n-1,
$$
and each of these negative outcomes holds with probability $\operatorname{Tr}(\mathbb{Q}\Phi(\rho_{k-1}))$. Therefore, the probability of returning to $\mathcal{H}_0$ for the first time in the $n$-th step is given by
$$
\pi_n(\rho\to\mathcal{H}_0) = 
\operatorname{Tr}(\mathbb{Q}\Phi(\rho_0)) \, 
\operatorname{Tr}(\mathbb{Q}\Phi(\rho_1)) \, \cdots \,
\operatorname{Tr}(\mathbb{Q}\Phi(\rho_{n-2})) \,
\operatorname{Tr}(\mathbb{P}\Phi(\rho_{n-1}))  = 
\operatorname{Tr}(\mathbb{P}\Phi(\mathbb{Q}\Phi)^{n-1}(\rho)),
$$ 
so that the {\bf return probability}, i.e. the probability of eventually returning to $\mathcal{H}_0$ regardless of the time step at which this return occurs, becomes
$$
\pi(\rho\to\mathcal{H}_0) = 
\sum_{n\geq 1}\pi_n(\rho\to\mathcal{H}_0). 
$$
When this probability is 1, the {\bf expected return time} is
$$
\tau(\rho\to\mathcal{H}_0) = 
\sum_{n\geq 1}n\pi_n(\rho\to\mathcal{H}_0),
$$ 
otherwise it is defined to be infinity, which corresponds to considering the {\color{black} non-return probability} $1-\pi(\rho\to\mathcal{H}_0)$ as the probability of returning in an infinite time. For convenience, when $\rho=\rho_\psi$ is a pure state we will denote by $\pi(\psi\to\mathcal{H}_0)$ and $\tau(\psi\to\mathcal{H}_0)$ the return probability and expected return time, respectively. The recurrence of a pure state $\rho_\psi$ corresponds to choosing $\mathcal{H}_0=\operatorname{span}\{\psi\}$, in which case we will simplify the notation to $\pi(\psi\to\psi)$ and $\tau(\psi\to\psi)$.

{\color{black} When dealing with the above notions, it is often advantageous to make use of generating functions.} Following \cite{gvfr}, we introduce the function
\beq\label{frfcn}
\mathbb{F}(z) = 
\mathbb{P}\Phi(I-z\mathbb{Q}\Phi)^{-1}\mathbb{P} = 
\sum_{n\geq 1}z^{n-1}\mathbb{A}_n,
\qquad 
\mathbb{A}_n = \mathbb{P}\Phi(\mathbb{Q}\Phi)^{n-1}\mathbb{P},
\eeq
which is analytic on the open unit disk $\mathbb{D}:=\{z\in\mathbb{C}:|z|<1\}$ (see Theorem~\ref{thm:SCHUR}). Then,
\begin{equation} \label{eq:ptf}
\begin{gathered}
\pi_n(\rho\to\mathcal{H}_0) = 
\operatorname{Tr}(\mathbb{A}_n(\rho)),
\qquad\quad 
\pi(\rho\to\mathcal{H}_0) = 
\sum_{n\geq 1} \operatorname{Tr}(\mathbb{A}_n(\rho)) = 
\lim_{x\uparrow 1} \operatorname{Tr}(\mathbb{F}(x)(\rho)),
\\
\tau(\rho\to\mathcal{H}_0) =  
\sum_{n\geq 1} n\operatorname{Tr}(\mathbb{A}_n(\rho)) = 
\lim_{x\uparrow 1}\frac{d}{dx}x\operatorname{Tr}(\mathbb{F}(x)(\rho)) =
1+\lim_{x\uparrow 1}\frac{d}{dx}\operatorname{Tr}(\mathbb{F}(x)(\rho)),
\quad \textrm{ if } \pi(\rho\to\mathcal{H}_0)=1.
\end{gathered}
\end{equation}
Note that in the case $\pi(\rho\to\mathcal{H}_0)<1$ we have by definition $\tau(\rho\to\mathcal{H}_0)=\infty$ but $\sum_{n\geq 1} n\operatorname{Tr}(\mathbb{A}_n(\rho))$ could be finite. Other recurrence notions can be handled with the generating function $\mathbb{F}$. For instance, the probability of landing on a pure state $\psi\in\mathcal{H}_0$ when returning to $\mathcal{H}_0$ starting at the density $\rho\in\mathcal{I}(\mathcal{H}_0)$, reads as
\begin{equation} \label{eq:pf}
\pi(\rho\stackrel{\mathcal{H}_0}{\to}\psi) = 
\sum_{n\geq 1}\pi_n(\rho\stackrel{\mathcal{H}_0}{\to}\psi) = 
\sum_{n\geq 1}
\operatorname{Tr}(|\psi\rangle\langle\psi|\Phi(\mathbb{Q}\Phi)^{n-1}(\rho)) =
\sum_{n\geq 1}\langle\psi|\mathbb{A}_n(\rho)\psi\rangle =
\lim_{x\uparrow 1}\langle\psi|\mathbb{F}(x)(\rho)\psi\rangle.
\end{equation}

The coefficients $\mathbb{A}_n$ are CP maps on $\operatorname{ran}(\mathbb{P})=\mathcal{I}(\mathcal{H}_0)$, thus $\mathbb{F}$ maps the open unit disk into operators on $\mathcal{I}(\mathcal{H}_0)$. Since the projection $\mathbb{P}+\mathbb{Q}$ is not the identity in general, $\mathbb{F}$ is not an FR-function \eqref{eq:FR}. Nevertheless, $\mathbb{F}$ may be seen as a projection $\mathbb{F} = \mathbb{P}f\mathbb{P}$ of the FR-function 
\beq\label{ftrue}
f(z) = (I-\mathbb{Q})\Phi(I-z\mathbb{Q}\Phi)^{-1}(I-\mathbb{Q}) =
\sum_{n\geq 1} a_nz^{n-1},
\qquad
a_n = (I-\mathbb{Q})\Phi(\mathbb{Q}\Phi)^{n-1}(I-\mathbb{Q}),
\eeq
taking values in a larger space, but satisfying the splitting properties typical of such functions \cite[Sect.~6]{gvfr}. The coefficients $a_n$ are now CP maps on $\operatorname{ran}(I-\mathbb{Q})=\mathcal{I}(\mathcal{H}_0)\oplus\mathcal{I}(\mathcal{H}_0,\mathcal{H}_0^\bot)$, hence the values of $f$ are operators on $\mathcal{I}(\mathcal{H}_0)\oplus\mathcal{I}(\mathcal{H}_0,\mathcal{H}_0^\bot)$. Since $I-\mathbb{Q}-\mathbb{P}$ is a projection onto $\mathcal{I}(\mathcal{H}_0,\mathcal{H}_0^\bot)$ and $\langle\psi|\rho\psi\rangle=0$ for all $\rho\in\mathcal{I}(\mathcal{H}_0,\mathcal{H}_0^\bot)$ and $\psi\in\mathcal{H}_0$, the function $\mathbb{F}$ may be substituted by $f$ in the expressions \eqref{eq:ptf} and \eqref{eq:pf}.  

The above discussion may be directly extended to quantum Markov chains, understood as iterated CPTP maps on $\mathcal{I}(\mathcal{H}\otimes\mathcal{S})$ with the form \eqref{eq:TOM2}, where $\mathcal{S}$ stands for the space of sites of a graph. However, if we consider a quantum Markov chain as an iteration of a TOM \eqref{eq:TOM} on $\mathcal{I}_\mathcal{S}(\mathcal{H})$, the monitoring process requires the return subspace $\mathcal{K}\subset\mathcal{H}\otimes\mathcal{S}$ to be an admissible subspace \eqref{eq:adm}, so that the iteration of the TOM is possible after each measurement. If $\mathcal{K}=\oplus_{i\in V}\mathcal{H}_i\otimes|i\rangle$, it makes sense to ask about the return probability $\pi(\rho\to\mathcal{K})$ and the expected return time $\tau(\rho\to\mathcal{K})$ of a state $\rho$ lying on the Banach subspace of $\mathcal{I}_\mathcal{S}(\mathcal{H})$ given by
$$
\mathcal{I}_\mathcal{S}(\{\mathcal{H}_i\}_{i\in V}) = \textstyle
\{\sum_{i\in V}\rho_i\otimes|i\rangle\langle i|\in\mathcal{I}(\mathcal{H}\otimes\mathcal{S}):
\rho_i\in\mathcal{I}(\mathcal{H}_i)\} =
\bigoplus_{i\in V} \mathcal{I}(\mathcal{H}_i)\otimes|i\rangle\langle i|.
$$
A particularly simple situation arises when considering the return properties of a sum of sites $\oplus_k \mathcal{H}\otimes|k\rangle$ since in this case $\mathbb{P}+\mathbb{Q}=I$ because no element of $\mathcal{I}_\mathcal{S}(\mathcal{H})$ exchanges different sites, hence $f=\mathbb{F}$. 

In what follows we will prove different results for quantum channels, with the understanding that they hold also for TOMs since they can be viewed as CPTP maps.  


The work \cite{gvfr} makes it clear that FR-functions constitute natural generalizations of the so-called {\bf Schur functions}, which are contractive operator valued analytic maps on the open unit disk. Actually, we are going to see that \eqref{ftrue} is a true Schur function taking values in operators on the Banach subspace $\operatorname{ran}(I-\mathbb{Q})\subset\mathcal{I}(\mathcal{H})$. 

\begin{theorem} \label{thm:SCHUR}
Let $\Phi$ be a quantum channel on a Hilbert space $\mathcal{H}$, and $\mathbb{Q}=Q\cdot Q$ the projection \eqref{eq:PQ} on $\mathcal{I}(\mathcal{H})$ corresponding to an orthogonal projection $Q$ on $\mathcal{H}$. Then, the function $f$ taking values in operators on $\operatorname{ran}(I-\mathbb{Q})$ defined by \eqref{ftrue} is a Schur function, i.e. it is analytic on the open unit disk, where it is contractive with respect to the trace norm $\|\cdot\|_1$,
\beq\label{schurineq}
\|f(z)(\rho)\|_1 \leq \|\rho\|_1,
\qquad
\forall\rho\in\operatorname{ran}(I-\mathbb{Q}),
\qquad 
\forall z\in\mathbb{D}.
\eeq
\end{theorem}

{\bf Proof.}
From \eqref{eq:CPnorm}, we know that $\|\Phi\|=\|\mathbb{Q}\|=1$, which implies the analyticity of $f$ given by \eqref{ftrue}. The proof of the contractivity involves an argument similar to that of \cite[Prop. 10.1]{gvfr}. Given an arbitrary $\rho\in\operatorname{ran}(I-\mathbb{Q})$, we have that
\beq\label{eq:normfa}
\|f(z)(\rho)\|_1 \leq
\sum_{n\geq 1}|z|^{n-1} \|a_n(\rho)\|_1 \leq
\sum_{n\geq 1}\|a_n(\rho)\|_1,
\qquad
z\in\mathbb{D}.
\eeq
Rewriting 
\beq\label{eq:ar}
a_n(\rho) = (I-\mathbb{Q})\Phi(\mathbb{Q}\Phi)^{n-1}(\rho) =
\Phi(\mathbb{Q}\Phi)^{n-1}(\rho)-(\mathbb{Q}\Phi)^{n}(\rho)
\eeq
and using that $\mathbb{Q}$, $I-\mathbb{Q}$ and $\Phi$ are positive, the latter being trace preserving, yields 
\beq\label{eq:sumtra}
\sum_{n=1}^N \|a_n(\rho)\|_1 =
\sum_{n=1}^N \operatorname{Tr}(a_n(\rho)) = 
\operatorname{Tr}(\rho)-\operatorname{Tr}((\mathbb{Q}\Phi)^{N}(\rho)) \leq 
\operatorname{Tr}(\rho),
\qquad \rho\ge0.
\eeq
From \eqref{eq:normfa} and \eqref{eq:sumtra} we conclude that
$$
\|f(z)(\rho)\|_1 \leq \operatorname{Tr}(\rho) = \|\rho\|_1,
\qquad z\in\mathbb{D},
\qquad \rho\ge0.
$$

The inequality for an arbitrary $\rho\in\mathcal{I}(\mathcal{H})$ follows from the observation that, as a composition of CP maps, $a_n$ is also a CP map. Then, \cite[Lemma 4.2]{schrader} implies that 
\beq \label{eq:2pos}
\|a_n(\rho)\|_1 \le \|a_n(|\rho|)\|_1^{1/2} \|a_n(U|\rho|U^*)\|_1^{1/2},
\eeq
where $U$ is the partial isometry such that $\ker(U)=\ker(\rho)$ appearing in the polar decomposition $\rho=U|\rho|$. Combined with \eqref{eq:normfa} and \eqref{eq:sumtra}, the above inequality gives 
$$
\begin{aligned}
\|f(z)(\rho)\|_1 
& \leq \sum_{n\ge1} \|a_n(|\rho|)\|_1^{1/2} \|a_n(U|\rho|U^*)\|_1^{1/2} 
\le \bigg(\sum_{n\ge1}\|a_n(|\rho|)\|_1\bigg)^{1/2} 
\bigg(\sum_{n\ge1}\|a_n(U|\rho|U^*)\|_1\bigg)^{1/2} 
\\
& \le \operatorname{Tr}(|\rho|)^{1/2} \, \operatorname{Tr}(U|\rho|U^*)^{1/2} 
= \operatorname{Tr}(|\rho|) = \|\rho\|_1, 
\qquad z\in\mathbb{D},
\end{aligned}
$$
where we have used that $U^*\rho=|\rho|$.
\qed

\begin{remark} \label{SCHUR-GEN}
The above proposition actually holds for any 2-positive trace preserving map $\Phi$ because then $a_n$ are also 2-positive and the inequality \eqref{eq:2pos} is valid for any such a map \cite[Lemma 4.2]{schrader}. Besides, $\Phi$ does not need to be trace preserving, but it suffices to be trace non-increasing, i.e. $\operatorname{Tr}(\Phi(\rho))\le\operatorname{Tr}(\rho)$ for $\rho\ge0$.  
\end{remark}

Note that $\mathbb{F}=\mathbb{P}f\mathbb{P}$ is also a Schur function, but with values in the reduced Banach space $\operatorname{ran}(\mathbb{P})$. We will refer to \eqref{ftrue} and \eqref{frfcn} as the {\bf Schur function} $f$ and the {\bf reduced Schur function} $\mathbb{F}$ of the subspace $\mathcal{H}_0=\operatorname{ran}(Q)^\bot$ with respect to the CP map $\Phi$.  

The previous result also applies to TOMs with a space $\mathcal{S}$ of sites, when viewed as CPTP maps on $\mathcal{I}(\mathcal{H}\otimes\mathcal{S})$. If we consider a TOM as a restriction of such a CPTP map to $\mathcal{I}_\mathcal{S}(\mathcal{H})$, the orthogonal projection $Q$ must project onto an admissible subspace, i.e. $Q=\sum_{i\in V}Q_i\otimes|i\rangle\langle i|$ for some orthogonal projections $Q_i$ on $\mathcal{H}$. Then, the Schur function $f$ takes values in the subspace of $\mathcal{I}_\mathcal{S}(\mathcal{H})$ given by $\operatorname{ran}(I-\mathbb{Q})=\sum_{i\in V}\operatorname{ran}(I-\mathbb{Q}_i)\otimes|i\rangle\langle i|$, while the reduced one $\mathbb{F}$ takes values in $\operatorname{ran}(\mathbb{P})=\sum_{i\in V}\operatorname{ran}(\mathbb{P}_i)\otimes|i\rangle\langle i|$. When $Q$ projects onto a sum of sites, there is a single Schur function, $f=\mathbb{F}$, whose values are maps on the space of trace-class operators on the complementary sum of sites.

We are interested in finding simple rules to decide {\color{black} when a return probability equals 1} and, in the affirmative case, whether the expected return time is finite or not. For this purpose, we introduce some terminology which extends to the quantum setting usual notions for classical Markov chains. Given a quantum channel $\Phi$ on $\mathcal{H}$, a subspace $\mathcal{H}_0\subset\mathcal{H}$ is called:
\begin{itemize}
\item {\bf Recurrent} if $\pi(\rho\to\mathcal{H}_0)=1$ for every state $\rho\in\mathcal{D}(\mathcal{H}_0)$. 
\item {\bf Positive recurrent} if $\tau(\rho\to\mathcal{H}_0)<\infty$ for every state $\rho\in\mathcal{D}(\mathcal{H}_0)$.
\end{itemize} 
We refer to a (positive) recurrent pure state $\psi\in\mathcal{H}$ when the subspace spanned by $\psi$ is (positive) recurrent. These notions have an obvious meaning in the case of quantum Markov chains {\color{black} and we note that such quantities can be described in terms of the linear functionals} $s_n$ on $\mathcal{I}(\mathcal{H})$ given by
$$
s_n(\rho)=\operatorname{Tr}((\mathbb{Q}\Phi)^n(\rho)).
$$
If $\rho\in\mathcal{D}(\mathcal{H}_0)$, we refer to $s_n(\rho)$ as the {\bf survival probabilities} for the state $\rho$, which constitute a non-increasing sequence giving the probability of not returning to $\mathcal{H}_0$ during the first $n$ steps when starting at $\rho$. We have that 
\beq \label{eq:tq}
\pi(\rho\to\mathcal{H}_0) = 1 - \lim_{n\to\infty} s_n(\rho),
\qquad\quad
\tau(\rho\to\mathcal{H}_0) = \sum_{n\ge0} s_n(\rho),
\qquad\quad
\forall\rho\in\mathcal{D}(\mathcal{H}_0),
\eeq
so that 
$$
\pi(\rho\to\mathcal{H}_0)=1 \; \Leftrightarrow \; \lim_{n\to\infty}s_n(\rho)=0,
\qquad\quad
\tau(\rho\to\mathcal{H}_0)<\infty \; \Leftrightarrow \; \sum_{n\ge0}s_n(\rho)<\infty.
$$
The first identity of \eqref{eq:tq} is a direct consequence of \eqref{eq:ar}, while the second one follows from the observation that \eqref{eq:ar} leads to
$$
\tau(\rho\to\mathcal{H}_0) = 
\begin{cases}
\lim_{n\to\infty} \left(\sum_{k=0}^{n-1}s_k(\rho)-ns_n(\rho) \right), 
& \text{if} \; \pi(\rho\to\mathcal{H}_0) = 1,
\\
\infty, 
& \text{if} \; \pi(\rho\to\mathcal{H}_0) < 1.
\end{cases}
$$ 
Since $\pi(\rho\to\mathcal{H}_0)<1$ is equivalent to $\lim_{n\to\infty}s_n(\rho)>0$, in this case \eqref{eq:tq} obviously holds. On the other hand, when $\pi(\rho\to\mathcal{H}_0)=1$, \eqref{eq:tq} is a consequence of the fact that $\lim_{n\to\infty}ns_n(\rho)=0$ whenever $\sum_{k=0}^{n-1}s_k(\rho)-ns_n(\rho)$ converges (see \cite[Lemma A.3]{bourg} or \cite[Lemma A]{sinkovicz}).

The advantage of \eqref{eq:tq} over \eqref{eq:ptf} is that the expression of the expected return time is always valid --even when the return {\color{black}probability is not equal to 1}--, which makes \eqref{eq:tq} useful to prove general properties of the expected return time, as the following remark shows.  

\begin{remark} \label{rem:rec-ps}
From \eqref{eq:tq} and the linearity of $s_n$ it is clear that the condition $\pi(\rho\to\mathcal{H}_0)=1$ is closed under convex combinations of countably many states, while $\tau(\rho\to\mathcal{H}_0)<\infty$ is closed under convex combinations of finitely many states. Therefore, the recurrence of $\mathcal{H}_0$ is characterized by the recurrence of each of its pure states, while a similar characterization holds for positive recurrence if $\dim\mathcal{H}_0<\infty$. 

On the other hand, suppose that a linear functional $\omega\colon\mathcal{I}(\mathcal{H})\to\mathbb{C}$ is positive, i.e. $\omega(\rho)\ge0$ for $\rho\ge0$. Then, given two pure states $\rho_\phi=|\phi\rangle\langle\phi|$, $\rho_\psi=|\psi\rangle\langle\psi|$, for any linear combination $\alpha\phi+\beta\psi$,
$$
\omega(\rho_{\alpha\phi+\beta\psi}) = 
|\alpha|^2 \omega(\rho_\phi) + |\beta|^2 \omega(\rho_\psi) + 
2\operatorname{Re}(\overline\alpha\beta\,\omega(|\psi\rangle\langle\phi|)) 
\ge 0,
$$
because $\omega(\rho^*)=\overline{\omega(\rho)}$ as a consequence of the positive decomposition of trace class operators. This inequality implies that 
$|\omega(|\psi\rangle\langle\phi|)|^2 \le \omega(\rho_\phi)\,\omega(\rho_\psi)$, which combined with the expression for $\omega(\rho_{\alpha\phi+\beta\psi})$ gives
$$
\omega(\rho_{\alpha\phi+\beta\psi})^{1/2} \le 
|\alpha|\,\omega(\rho_\phi)^{1/2} + |\beta|\,\omega(\rho_\psi)^{1/2}. 
$$
Applying this to $\omega=s_n$ shows that $\lim_{n\to\infty} s_n(\rho_{\alpha\phi+\beta\psi}) = 0$ when $\lim_{n\to\infty} s_n(\rho_\phi) = \lim_{n\to\infty} s_n(\rho_\psi) = 0$. This means that the condition $\pi(\psi\to\mathcal{H}_0)=1$ is closed under finite linear combinations of vectors. A similar argument taking $\omega=\sum_{k=0}^ns_k$ proves that the condition $\tau(\psi\to\mathcal{H}_0)<\infty$ is also closed under finite linear combinations of vectors. In consequence, the (positive) recurrence of a finite-dimensional subspace $\mathcal{H}_0$ is characterized by the (positive) recurrence of the pure states given by a basis of $\mathcal{H}_0$.
\end{remark}

{
\begin{remark}\label{rem21r1}(Notions of recurrence for OQWs). Regarding other notions of quantum recurrence in the literature, it is worth noting that in \cite{fagnola} the authors defined  a notion of recurrence for quantum dynamical semigroups. For discrete-time semigroups $(\mathcal{M}^n)_n$, such definition is the following: the semigroup will be recurrent if for any $A\in \mathcal{B}(\mathcal{H})$ satisfying $\langle\varphi|A\varphi\rangle>0$, for any $\varphi\in\mathcal{H}\setminus\{0\}$, the set
$$D(\mathcal{U}(A))=\{\varphi=\sum_{i\in V} \varphi_i\otimes |i\rangle \text{ s.t. } \sum_{k\geq 0}\langle\varphi,(\mathcal{M}^*)^k(A)\varphi\rangle<\infty\}$$
equals $\{0\}$. In [\cite{bardet}, Sect. 3.3], the authors explain that in the case of irreducible OQWs for which the internal degree of freedom of the vertices is finite, i.e., $\dim\mathcal{H}<\infty$, it holds that recurrence in the sense of \cite{fagnola} is equivalent to monitored recurrence. We also recall that the same notion of monitoring presented in this work is discussed in \cite{ls2015} and in \cite{cgl} the authors show that, in the case of OQWs, site recurrence in the sense of P\'olya,
$$\sum_{n\geq 1} \operatorname{Tr}(\mathbb{P}\Phi^n\mathbb{P}\rho)=\infty,\;\;\;\rho\in\mathcal{D}(\mathcal{H}_0),$$ 
which corresponds to recurrence without monitoring, is implied by the assumption of monitored recurrence of a site with respect to any initial density matrix. This notion is referred to as  SJK-recurrence \cite{stefanak} in the literature, see \cite{cgl,werner}. A converse also holds: for every irreducible OQW, if a site is SJK-recurrent with respect to some density then the site is monitored recurrent (with respect to every density located on such site). With respect to this last property, valid for dissipative walks, we remark that in the case of UQWs the authors in \cite{werner} describe an example showing that monitored recurrence and SJK-recurrence are not equivalent notions. At this point we are able to appreciate the value of  studying recurrence in terms of Schur functions, as this provides a unified description of monitored recurrence associated with classical and (open and closed) quantum settings \cite{gvfr}. Besides, new results for monitored recurrence come from the appropriate use of classical results for Schur functions \cite{bourg,gvfr,werner}. Examples are the splitting rules discussed in \cite{gvfr} for classical Markov chains and UQWs, which we extend in this paper to quantum Markov chains.
\end{remark}
}

The FR-function machinery previously introduced yields a characterization of recurrence and positive recurrence in terms of Schur functions. To prove it we will use the following result.

\begin{lemma} \label{lem:monconv}
Given a Hilbert space $\mathcal{H}$, any non-decreasing uniformly bounded sequence of positive operators in $\mathcal{I}(\mathcal{H})$ is norm convergent. 
\end{lemma}

{\bf Proof.}
Let $\rho_n\in\mathcal{I}(\mathcal{H})$ such that $0\le\rho_m\le\rho_n$ for $m<n$ and $\|\rho_n\|_1=\operatorname{Tr}(\rho_n)$ is bounded. Then, if $m<n$, we have that $\rho_n-\rho_m\ge0$, so that $\|\rho_n-\rho_m\|_1=\operatorname{Tr}(\rho_n-\rho_m)$. Hence, $\operatorname{Tr}(\rho_n)$ is a non-decreasing bounded sequence of real numbers, thus it is convergent. Therefore, $\operatorname{Tr}(\rho_n)$ is a Cauchy sequence, but then the equality $\|\rho_n-\rho_m\|_1=\operatorname{Tr}(\rho_n)-\operatorname{Tr}(\rho_m)$ shows that $\rho_n$ is also a Cauchy sequence in the Banach space $\mathcal{I}(\mathcal{H})$, thus it is convergent with respect to the trace norm.  
\qed 

Now we can deal with the Schur characterization of recurrence and positive recurrence. 

\begin{pro} \label{pro:rec-F}
Let $\Phi$ be a quantum channel on a Hilbert space $\mathcal{H}$ and $\mathcal{H}_0\subset\mathcal{H}$ a subspace. If $\mathbb{F}$ is the reduced Schur function of $\mathcal{H}_0$ with respect to $\Phi$, then: 
\begin{itemize}
\item[(i)] $\mathbb{F}(1):=\lim_{x\uparrow1}\mathbb{F}(x)$ exists as an operator on $\mathcal{I}(\mathcal{H}_0)$ in the strong sense and it is a CP map which is trace non-increasing, i.e. $\operatorname{Tr}(\mathbb{F}(1)\rho)\le\operatorname{Tr}(\rho)$ for every positive $\rho\in\mathcal{I}(\mathcal{H}_0)$. Besides,
$$
\pi(\rho\to\mathcal{H}_0)=\operatorname{Tr}(\mathbb{F}(1)(\rho)),
\qquad
\rho\in\mathcal{D}(\mathcal{H}_0).
$$
\item[(ii)] $\mathcal{H}_0$ is recurrent iff $\mathbb{F}(1)$ is trace preserving.
\item[(iii)] $\mathcal{H}_0$ is positive recurrent iff $\mathbb{F}(1)$ is trace preserving and $\mathbb{F}'(1):=\lim_{x\uparrow1}\mathbb{F}'(x)$ exists as an operator on $\mathcal{I}(\mathcal{H}_0)$ in the strong sense. Then, $\mathbb{F}'(1)$ is a CP map on $\mathcal{I}(\mathcal{H}_0)$ such that
$$
\tau(\rho\to\mathcal{H}_0) = 1 + \operatorname{Tr}(\mathbb{F}'(1)(\rho)),
\qquad
\rho\in\mathcal{D}(\mathcal{H}_0).
$$  
\end{itemize}
Similar results hold for reduced Schur functions with respect to TOMs.   
\end{pro}

{\bf Proof.}
The operators $\mathbb{A}_n=\mathbb{P}\Phi(\mathbb{Q}\Phi)^{n-1}\mathbb{P}$ are CP maps on $\mathcal{I}(\mathcal{H}_0)$, thus this is also true for
$\mathbb{F}(x)=\sum_{n\ge1}x^{n-1}\mathbb{A}_n$ when $x\in(0,1)$. Therefore, $\mathbb{F}(x)(\rho)\in\mathcal{I}(\mathcal{H}_0)$ is positive and non-decreasing in $x$ for $x\in(0,1)$ and any positive $\rho\in\mathcal{I}(\mathcal{H}_0)$. It is also uniformly bounded with respect to $x$ because 
\beq\label{eq:trFr}
\|\mathbb{F}(x)(\rho)\|_1 = \operatorname{Tr}(\mathbb{F}(x)(\rho)) =  
\operatorname{Tr}(f(x)(\rho)) = \sum_{n\ge1}x^{n-1}\operatorname{Tr}(a_n(\rho)) \le
\sum_{n\ge1}\operatorname{Tr}(a_n(\rho)) \le \operatorname{Tr}(\rho),
\qquad
x\in(0,1),
\eeq
the last inequality due to \eqref{eq:sumtra}. As a consequence of Lemma~\ref{lem:monconv}, if $\rho\in\mathcal{I}(\mathcal{H})$ is positive, $\mathbb{F}(x)(\rho)$ converges in $\mathcal{I}(\mathcal{H})$ with respect to the trace norm for $x\uparrow1$. The convergence for an arbitrary $\rho\in\mathcal{I}(\mathcal{H})$ follows from the positive decomposition of a trace-class operator. This proves that $\mathbb{F}(1)=\lim_{x\uparrow1}\mathbb{F}(x)$ exists in the strong sense as an operator on $\mathcal{I}(\mathcal{H}_0)$. As a strong limit of CP maps, $\mathbb{F}(1)$ is also a CP map. Besides, $\operatorname{Tr}(\mathbb{F}(1)(\rho))\le\operatorname{Tr}(\rho)$ for any positive $\rho\in\mathcal{I}(\mathcal{H}_0)$ due to \eqref{eq:trFr}. Bearing in mind \eqref{eq:ptf}, this proves (i).

If $\mathbb{F}(1)$ is trace preserving, then (i) implies that $\mathcal{H}_0$ is recurrent. To see the converse, assume that $\mathcal{H}_0$ is recurrent. Then $\operatorname{Tr}(\mathbb{F}(1)(\rho)) = \pi(\rho\to\mathcal{H}_0) = 1$ for every state $\rho\in\mathcal{D}(\mathcal{H}_0)$. This obviously extends as $\operatorname{Tr}(\mathbb{F}(1)(\rho)) = \operatorname{Tr}(\rho)$ for any positive $\rho\in\mathcal{I}(\mathcal{H}_0)$, an identity which finally generalizes to every $\rho\in\mathcal{I}(\mathcal{H}_0)$ by using the positive decomposition of trace-class operators. This yields (ii).

Let us prove now (iii). If $\mathbb{F}(1)$ is trace preserving and $\mathbb{F}'(1)$ exists in the strong sense, then (i) and \eqref{eq:ptf} imply that $\tau(\rho\to\mathcal{H}_0) = 1 + \operatorname{Tr}(\mathbb{F}'(1)(\rho))<\infty$ for all $\rho\in\mathcal{D}(\mathcal{H}_0)$. Conversely, suppose that $\mathcal{H}_0$ is positive recurrent. If $x\in(0,1)$, the same arguments given for $\mathbb{F}(x)$ prove that $\mathbb{F}'(x)=\sum_{n\ge2}nx^{n-2}\mathbb{A}_n$ is a CP map on $\mathcal{I}(\mathcal{H}_0)$ and $\mathbb{F}'(x)(\rho)\in\mathcal{I}(\mathcal{H}_0)$ is positive and non-decreasing in $x$ for every positive $\rho\in\mathcal{I}(\mathcal{H}_0)$. Also, for any state $\rho\in\mathcal{D}(\mathcal{H}_0)$,
$$
\|\mathbb{F}'(x)(\rho)\|_1 = \operatorname{Tr}(\mathbb{F}'(x)(\rho)) =  
\sum_{n\ge1}(n-1)\,x^{n-2}\,\operatorname{Tr}(\mathbb{A}_n(\rho)) \le
\sum_{n\ge1}(n-1)\,\operatorname{Tr}(\mathbb{A}_n(\rho)) = 
\tau(\rho\to\mathcal{H}_0)-1,
$$
the last equality due to \eqref{eq:ptf} and the recurrence of $\mathcal{H}_0$. Therefore, $\mathbb{F}'(x)(\rho)$ is uniformly bounded with respect to $x$ for any positive $\rho\in\mathcal{I}(\mathcal{H}_0)$ due to the positive recurrence of $\mathcal{H}_0$. In consequence, Lemma~\ref{lem:monconv} combined with the positive decomposition of trace-class operators prove that $\mathbb{F}'(x)(\rho)$ converges in norm in $\mathcal{I}(\mathcal{H}_0)$ for $x\uparrow1$ and any $\rho\in\mathcal{I}(\mathcal{H}_0)$. We conclude that $\mathbb{F}'(1)=\lim_{x\uparrow1}\mathbb{F}'(x)$ exists in the strong sense as a CP map on $\mathcal{I}(\mathcal{H}_0)$, whose relation with the expected return time follows from \eqref{eq:ptf}.
\qed

The previous Schur characterization of recurrence also leads to a spectral sufficient condition for positive recurrence.

\begin{pro} \label{pro:rec-sp}
Let $\Phi$ be a quantum channel on a Hilbert space $\mathcal{H}$ and $\mathcal{H}_0\subset\mathcal{H}$ a subspace. If $\mathbb{Q}=Q\cdot Q$, with $Q$ the orthogonal projection of $\mathcal{H}$ onto $\mathcal{H}_0^\bot$, the subspace $\mathcal{H}_0$ is positive recurrent whenever 1 does not lie in the spectrum of $\mathbb{Q}\Phi$. In this case, 
\beq\label{eq:tinv}
\tau(\rho\to\mathcal{H}_0) = 
\operatorname{Tr}\left((I-\mathbb{Q}\Phi)^{-1}(\rho)\right),
\qquad
\forall\rho\in\mathcal{D}(\mathcal{H}_0).
\eeq
Similar results hold for TOMs.
\end{pro}

{\bf Proof.}
Since $(1-z\mathbb{Q}\Phi)^{-1}$ is an analytic function of $z$ when $z^{-1}$ lies in the resolvent set of $\mathbb{Q}\Phi$, the Schur function $f$ given by \eqref{ftrue} is analytic around $z=1$ if the spectrum of $\mathbb{Q}\Phi$ excludes the point 1. Then, both $f$ and $\mathbb{F}=\mathbb{P}f\mathbb{P}$, as well as their derivatives, are continuous at $z=1$, which implies that $\mathbb{F}(1)$ and $\mathbb{F}'(1)$ exist as operators on $\mathcal{I}(\mathcal{H}_0)$. In this case, for the same reason previously pointed out, the reduced Schur function $\mathbb{F}$ may be substituted by $f$ in the expressions for the return probability and for the expected return time given in Proposition~\ref{pro:rec-F}. Using that $\Phi$ is trace preserving, we find for every state $\rho\in\mathcal{D}(\mathcal{H}_0)$ that
$$
\begin{aligned}
\pi(\rho\to\mathcal{H}_0) 
& = \operatorname{Tr}(f(1)(\rho)) 
= \operatorname{Tr}\left((I-\mathbb{Q})\Phi(I-\mathbb{Q}\Phi)^{-1}(\rho)\right) 
= \operatorname{Tr}\left((I-\mathbb{Q}\Phi)^{-1}(\rho)\right) -
\operatorname{Tr}\left(\mathbb{Q}\Phi(I-\mathbb{Q}\Phi)^{-1}(\rho)\right)
\\ 
& = \operatorname{Tr}(\rho) = 1,
\end{aligned} 
$$ 
so that the subspace $\mathcal{H}_0$ is recurrent. Besides, 
$$
\begin{aligned}
\tau(\rho\to\mathcal{H}_0) 
& = 1 + \operatorname{Tr}(f'(1)(\rho)) 
= \operatorname{Tr}(f(1)(\rho)) + \operatorname{Tr}(f'(1)(\rho)) 
\\ 
& = \operatorname{Tr}\left((I-\mathbb{Q})\Phi(I-\mathbb{Q}\Phi)^{-1}(\rho)\right) + 
\operatorname{Tr}
\left((I-\mathbb{Q})\Phi\mathbb{Q}\Phi(I-\mathbb{Q}\Phi)^{-2}(\rho)\right) 
= \operatorname{Tr}\left((I-\mathbb{Q})\Phi(I-\mathbb{Q}\Phi)^{-2}(\rho)\right)
\\
& = \operatorname{Tr}\left((I-\mathbb{Q}\Phi)^{-2}(\rho)\right) -
\operatorname{Tr}\left(\mathbb{Q}\Phi(I-\mathbb{Q}\Phi)^{-2}(\rho)\right) 
= \operatorname{Tr}\left((I-\mathbb{Q}\Phi)^{-1}(\rho)\right) < \infty,
\end{aligned}
$$
hence $\mathcal{H}_0$ is positive recurrent.
\qed 

The previous spectral characterization of positive recurrence yields a recurrence property of irreducible quantum channels and quantum Markov chains which extends a known one for classical Markov chains.

\begin{theorem} \label{thm:rec-fi} 
With respect to an irreducible quantum channel on a finite-dimensional Hilbert space, every non-trivial subspace is positive recurrent and the expected return time is given by \eqref{eq:tinv}. A similar result holds for an irreducible quantum Markov chain on a finite graph with a finite-dimensional space of internal degrees of freedom.
\end{theorem}

{\bf Proof of Theorem~\ref{thm:rec-fi}.}
According to Proposition~\ref{pro:rec-sp}, and bearing in mind that $\dim\mathcal{H}<\infty$, it is enough to show that, for every irreducible quantum channel $\Phi$ and every orthogonal projection $Q\ne I$ on $\mathcal{H}$, 1 is not an eigenvalue of $\mathbb{Q}\Phi$, i.e. $\mathbb{Q}\Phi$ has no fixed point in $\mathcal{I}(\mathcal{H})$, where $\mathbb{Q} = Q \cdot Q$.

We know that $\|\mathbb{Q}\Phi\|\le1$. If $\|\mathbb{Q}\Phi\|<1$, the operator $I-\mathbb{Q}\Phi$ is invertible, hence 1 is not an eigenvalue of $\mathbb{Q}\Phi$. 
Suppose that $\|\mathbb{Q}\Phi\|=1$ and $\mathbb{Q}\Phi(\rho)=\rho$ for some $\rho\in\mathcal{I}(\mathcal{H})$. Then, \cite[Thm. 4.1]{schrader} guarantees the existence of a positive fixed point for $\mathbb{Q}\Phi$, so we can assume without loss that $\rho\in\mathcal{D}(\mathcal{H})$. Since $\Phi$ is trace preserving, the identity $\mathbb{Q}\Phi(\rho)=\rho$ leads to 
$$
\operatorname{Tr}(\mathbb{Q}\Phi(\rho)) = \operatorname{Tr}(\rho) = \operatorname{Tr}(\Phi(\rho)) =
\operatorname{Tr}(\mathbb{P}\Phi(\rho)) + 
\operatorname{Tr}(\mathbb{Q}\Phi(\rho)),
$$
which implies that the positive operator $\mathbb{P}\Phi(\rho)$ must vanish. Then, the positivity of $(\mathbb{P}+\lambda\mathbb{Q})\Phi(\rho)$ for every $\lambda>0$ yields
$$
P\Phi(\rho)Q + Q\Phi(\rho)P + \lambda Q\Phi(\rho)Q \ge0, 
\qquad
\lambda>0.
$$
Taking $\lambda\downarrow0$, it follows that the traceless operator $P\Phi(\rho)Q + Q\Phi(\rho)P$ is positive, hence $P\Phi(\rho)Q=Q\Phi(\rho)P=0$. We conclude that $\Phi(\rho)=\mathbb{Q}\Phi(\rho)=\rho$, i.e. $\rho$ is an invariant state of $\Phi$ which is not faithful because $\rho=Q\rho Q$ for some orthogonal projection $Q\ne I$. Therefore, {\color{black}by Theorem \ref{cp_irr},} $\Phi$ cannot be irreducible.
\qed

\begin{remark}
Some comments on the above result are pertinent:
\begin{enumerate}
\item In general, the hypothesis of finiteness of the dimension of the Hilbert space cannot be dropped: consider the biased nearest neighbour random walk on the integers, for which no site is recurrent.
\item In general, the irreducibility assumption cannot be omitted: consider the 1-qubit amplitude damping channel \cite{nielsen}, for which the pure state $\psi=[0\;1]^T$ is not recurrent.
\end{enumerate} 
\end{remark}

Despite the above remarks, it is possible to weaken the assumptions of Theorem~\ref{thm:rec-fi} by referring to the relevant Hilbert space (also known as minimal enclosure) of a finite-dimensional subspace $\mathcal{H}_0\subset\mathcal{H}$, defined as 
$$
\widetilde{\mathcal{H}}_0 = 
\underset{n\ge0}{\overline{\operatorname{span}}} 
\{\operatorname{supp}\Phi^n(P)\},
$$
where $P$ is the orthogonal projection of $\mathcal{H}$ onto $\mathcal{H}_0$. Given an unravelling $\{B_i\}$ of $\Phi$, then
$$
\operatorname{supp}\Phi^n(P) = 
\bigcap_{i_1,\dots,i_n} (\ker(PB_{i_1}^*\cdots B_{i_n}^*))^\bot =
\underset{i_1,\dots,i_n}{\overline{\operatorname{span}}}
\{B_{i_n}\cdots B_{i_1}\mathcal{H}_0\}. 
$$
This shows that $B_i\widetilde{\mathcal{H}}_0\subset\widetilde{\mathcal{H}}_0$, which means that for any state $\rho\in\mathcal{D}(\mathcal{H})$ with $\operatorname{supp}\rho\subset\widetilde{\mathcal{H}}_0$, we have that $\operatorname{supp}\Phi(\rho)\subset\widetilde{\mathcal{H}}_0$, i.e. the set of states $\mathcal{D}(\widetilde{\mathcal{H}}_0)$ is invariant for $\Phi$, a property that extends to $\mathcal{I}(\widetilde{\mathcal{H}}_0)$ via positive decomposition. In other words, the orthogonal projection of $\mathcal{H}$ onto $\widetilde{\mathcal{H}}_0$ reduces $\Phi$. Therefore, concerning the recurrence properties of $\mathcal{H}_0$, the quantum channel $\Phi$ may be restricted without loss to $\mathcal{I}(\widetilde{\mathcal{H}}_0)$. Thus, even if $\Phi$ is not irreducible or $\dim\mathcal{H}=\infty$, a non-trivial subspace $\mathcal{H}_0$ remains positive recurrent as long as $\dim\widetilde{\mathcal{H}}_0<\infty$ and the restriction of $\Phi$ to $\mathcal{I}(\widetilde{\mathcal{H}}_0)$ is irreducible. 

As it was proved in \cite{sinkovicz}, the above result also holds when assuming unitality instead of irreducibility, at least for $\dim\mathcal{H}_0=1$: any pure state with a finite-dimensional relevant Hilbert space $\widetilde{\mathcal{H}}_0$ such that the corresponding restriction of $\Phi$ is unital, must be recurrent. Indeed, \cite{sinkovicz} proves that the corresponding expected return time must be an integer, namely $\dim\widetilde{\mathcal{H}}_0$. This extends to unital quantum channels a result uncovered for the first time in \cite{werner} in the context of unitary dynamics. The next theorem, whose proof is just an adaptation of arguments in \cite{sinkovicz} combined with ideas from \cite{bourg}, generalizes the result obtained in \cite{sinkovicz} to subspaces. 

The following theorem refers to the {\bf averaged expected return time} to a finite-dimensional subspace $\mathcal{H}_0 \subset \mathcal{H}$, defined as  
$$
\tau(\mathcal{H}_0\to\mathcal{H}_0) := 
\int_{\mathcal{D}_p(\mathcal{H}_0)} d\psi \; \tau(\psi\to\mathcal{H}_0), 
$$
where $d\psi$ stands for the uniform probability measure on the subset $\mathcal{D}_p(\mathcal{H}_0)$ of pure states of $\mathcal{D}(\mathcal{H}_0)$. The main ingredient of the measure $d\psi$ that we will use is the identity
\beq \label{eq:dpsi}
\int_{\mathcal{D}_p(\mathcal{H}_0)} d\psi \; \rho_\psi = 
\frac{P}{\dim\mathcal{H}_0},
\qquad
P \; \text{orthogonal projection of} \; \mathcal{H} \; \text{onto} \; \mathcal{H}_0,
\eeq
which follows from the representation $P = \sum_k \rho_{\psi_k}$ in terms of an orthonormal basis $\{\psi_k\}$ of $\mathcal{H}_0$ and the invariance of $d\psi$ under unitary transformations.

\begin{theorem} \label{thm:rec-ui}
Let $\Phi$ be a quantum channel on a Hilbert space $\mathcal{H}$ and $\mathcal{H}_0\subset\mathcal{H}$ a subspace with a finite-dimensional relevant Hilbert space $\widetilde{\mathcal{H}}_0$. If the restriction of $\Phi$ to $\widetilde{\mathcal{H}}_0$ is unital, then $\mathcal{H}_0$ is positive recurrent and the averaged expected return time to $\mathcal{H}_0$ is given by
\beq \label{eq:aet}
\tau(\mathcal{H}_0\to\mathcal{H}_0) = 
\frac{\dim\widetilde{\mathcal{H}}_0}{\dim{\mathcal{H}_0}}.
\eeq
\end{theorem}

\begin{proof}
We can substitute $\Phi$ by its restriction to $\widetilde{\mathcal{H}}_0$, thus in what follows we will suppose without loss that $\mathcal{H}$ is such a finite-dimensional relevant Hilbert space and $\Phi$ is unital on $\mathcal{H}$, i.e. $\Phi^*$ is trace-preserving. Then, denoting by $P$ the orthogonal projection of $\mathcal{H}$ onto $\mathcal{H}_0$, we get for any $\rho\ge0$,
$$
\operatorname{Tr}(\rho(\mathbb{Q}\Phi)^n(P)) =
\operatorname{Tr}(\mathbb{P}(\Phi^*\mathbb{Q})^n(\rho)) =
\operatorname{Tr}((\Phi^*\mathbb{Q})^n(\rho)) - 
\operatorname{Tr}(\mathbb{Q}(\Phi^*\mathbb{Q})^n(\rho)) =
\operatorname{Tr}((\Phi^*\mathbb{Q})^n(\rho)) - 
\operatorname{Tr}((\Phi^*\mathbb{Q})^{n+1}(\rho)),
$$
which implies that 
$\sum_{n\ge0}\operatorname{Tr}(\rho(\mathbb{Q}\Phi)^n(P)) \le \operatorname{Tr}(\rho)$. Using Lemma~\ref{lem:monconv} and \eqref{eq:pos-dual} we find that $\sum_{n\ge0}(\mathbb{Q}\Phi)^n(P)$ converges to a positive operator $T \le I$. Therefore, if $\rho$ is any state supported on $\mathcal{H}_0$, then $\rho\le P$ and $\tau(\rho\to\mathcal{H}_0)=\sum_{n\ge0}\operatorname{Tr}(\mathbb{Q}\Phi)^n(\rho)\le\operatorname{Tr}(T)\le\dim\mathcal{H}<\infty$. In consequence, $\mathcal{H}_0$ is positive recurrent.

Besides, $\mathbb{Q}\Phi(I-T) = Q - (T-P) = I-T$, i.e. $\sigma=I-T$ is invariant for $\mathbb{Q}\Phi$. On the other hand, $\mathbb{Q}\Phi(\sigma)=\sigma$ yields $\operatorname{Tr}(\sigma) = \operatorname{Tr}(\mathbb{Q}\Phi(\sigma)) = \operatorname{Tr}(\sigma) - \operatorname{Tr}(\mathbb{P}\Phi(\sigma))$. Hence, $\mathbb{P}\Phi(\sigma)=0$ and 
$$
\operatorname{Tr}(\sigma\,\Phi^n(P)) = 
\operatorname{Tr}(\mathbb{P}(\Phi^*)^n(\sigma)) = 
\operatorname{Tr}(\mathbb{P}(\sigma)) = 0, 
\qquad n\ge0,
$$
where we have used that the invariant states of $\Phi$ are also invariant for $\Phi^*$ when $\Phi$ is unital \cite{sinkovicz}. Since $\operatorname{supp}\Phi^n(P)$ span the (relevant) Hilbert space $\mathcal{H}$, we conclude that $\sigma=0$, i.e. $T=I$. Hence, according to \eqref{eq:tq}, \eqref{eq:dpsi} and using monotone convergence for $d\psi$, we conclude that
$$
\tau(\mathcal{H}_0\to\mathcal{H}_0) = 
\int_{\mathcal{D}_p(\mathcal{H}_0)} d\psi 
\sum_{n\ge0}\operatorname{Tr}((\mathbb{Q}\Phi)^n(\rho_\psi)) =
\frac{\sum_{n\ge0}\operatorname{Tr}((\mathbb{Q}\Phi)^n(P))}{\dim\mathcal{H}_0}  =
\frac{\operatorname{Tr}(T)}{\dim\mathcal{H}_0} = 
\frac{\dim\mathcal{H}}{\dim\mathcal{H}_0}.
$$
This proves \eqref{eq:aet}. 

\end{proof}
One would naively expect the inequality $\tau(\psi\to\mathcal{H}_0) \le \tau(\psi\to\psi)$ for any $\psi\in\mathcal{H}_0$, so that the positive recurrence of a subspace in the finite-dimensional unital case {\color{black}would} be a straightforward consequence of the positive recurrence of every pure state, already proved in \cite{sinkovicz}. However, the alluded inequality is not true even for unitary evolutions, which are covered by the formalism of CPTP maps. Furthermore, there are examples of unitaries which make recurrent a state whose return probability to some subspace is less than one \cite{bourg}. Therefore, the statement that unital CPTP maps in finite dimension make positive recurrent any subspace does not follow directly from the version of this statement for one-dimensional subspaces. 

Let us illustrate the above techniques and results with a concrete example. In what follows $I_n$ and $0_n$ stand for the $n\times n$ identity and null matrix respectively.
\bex\label{ex261} 
For $p,q \in(0,1)$, consider a TOM on a graph with a couple of two-dimensional sites, acting on $\mathcal{I}_\mathcal{S}(\mathcal{H})$, $\mathcal{S}=\operatorname{span}\{|1\rangle,|2\rangle\}$, $\mathcal{H}=\mathbb{C}^2$, as
$$
\begin{gathered}
\mathcal{E}=\begin{bmatrix} 
\mathcal{E}_1^1 & \mathcal{E}_1^2 
\\[2pt] 
\mathcal{E}_2^1 & \mathcal{E}_2^2
\end{bmatrix},
\qquad
\mathcal{E}_i^j=A_i^j\cdot{A_i^j}^*+B_i^j\cdot{B_i^j}^*,
\\
A_1^1=\frac{1}{2}\sqrt{4-3p}\begin{bmatrix} 1 & 0 \\ 0 & 1\end{bmatrix},
\quad
B_1^1=\frac{\sqrt{p}}{2}\begin{bmatrix} 0 & 1 \\ 1 & 0\end{bmatrix},
\quad
A_2^1=\frac{\sqrt{p}}{2}\begin{bmatrix} 0 & -i \\ i & 0\end{bmatrix},
\quad
B_2^1=\frac{\sqrt{p}}{2}\begin{bmatrix} 1 & 0 \\ 0 & -1\end{bmatrix},
\\
A_1^2=\sqrt{\frac{q}{2}}\begin{bmatrix} 1 & 1 \\ 0 & 0\end{bmatrix},
\quad 
B_1^2=\sqrt{\frac{q}{2}}\begin{bmatrix} 0 & 0 \\ 1 & -1\end{bmatrix},
\quad
A_2^2=\sqrt{\frac{1-q}{3}}\begin{bmatrix} 1 & 1 \\ 0 & 1\end{bmatrix},
\quad 
B_2^2=\sqrt{\frac{1-q}{3}}\begin{bmatrix} 1 & 0 \\ -1 & 1\end{bmatrix}.
\end{gathered}
$$

Note that $\mathcal{E}_1^1+\mathcal{E}_2^1$ equals a kind of depolarizing channel \cite{nielsen} spread on two sites, whereas $\mathcal{E}_1^2+\mathcal{E}_2^2$ is a convex combination of the channel in \cite[Sec. 4]{attal} and the OQW version of the Hadamard channel \cite{cgl}. The block matrix representation of this TOM is 
$$
\widehat{\mathcal{E}} =
\begin{bmatrix} 
\widehat{\mathcal{E}}_1^1 & \widehat{\mathcal{E}}_1^2
\\[2pt]
\widehat{\mathcal{E}}_2^1 & \widehat{\mathcal{E}}_2^2
\end{bmatrix},
\qquad
\widehat{\mathcal{E}}_i^j = 
\lceil A_i^j \rceil + \lceil B_i^j \rceil = 
A_i^j\otimes\overline{A_i^j} + B_i^j\otimes\overline{B_i^j}. 
$$
Explicitly,
\beq \label{eq:ex1-wE}
\widehat{\mathcal{E}} = 
\left[\begin{array}{cccc|cccc}
1-\frac{3p}{4} & 0 & 0 & \frac{p}{4} & \frac{q}{2} & \frac{q}{2} & 
\frac{q}{2} & \frac{q}{2} \\[2pt]
0 & 1-\frac{3p}{4} & \frac{p}{4} & 0 & 0 & 0 & 0 & 0 \\[2pt]
0 & \frac{p}{4} & 1-\frac{3p}{4} & 0 & 0 & 0 & 0 & 0 \\[2pt]
\frac{p}{4} & 0 & 0 & 1-\frac{3p}{4} & \frac{q}{2} & -\frac{q}{2} & 
-\frac{q}{2} & \frac{q}{2} \\[2pt]
\hline
& & & & & & & \\[-11pt]
\frac{p}{4} & 0 & 0 & \frac{p}{4} & \frac{2(1-q)}{3} & \frac{1-q}{3} & 
\frac{1-q}{3} & \frac{1-q}{3} \\[2pt]
0 & -\frac{p}{4} & -\frac{p}{4} & 0 & \frac{q-1}{3} & \frac{2(1-q)}{3} & 0 & 
\frac{1-q}{3} \\[2pt]
0 & -\frac{p}{4} & -\frac{p}{4} & 0 & \frac{q-1}{3} & 0 & \frac{2(1-q)}{3} & 
\frac{1-q}{3} \\[2pt]
\frac{p}{4} & 0 & 0 & \frac{p}{4} & \frac{1-q}{3} & \frac{q-1}{3} & 
\frac{q-1}{3} & \frac{2(1-q)}{3}
\end{array}\right].
\eeq
This TOM has a unique invariant state,
\beq \label{fptomf1}
\rho_*=
\frac{q}{2q+p} \begin{bmatrix} 1 & 0 \\ 0 & 1 \end{bmatrix} 
\otimes|1\rangle\langle 1| +
\frac{p}{2(2q+p)} \begin{bmatrix}1 & 0 \\ 0 & 1 \end{bmatrix}
\otimes|2\rangle\langle 2|,
\eeq
which is faithful, hence $\mathcal{E}$ is irreducible for every choice of $p,q\in (0,1)$, but it is unital iff $p=2q$. Therefore, Theorem~\ref{thm:rec-fi} implies that every subspace is positive recurrent for any $p$, $q$. Also, if $p=2q$, the corresponding averaged expected return time is the dimension of the relevant Hilbert space divided by the dimension of the subspace, due to Theorem~\ref{thm:rec-ui}. 

\smallskip

\noindent(a) {\bf Recurrence of a site.} 
Given the orthogonal projections $P=I_2\otimes|1\rangle\langle1|$ and $Q=I_2\otimes|2\rangle\langle2|$, complementary with respect to $\mathbb{C}^2\otimes\mathcal{S}$, let $\mathbb{P}=P\cdot P$ and $\mathbb{Q}=Q\cdot Q$ be the related projections of $\mathcal{I}_{\mathcal{S}}(\mathbb{C}^2)$ onto $\mathcal{I}(\mathbb{C}^2)\otimes|1\rangle$ and $\mathcal{I}(\mathbb{C}^2)\otimes|2\rangle$ respectively. Although $\mathbb{P}$ and $\mathbb{Q}$ are not complementary with respect to $\mathcal{I}(\mathbb{C}^2\otimes\mathcal{S})$, they are complemetary with respect to the space $\mathcal{I}_{\mathcal{S}}(\mathbb{C}^2)$ where the TOM acts, the corresponding matrix representations being given by $\widehat{\mathbb{P}}=I_4\oplus0_4$ and $\widehat{\mathbb{Q}}=0_4\oplus I_4$. Therefore, the reduced Schur function for site 1, understood as a function on $\mathcal{I}(\mathbb{C}^2)\otimes|1\rangle$, has the matrix representation
$$
\begin{aligned}
\widehat{\mathbb{F}}(z) & = 
\widehat{\mathbb{P}}
\widehat{\mathcal{E}}(I_8-z\widehat{\mathbb{Q}}\widehat{\mathcal{E}})^{-1} \widehat{\mathbb{P}}
\\
& = \begin{bmatrix}
1-\frac{3p}{4}+\frac{pqz}{4(1+z(q-1))} & 
-\frac{pqz(3+z(q-1))}{4((2q^2-4q+2)z^2+(3q-3)z+3)} &
-\frac{pqz(3+z(q-1))}{4((2q^2-4q+2)z^2+(3q-3)z+3)} &
\frac{p}{4}+\frac{pqz}{4(1+z(q-1))} 
\\[3pt]
0 & 1-\frac{3p}{4} & \frac{p}{4} & 0 \\[3pt]
0 & \frac{p}{4} & 1-\frac{3p}{4} & 0 \\[3pt]
\frac{p}{4}+\frac{pqz}{4(1+z(q-1))} & 
\frac{pqz(3+z(q-1))}{4((2q^2-4q+2)z^2+(3q-3)z+3)} & 
\frac{pqz(3+z(q-1))}{4((2q^2-4q+2)z^2+(3q-3)z+3)} & 
1-\frac{3p}{4}+\frac{pqz}{4(1+z(q-1))} 
\end{bmatrix}.
\end{aligned}
$$
For every state $\rho$ supported on site $1$ this gives
$$
\begin{gathered}
\operatorname{Tr}(\mathbb{F}(z)(\rho)) =
(\widehat{\mathbb{F}}(z)\,vec(\rho))_1 + 
(\widehat{\mathbb{F}}(z)\,vec(\rho))_4 =
\frac{(2q+p-2)z+2-p}{2((q-1)z+1)},
\\
\pi(\rho\to|1\rangle) = \operatorname{Tr}(\mathbb{F}(1)(\rho)) = 1,
\qquad
\tau(\rho\to|1\rangle) = 
1 + \operatorname{Tr}(\mathbb{F}'(1)(\rho)) = 1+\frac{p}{2q}.
\end{gathered}
$$
As expected, site 1 is positive recurrent. In the unital case, $p=2q$, its averaged mean return time is $1+p/2q=2$, so that the relevant Hilbert space has dimension $4$, i.e. it is the full Hilbert space $\mathbb{C}^2\otimes\mathcal{S}$. 

Due to the irreducibility, the expected return time may be alternatively obtained from \eqref{eq:tinv}. Using the matrix and $vec$ representations, this means that $\tau(\rho\to|1\rangle) = \operatorname{Tr}(\sigma) = \overrightarrow{\sigma}_1+\overrightarrow{\sigma}_4+\overrightarrow{\sigma}_5+\overrightarrow{\sigma}_8$, where
$$
\overrightarrow{\sigma} = 
(I_8-\widehat{\mathbb{Q}}\widehat{\mathcal{E}})^{-1} \overrightarrow{\rho} = 
\left[\begin{array}{c}
\rho_{11} \\ \rho_{12} \\ \rho_{21} \\ \rho_{22} \\[2pt]
\hline 
\\[-11pt]
\frac{2p-(2\rho_{12}+2\rho_{21}+1)pq+(\rho_{12}+\rho_{21}+1)2pq^2}
{4q(2-q+2q^2)} \\[4pt] 
-\frac{(\rho_{12}+\rho_{21})(2p+pq)}{4q(2-q+2q^2)} \\[4pt] 
-\frac{(\rho_{12}+\rho_{21})(2p+pq)}{4q(2-q+2q^2)} \\[4pt] 
\frac{2p+(2\rho_{12}+2\rho_{21}-1)pq-(\rho_{12}+\rho_{21}-1)2pq^2}
{4q(2-q+2q^2)}
\end{array}\right],
\qquad
\overrightarrow{\rho} = 
\left[\begin{array}{c}
\rho_{11} \\ \rho_{12} \\ \rho_{21} \\ \rho_{22} \\[2pt]
\hline
\\[-12pt]
0 \\ 0 \\ 0 \\ 0
\end{array}\right].
$$

Similar calculations show that the expected return time for any density concentrated on site 2 is $1+2q/p$.

\smallskip

\noindent(b) {\bf Recurrence of a pure state supported on one site.} 
Consider a pure state  supported on site 1,
$$
\rho_\psi=|\psi\rangle\langle\psi|=|\phi\rangle\langle\phi|\otimes|1\rangle\langle1|,
\qquad
|\psi\rangle=|\phi\rangle\otimes|1\rangle,
\qquad
\phi\in\mathbb{C}^2.
$$ 
The recurrence of $\psi$ is encoded by a reduced Schur function $\mathbb{F}$ built out of the projections $P=\rho_\psi$ and {\color{black}$Q=I-\rho_\psi=(I_2-|\phi\rangle\langle\phi|)\otimes|1\rangle\langle1| \oplus I_2\otimes|2\rangle\langle2|$}. The matrix representation $\widehat{\mathbb{F}}$ of $\mathbb{F}$ follows from those of the corresponding projections on $\mathcal{I}_{\mathcal{S}}(\mathbb{C}^2)$, $\mathbb{P}=P\cdot P$ and $\mathbb{Q}=Q\cdot Q$, which are given by $\widehat{\mathbb{P}}=\lceil\phi\phi^*\rceil\oplus0_4$ and $\widehat{\mathbb{Q}}=\lceil I_2-\phi\phi^*\rceil\oplus I_4$ respectively. The return properties of $\psi$ come from the behaviour around $z=1$ of $\operatorname{Tr}(\mathbb{F}(z)(\rho_\psi))=(\widehat{\mathbb{F}}(z)\,vec(\rho_\psi))_1+(\widehat{\mathbb{F}}(z)\,vec(\rho_\psi))_4$. For instance, 
$$
\begin{aligned}
& \operatorname{Tr}(\mathbb{F}(z)(\rho_\psi)) = 
\frac{(4-6p-4q+2p^2+4pq)z^2+(9p+4q-2p^2-2pq-8)z+4-3p}{(4-3p-4q+2pq)z^2+(3p+4q-8)z+4},
& \qquad & \phi=\begin{bmatrix}1\\0\end{bmatrix},
\\
& \operatorname{Tr}(\mathbb{F}(z)(\rho_\psi)) = 
\frac{(4-4p-4q+p^2+2pq)z^2+(6p+4q-p^2-pq-8)z+4-2p}
{(4-2p-4q+pq)z^2+(2p+4q-8)z+4},
& & \phi=
\begin{bmatrix}\frac{1}{\sqrt{2}}\\[4pt]\frac{1}{\sqrt{2}}\end{bmatrix},
\end{aligned}
$$
which confirms that $\pi(\psi\to\psi)=\operatorname{Tr}(\mathbb{F}(1)(\rho_\psi))=1$, giving $\tau(\psi\to\psi)=1+\operatorname{Tr}(\mathbb{F}'(1)(\rho_\psi))=2+p/q$ for both states. When $p=2q$, due to unitality, we get an integer mean return time $\tau(\psi\to\psi)=4$ which equals the dimension of the relevant Hilbert space, {\color{black}corresponding to the} whole Hilbert space $\mathbb{C}^2\otimes\mathcal{S}$ for both states.

\smallskip

\noindent(c) {\bf Recurrence of a subspace generated by two states on distinct sites}. 
Let us inspect the recurrence of an admissible subspace $\mathcal{K}$ spanned by the states
$$
|\psi_i\rangle = |\phi_i\rangle\otimes|i\rangle, 
\qquad \phi_i\in\mathbb{C}^2, \qquad i=1,2.
$$
The projections related to such a subspace are $P=\sum_{i=1,2}|\psi_i\rangle\langle\psi_i|=\sum_{i=1,2}|\phi_i\rangle\langle\phi_i|\otimes|i\rangle\langle i|$ and $Q=I-P=\sum_{i=1,2}(I_2-|\phi_i\rangle\langle\phi_i|)\otimes|i\rangle\langle i|$, which lead to projections $\mathbb{P}$ and $\mathbb{Q}$ with matrix representations $\widehat{\mathbb{P}}=\oplus_{i=1,2}\lceil\phi_i\phi_i^*\rceil$ and $\widehat{\mathbb{Q}}=\oplus_{i=1,2}\lceil I_2-\phi_i\phi_i^*\rceil$. This provides the matrix representation of the reduced Schur function $\mathbb{F}$ for the subspace $\mathcal{K}$. Taking for instance $\phi_1=[1\;0]^T$ and $\phi_2=[\frac{1}{\sqrt{2}}\;\frac{1}{\sqrt{2}}]^T$, this gives $\pi(\psi_1\to\mathcal{K})=\pi(\psi_2\to\mathcal{K})=1$ and
$$
\tau(\psi_1\to\mathcal{K})=1+\frac{6p+11q+1}{3(1+3q)},
\qquad
\tau(\psi_2\to\mathcal{K})=1+\frac{(1-q)(3p+4q)}{3p(1+3q)},
$$
which, in view of Remark~\ref{rem:rec-ps}, implies that $\mathcal{K}$ is positive recurrent, as expected. For any other $\psi\in\mathcal{K}$, the density has the form 
$
\rho_\psi = 
|\alpha_1|^2 |\psi_1\rangle\langle\psi_1| + 
|\alpha_2|^2 |\psi_2\rangle\langle\psi_2| +
2\operatorname{Re}(\alpha_1\overline{\alpha}_2 |\psi_1\rangle\langle\psi_2|).
$
Thus, the averaged expected return time for $\mathcal{K}$ becomes an average on the set $\{[\alpha_1,\alpha_2]\in\mathbb{C}^2:|\alpha_1|^2+|\alpha_2|^2=1\}$. Due to the symmetries of such an average and the linearity of {\color{black}$\operatorname{Tr}(\mathbb{F}'(1)(\rho))$} on the operator $\rho$, 
$$
\tau(\mathcal{K}\to\mathcal{K}) = 
\int_{\mathcal{D}_p(\mathcal{K})} d\psi\;\tau(\psi\to\mathcal{K}) =
\frac{1}{2} (\tau(\psi_1\to\mathcal{K})+\tau(\psi_2\to\mathcal{K})) = 
1 + \frac{2p+2q+3p^2-2q^2+4pq}{3p(1+3q)}.
$$
In the unital case, $p=2q$, we find that $\tau(\mathcal{K}\to\mathcal{K}) = 2$ in agreement with the fact, already pointed out in (b), that the relevant Hilbert space of $\psi_i$ --and thus of $\mathcal{K}$-- is the four-dimensional Hilbert space $\mathbb{C}^2\otimes\mathcal{S}$.

\smallskip

\noindent(d) {\bf Recurrence of a pure state supported on two sites.}  
Finally, we will consider an example of a non-admissible subspace. We will analyze the recurrence of a pure state $\rho_\psi=|\psi\rangle\langle\psi|$ concentrated on sites 1 and 2, with
\beq \label{astate1}
|\psi\rangle = 
|\phi_1\rangle\otimes|1\rangle+|\phi_2\rangle\otimes|2\rangle =
\frac{1}{\sqrt{2}} \begin{bmatrix}1\\0\end{bmatrix}\otimes|1\rangle +
\frac{1}{\sqrt{2}} \begin{bmatrix}0\\1\end{bmatrix}\otimes|2\rangle.
\eeq
Explicitly,
$$
\rho_\psi = 
|\phi_1\rangle\langle\phi_1|\otimes|1\rangle\langle 1| +
|\phi_1\rangle\langle\phi_2|\otimes|1\rangle\langle 2| +
|\phi_2\rangle\langle\phi_1|\otimes|2\rangle\langle 1| +
|\phi_2\rangle\langle\phi_2|\otimes|2\rangle\langle 2| =
\frac{1}{2}
\begin{bmatrix} 
1 & 0 & 0 & 1 \\ 0 & 0 & 0 & 0 \\ 0 & 0 & 0 & 0 \\ 1 & 0 & 0 & 1
\end{bmatrix},
$$
which shows that $\rho_\psi$ is not a TOM density. Therefore, instead of the map $\mathcal{E}$ acting on $\mathcal{I}_S(\mathbb{C}^2)$, we need to consider the quantum channel $\Phi$ \eqref{eq:TOM2} acting on the entire space $\mathcal{I}(\mathbb{C}^2\otimes\mathcal{S})$, whose matrix representation is obtained by writing the corresponding Kraus matrices for $\mathcal{E}^i_j\otimes|i\rangle\langle j|$. This amounts to replacing $A^i_j$ by $A^i_j \otimes|i\rangle\langle j|$, and analogously for $B^i_j$, leading to a matrix representation of order 16,  
$$
\widehat{\Phi} = \sum_{i=1,2} \lceil L_i^j \rceil + \lceil M_i^j \rceil,
\qquad
L_i^j = A_i^j \otimes P_{ij},
\qquad 
M_i^j = B_i^j \otimes P_{ij},
\qquad
\begin{aligned}
& P_{11}=\left[\begin{smallmatrix}1&0\\[1pt]0&0\end{smallmatrix}\right],
& & P_{12}=\left[\begin{smallmatrix}0&1\\[1pt]0&0\end{smallmatrix}\right],
\\
& P_{21}=\left[\begin{smallmatrix}0&0\\[1pt]1&0\end{smallmatrix}\right],
& & P_{22}=\left[\begin{smallmatrix}0&0\\[1pt]0&1\end{smallmatrix}\right],
\end{aligned}
$$ 
which has exactly the same nonzero eigenvalues (counting multiplicities) as the representation $\widehat{\mathcal{E}}$ of order 8 given in \eqref{eq:ex1-wE}. The map $\Phi$ also acts irreducibly on $\mathcal{I}(\mathbb{C}^2\otimes\mathcal{S})$ since its unique invariant state --essentially a rearrangement of the fixed point $\rho_*$ of $\mathcal{E}$ given in \eqref{fptomf1}--,  
$$
\rho_{**}=\frac{1}{2q+p}
\begin{bmatrix}
q&0&0&0\\0&\frac{p}{2}&0&0\\0&0&q&0\\0&0&0&\frac{p}{2}
\end{bmatrix},
$$
is faithful. From now on, the calculation follows the previous discussions, namely, by setting
$P=\rho_\psi$, $Q=I-\rho_\psi$, $\mathbb{P}=P\cdot P$ and $\mathbb{Q}=Q\cdot Q$. The matrix representations $\widehat{\mathbb{P}}=\lceil\rho_\psi\rceil$ and $\widehat{\mathbb{Q}}=\lceil I_4-\rho_\psi\rceil$ provide the one for the corresponding reduced Schur function $\mathbb{F}(z)=\mathbb{P}\Phi(I-z\mathbb{Q}\Phi)^{-1}\mathbb{P}$, from which we confirm once more that $\pi(\psi\to\psi)=1$ and we get
$$
\tau(\psi\to\psi) =
\frac{2}{3p(2+q+2q^2)}\Big(10p+4q+2p^2-2q^2+3pq+4q^3+p^2q+8pq^2\Big).
$$
Alternatively, the expected return time may be obtained directly from the matrix representation of \eqref{eq:tinv}. When $p=2q$ the CPTP map $\Phi$ is unital, hence the mean return time becomes an integer $\tau(\psi\to\psi)=4$ which shows that the relevant Hilbert space of $\psi$ is the full space $\mathbb{C}^2\otimes\mathcal{S}$.    
\eex
\qee

\subsection{Kac's Lemma for quantum states and subspaces} 
\label{ssec:kac}

Let us recall a classical result for stochastic processes: given an irreducible Markov chain, the existence of an invariant probability distribution $\pi=(\pi_i)$, $\pi_i>0$, $\sum_i \pi_i=1$ is equivalent to the positive recurrence of its states. In this case, Kac's Lemma \cite{durrett,kac} states that the mean return time to any given vertex $i$ is given by $1/\pi_i$. In recent years the problem of finding quantum versions of this lemma has been investigated. In the context of OQWs, versions of Kac's Lemma for the mean return time to some given vertex have been proved in \cite{bardet,cgl} (the result is essentially the same in both works, but the proofs employ different techniques).  
In such a  context  the OQW is assumed to be irreducible and the mean return time to a vertex $|i\rangle$ is conditioned on starting with the $i$-th positive matrix of the stationary density $\chi=\sum_i\chi_i\otimes|i\rangle\langle i|$ of the walk. Then, we have that
\beq\label{eq:kacsites} 
\tau\big({\textstyle\frac{\chi_i}{\operatorname{Tr}(\chi_i)}}\otimes|i\rangle\langle i|\to|i\rangle\big) =
\frac{1}{\operatorname{Tr}(\chi_i)}.
\eeq 
In the setting of quantum channels on finite-dimensional Hilbert spaces, a version of Kac's Lemma has been proved in \cite{sinkovicz2}, regarding the mean return time to an eigenvector $\psi$ of a steady state $\chi$ of a channel $\Phi$. Namely, if $\chi\psi=\lambda\psi$ with an eigenvalue $\lambda=\langle\psi|\chi\psi\rangle\neq0$, then
\beq \label{asbotheq1} 
\tau(\psi\to\psi) = 
\frac{1}{\langle\psi|\chi\psi\rangle}.
\eeq
The authors of the mentioned work also note that the assumption of $\psi$ being an eigenvector of $\chi$ is sufficient for the validity of (\ref{asbotheq1}), but not necessary for certain channels \cite{sinkovicz2}. In light of this, a related question is to ask about  the origin of the failure of \eqref{asbotheq1} for an arbitrary pure state. 

To address the above question, let us make the following considerations, which consist of a variation of the reasoning presented in \cite{sinkovicz2}. Let $\Phi$ be a quantum channel on a finite-dimensional Hilbert space $\mathcal{H}$, and assume for convenience that it is irreducible so that there is a unique invariant state $\chi$, and this state is faithful. According to Theorem~\ref{thm:rec-fi}, every $\psi\in\mathcal{H}$ is positive recurrent. If $\mathbb{Q}=Q\cdot Q$ with $Q=I-\rho_\psi$, we can write
\beq
\label{bas_eq1}
\chi = \langle\psi|\chi\psi\rangle\rho_\psi + 
\mathbb{Q}(\chi) + \varphi_\psi,
\eeq
where $\varphi_\psi=Q\chi\rho_\psi+\rho_\psi\chi Q$ is self-adjoint, traceless and vanishes when $\psi$ is an eigenvector of $\chi$. Then, combining \eqref{eq:tinv} with 
$$
(I-\mathbb{Q}\Phi)(\chi) = (I-\mathbb{Q})(\chi) = 
\langle\psi|\chi\psi\rangle\rho_\psi + \varphi_\psi
$$
yields
{\color{black}\beq\label{gen_kac1}
\tau(\psi\to\psi) = \frac{1}{\langle\psi|\chi\psi\rangle}
\left[1-\operatorname{Tr}
\left((I-\mathbb{Q}\Phi)^{-1}(\varphi_\psi)\right) 
\right],
\eeq}
The above expression makes clear the difference between the calculation of the mean return time of a state which is an eigenvector of the stationary density and the mean return time for a general state. In the latter case one has to consider a correction term generated by $\varphi_\psi$. 

\bex\label{example_kac1} 
The map
$$
\Phi(\rho)=B_1\rho B_1^*+B_2\rho B_2^*,
\qquad
B_1 = 
\begin{bmatrix} 
\frac{1}{\sqrt{3}} & \frac{1}{\sqrt{2}} \\[4pt] \frac{1}{\sqrt{3}} & 0
\end{bmatrix},
\qquad
B_2 = 
\begin{bmatrix} 
\frac{1}{\sqrt{3}} & -\frac{1}{\sqrt{2}} \\[4pt] 0 & 0 
\end{bmatrix},
$$
is an irreducible non-unital quantum channel because its unique fixed point is
$$
\chi = \frac{1}{4}
\begin{bmatrix} 
3 & \frac{1}{5}(6+\sqrt{6}) \\ \frac{1}{5}(6+\sqrt{6}) & 1
\end{bmatrix}.
$$
Given $\psi=[\frac{1}{\sqrt{2}}\;\frac{1}{\sqrt{2}}]^T$, the use of \eqref{eq:tinv} yields
$$
\tau(\psi\to\psi) = 2\frac{21-\sqrt{6}}{29} \approx 1.27934.
$$
We can split this mean return time according to the contributions identified in \eqref{gen_kac1}. 
Since
$$
\varphi_\psi = 2 \operatorname{Re}(Q\chi\rho_\psi) = 
\frac{1}{4} \begin{bmatrix} 1 & 0 \\ 0 & -1\end{bmatrix},
$$
we find that
$$
\frac{1}{\langle\psi|\chi\psi\rangle} = {{\frac{20}{16+\sqrt{6}}}}
\approx 1.08404,
\qquad
1-\operatorname{Tr}
\left((\mathbb{I}-\mathbb{Q}\Phi)^{-1}(\varphi_\psi)\right) = 
1 + \frac{8+\sqrt{6}}{58}
\approx 1.18016,
$$
thus the correction to $\tau(\psi\to\psi)$ coming from $\varphi_\psi$ is around 18\% of the ideal value $1/\langle\psi|\chi\psi\rangle$.
\eex
\qee

Some generalizations of the quantum version of Kac's Lemma arise by considering higher dimensional return subspaces. Let $\Phi$ be a quantum channel on a Hilbert space $\mathcal{H}$, and $P$ the orthogonal projection of $\mathcal{H}$ onto a finite-dimensional subspace $\mathcal{H}_0\subset\mathcal{H}$. Assuming $\dim\mathcal{H}<\infty$ and $\Phi$ irreducible not only guarantees that $\mathcal{H}_0$ is positive recurrent, but permits to apply \eqref{eq:tinv}, leading to an averaged mean return time
\beq\label{eq:avert}
\tau(\mathcal{H}_0\to\mathcal{H}_0) = 
\frac{\operatorname{Tr}\left((I-\mathbb{Q}\Phi)^{-1}(P)\right)}{\dim\mathcal{H}_0},
\eeq
where we have used \eqref{eq:dpsi}. 

On the other hand, if $\chi$ is the unique invariant state of $\Phi$ and $P$ commutes with $\chi$, then
$\chi = P\chi P + Q\chi Q$,
so that 
\beq\label{eq:kacinv}
(I-\mathbb{Q}\Phi)(\chi) = P\chi P, 
\eeq
and \eqref{eq:tinv} gives
\beq\label{eq:kacsub}
\tau(\rho\to\mathcal{H}_0) =
\frac{1}{\operatorname{Tr}(P\chi P)},
\qquad
\rho=\frac{P\chi P}{\operatorname{Tr}(P\chi P)}.
\eeq
This holds in particular for a TOM if $\mathcal{H}_0$ is a site $|i\rangle$ because then the corresponding orthogonal projection $P=P_i$ commutes with an invariant state $\chi=\sum_{j\in V}\chi_j\otimes|j\rangle\langle j|$ since $P_i\chi=\chi_i\otimes|i\rangle\langle i|=\chi P_i$. The application of \eqref{eq:kacsub} to this case extends the relation \eqref{eq:kacsites} to TOMs. This generalizes easily to sums of sites $\oplus_i|i\rangle$ by choosing $P=\sum_iP_i$, leading to
$$
\tau\left(
\frac{\sum_i\chi_i\otimes|i\rangle\langle i|}{\sum_i\operatorname{Tr}(\chi_i)}
\to\oplus_i|i\rangle\right) =
\frac{1}{\textstyle\sum_i\operatorname{Tr}(\chi_i)}.
$$ 

Another situation where the projection $P$ commutes with the invariant state $\chi$ arises when $\mathcal{H}_0$ is a subspace of an eigenspace of $\chi$, since then $P\chi=\lambda P=\chi P$ where $\lambda$ is the corresponding eigenvalue. In this case, in view of \eqref{eq:avert}, the relation \eqref{eq:kacinv} yields 
$$
\tau(\mathcal{H}_0\to\mathcal{H}_0) = \frac{1}{\lambda \dim\mathcal{H}_0}
$$
for $\lambda\ne0$, which is a higher-dimensional version of \eqref{asbotheq1}.

\section{Recurrence splitting rules for quantum Markov chains}
\label{sec:SPLIT}

In this section we begin our description of splitting rules for recurrence in quantum Markov chains, which constitute an application of results presented in \cite{gvfr} for operators on Banach spaces. When dealing with the recurrence properties of arbitrary admissible subspaces, the splitting rules run into some technical difficulties which arise from the fact that $\mathbb{P}+\mathbb{Q}$ is not necessarily the identity. While we expect to address this general case in a future work, in this section we restrict our attention to the recurrence properties of sums of sites in TOMs, where the alluded drawback disappears. The results will be illustrated with concrete {\color{black}TOM} examples.

Consider a graph with a set $V$ of vertices generating site space $\mathcal{S}=\operatorname{span}\{|i\rangle\}_{i\in V}$, a Hilbert space $\mathcal{H}$ of internal degrees of freedom and the trace-class subspace $\mathcal{I}_\mathcal{S}(\mathcal{H})=\oplus_{i\in V} \mathcal{I}(\mathcal{H}) \otimes |i\rangle\langle i|$ where a TOM on such a graph acts. Consider ``left" and ``right" subgraphs with sets of vertices $V_L,V_R\subset V$ overlapping on $V_0=V_L\cap V_R$. Given TOMs $\mathcal{E}_{L,R}$ on the left and right subgraphs, there are two natural ways of combining them to generate a TOM on the larger graph: if $V_- = V_L \setminus V_0$ and $V_+ = V_R \setminus V_0$, the partition $V = V_- \cup V_0 \cup V_+$ yields an obvious decomposition of site space
$$
\mathcal{S} = \mathcal{S}_- \oplus \mathcal{S}_0 \oplus \mathcal{S}_+ =
\mathcal{S}_L \oplus \mathcal{S}_+ = 
\mathcal{S}_- \oplus {S}_R,
$$
so that
\begin{equation} \label{eq:subdec}
\mathcal{I}_\mathcal{S}(\mathcal{H}) = 
\mathcal{I}_{\mathcal{S}_-} \oplus \mathcal{I}_{\mathcal{S}_0} \oplus \mathcal{I}_{\mathcal{S}_+} =
\mathcal{I}_{\mathcal{S}_L} \oplus \mathcal{I}_{\mathcal{S}_+} = 
\mathcal{I}_{\mathcal{S}_-} \oplus \mathcal{I}_{\mathcal{S}_R}.
\end{equation} 
The TOMs $\mathcal{E}_{L,R}$ may be trivially extended to $\mathcal{I}_\mathcal{S}(\mathcal{H})$ using the null or identity operator on $\mathcal{I}_{\mathcal{S}_\pm}(\mathcal{H})$, which we denote by $0_\pm$ and $I_\pm$ respectively. Then, we can define the following maps on $\mathcal{I}_\mathcal{S}(\mathcal{H})$:
\begin{equation}  \label{eq:fac}
(\mathcal{E}_L \oplus I_+) (I_- \oplus \mathcal{E}_R),
\end{equation}
and 
\begin{equation} \label{eq:dec0}
(\mathcal{E}_L \oplus 0_+) + (0_- \oplus \mathcal{E}_R).
\end{equation} 
The product \eqref{eq:fac} is a TOM, but the sum \eqref{eq:dec0} is not since, denoting by $\rho=\rho_-\oplus\rho_0\oplus\rho_+=\rho_L\oplus\rho_+=\rho_-\oplus\rho_R$ the decompositions of $\rho\in\mathcal{I}_\mathcal{S}(\mathcal{H})$ related to \eqref{eq:subdec} yields
$$
\operatorname{Tr}((\mathcal{E}_L \oplus 0_+)(\rho)+(0_-+\mathcal{E}_R)(\rho)) =
\operatorname{Tr}(\mathcal{E}_L(\rho_L)) + \operatorname{Tr}(\mathcal{E}_R(\rho_R)) =
\operatorname{Tr}(\rho_L) + \operatorname{Tr}(\rho_R) = 
\operatorname{Tr}(\rho) + \operatorname{Tr}(\rho_0).
$$
Nevertheless, this result shows that any TOM $\mathcal{E}_0$ on $\mathcal{I}_{\mathcal{S}_0}(\mathcal{H})$ generates a TOM on the original graph   
\begin{equation} \label{eq:dec}
(\mathcal{E}_L \oplus 0_+) + (0_- \oplus \mathcal{E}_R) 
- (0_- \oplus \mathcal{E}_0 \oplus 0_+),
\end{equation}
whenever $(\mathcal{E}_L)_i^j + (\mathcal{E}_R)_i^j - (\mathcal{E}_0)_i^j$
is a CP map for every $i,j\in V_0$. 
  
We refer to any TOM with the form \eqref{eq:fac}/\eqref{eq:dec} as an {\bf overlapping factorization/decomposition} which overlaps on $\mathcal{I}_{\mathcal{S}_0}(\mathcal{H})$, or, for short, with overlapping site subspace $\mathcal{S}_0$. Our purpose is to develop splitting recurrence rules allowing us to infer recurrence properties of the overlapping site subspace for such composite systems, starting from those of the left and right subsystems. The FR-function machinery becomes essential for this objective thanks to the special splitting properties of such functions, summarized in \cite[Sect.~6]{gvfr}. Specialized to our situation, these splitting rules for FR-functions imply that the Schur function $f$ of the subspace $\mathcal{H}\otimes\mathcal{S}_0$ for the TOM on the whole graph is related to the Schur functions $f_{L,R}$ of the same subspace for the TOMs on the L,R subgraphs by
\beq \label{eq:facdec}
\begin{aligned}
& \mathcal{E} = (\mathcal{E}_L \oplus 0_+) + (0_- \oplus \mathcal{E}_R) 
- (0_- \oplus \mathcal{E}_0 \oplus 0_+) 
\quad\Longrightarrow\quad
f = f_L+f_R-\mathcal{E}_0,
\\
& \mathcal{E} = (\mathcal{E}_L \oplus I_+) (I_- \oplus \mathcal{E}_R)
\quad\Longrightarrow\quad
f = f_Lf_R.
\end{aligned}
\eeq
We will examine the consequences of these splitting rules for the recurrence properties of $\mathcal{H}\otimes\mathcal{S}_0=\oplus_{i\in V_0}\mathcal{H}\otimes|i\rangle$, which, for short, we refer to as the {\color{black} {\bf sum of sites $\mathcal{S}_0$}}. It is convenient to remember that the Schur functions and reduced Schur functions coincide for any such a sum of sites, so that Proposition~\ref{pro:rec-F} holds for $f=\mathbb{F}$ and $f_{L,R}=\mathbb{F}_{L,R}$. In general, we will distinguish the objects and quantities related to the left and right subsystems by the subscript $L$ and $R$ respectively.

\begin{remark} \label{rem:SR-GEN}
The splitting rules \eqref{eq:facdec} are not only interesting to reduce return properties of a TOM to return properties of smaller ones, but also as a divide and conquer method to obtain Schur functions. In this regard we must point out that the above splitting rules are valid for arbitrary linear operators \cite[Sect.~6]{gvfr}. We will use sometimes this fact for the calculation of Schur functions for an infinite-dimensional TOM $\mathcal{E}$ by splitting it into trace non-increasing maps $\mathcal{E}_{L/R}$.
\end{remark}

\begin{theorem}[TOM recurrence decomposition rules for sums of sites] \label{thm:dec}
Suppose that the TOM of a quantum Markov chain on $\mathcal{H}\otimes\mathcal{S}$ has an overlapping decomposition $\mathcal{E} = (\mathcal{E}_L \oplus 0_+) + (0_- \oplus \mathcal{E}_R) - (0_- \oplus \mathcal{E}_0 \oplus 0_+)$ with overlapping site subspace $\mathcal{S}_0$, and let $f_{L,R}$ be the Schur functions of the sum of sites $\mathcal{S}_0$ with respect to $\mathcal{E}_{L,R}$. Then, we have the following consequences: 
\begin{itemize}
\item[a)] The return probability satisfies
\beq \label{eq:pdec}
\pi(\rho\to\mathcal{S}_0) = 
\pi_L(\rho\to\mathcal{S}_0) + \pi_R(\rho\to\mathcal{S}_0) - 1,
\qquad
\forall\rho\in\mathcal{D}_{\mathcal{S}_0}(\mathcal{H}).
\eeq
This has the following implications: 
\begin{itemize}
\item[(i)] If $\mathcal{S}_0$ is recurrent for $\mathcal{E}_L$, then $\pi(\rho\to\mathcal{S}_0)=\pi_R(\rho\to\mathcal{S}_0)$ for all $\rho\in\mathcal{D}_{\mathcal{S}_0}(\mathcal{H})$.
\item[(ii)] If $\mathcal{S}_0$ is recurrent for $\mathcal{E}_R$, then $\pi(\rho\to\mathcal{S}_0)=\pi_L(\rho\to\mathcal{S}_0)$ for all $\rho\in\mathcal{D}_{\mathcal{S}_0}(\mathcal{H})$.
\item[(iii)] At least one of $\pi_{L,R}(\rho\to\mathcal{S}_0)$ must be greater than or equal to $1/2$ for each $\rho\in\mathcal{D}_{\mathcal{S}_0}(\mathcal{H})$.
\item[(iv)] The following statements are equivalent:
\begin{enumerate}
\item $\mathcal{S}_0$ is recurrent for $\mathcal{E}$.
\item $\mathcal{S}_0$ is recurrent for $\mathcal{E}_L$ and  $\mathcal{E}_R$.
\item $f_{L,R}(1)$ are trace preserving.
\end{enumerate}
\end{itemize}
\item[b)] The expected return time satisfies
\beq \label{eq:tdec}
\tau(\rho\to\mathcal{S}_0) = 
\tau_L(\rho\to\mathcal{S}_0) + \tau_R(\rho\to\mathcal{S}_0) - 1,
\qquad
\forall\rho\in\mathcal{D}_{\mathcal{S}_0}(\mathcal{H}).
\eeq
Therefore, the following statements are equivalent:
\begin{enumerate}
\item $\mathcal{S}_0$ is positive recurrent for $\mathcal{E}$.
\item $\mathcal{S}_0$ is positive recurrent for $\mathcal{E}_L$ and  $\mathcal{E}_R$.
\item $f_{L,R}(1)$ are trace preserving and $f'_{L,R}(1)$ exist in the strong sense.
\end{enumerate}
\end{itemize}
\end{theorem}

{\bf Proof.}
Let $f$ be the Schur function of the sum of sites $\mathcal{S}_0$ with respect to $\mathcal{E}$. From the identity $f=f_L+f_R-\mathcal{E}_0$ and Proposition~\ref{pro:rec-F}.(i) we find that, for every $\rho\in\mathcal{D}_{\mathcal{S}_0}(\mathcal{H})$,
$$
\pi(\rho\to\mathcal{S}_0) = \operatorname{Tr}(f(1)(\rho)) = 
\operatorname{Tr}(f_L(1)(\rho)) + \operatorname{Tr}(f_R(1)(\rho)) - 
\operatorname{Tr}(\mathcal{E}_0(\rho)) =
\pi_L(\rho\to\mathcal{S}_0) + \pi_R(\rho\to\mathcal{S}_0) - 1.
$$
The results (i)-(iv) are direct implications of this identity and Proposition~\ref{pro:rec-F}.(ii).

As for the expected return time, we cannot assume the existence of $f'(1):=\lim_{x\uparrow1}f(x)$, thus we deal with limits. The second relation in \eqref{eq:ptf} yields
$$
\tau(\rho\to\mathcal{S}_0) = 1 + \lim_{x\uparrow1}\operatorname{Tr}(f'(x)(\rho)) = 
1 + \lim_{x\uparrow1}\operatorname{Tr}(f'_L(x)(\rho)) + 
\lim_{x\uparrow1}\operatorname{Tr}(f'_R(x)(\rho)) =
\tau_L(\rho\to\mathcal{S}_0) + \tau_R(\rho\to\mathcal{S}_0) - 1
$$
whenever $\pi(\rho\to\mathcal{S}_0)=1$, since this is equivalent to $\pi_{L,R}(\rho\to\mathcal{S}_0)=1$. If $\pi(\rho\to\mathcal{S}_0)<1$, then one of $\pi_{L,R}(\rho\to\mathcal{S}_0)$ is smaller than one, hence \eqref{eq:tdec} is trivially true because both sides {\color{black} are infinite}. The equivalences in b) are immediate from this identity and Proposition~\ref{pro:rec-F}.(iii).
\qed

Regarding the consequences of the overlapping factorization of TOMs, we will use the following simple result.

\begin{lemma} \label{lem:lim}
If $x_{n,m}\in\mathbb{R}$ are non-decreasing in $n$ and $m$, then $\lim_m\lim_nx_{m,n}=\lim_nx_{n,n}$.
\end{lemma}

{\bf Proof.}
As non-decreasing sequences in $n$, both $x_{m,n}$ and $x_{n,n}$ have a finite or infinite limit when $n\to\infty$. The result follows from the inequality
$x_{m,m}\le x_{m,n}\le x_{n,n}$ for $m\le n$, which implies that $x_{m,m}\le \lim_nx_{m,n}\le \lim_nx_{n,n}$, and finally gives $\lim_mx_{m,m}\le \lim_m\lim_nx_{m,n}\le \lim_nx_{n,n}$.
\qed

Now we are ready to examine the recurrence implications of overlapping factorizations of TOMs. They are given by the following proposition, where, for any positive trace-class operator $\rho\ne0$, we denote by $\widehat\rho$ the density obtained by normalizing $\rho$, i.e.
$$
\widehat\rho = \frac{\rho}{\operatorname{Tr}(\rho)}.
$$

\begin{theorem}[TOM recurrence factorization rules for sums of sites] \label{thm:fac}
Suppose that the TOM of a quantum Markov chain on $\mathcal{H}\otimes\mathcal{S}$ has an overlapping factorization $\mathcal{E} = (\mathcal{E}_L \oplus I_+) (I_- \oplus \mathcal{E}_R)$ with overlapping site subspace $\mathcal{S}_0$, and let $f_{L,R}$ be the Schur functions of the sum of sites $\mathcal{S}_0$ with respect to $\mathcal{E}_{L,R}$. Then, we have the following consequences: 
\begin{itemize}
\item[a)] The return probability satisfies
\beq \label{eq:pfac}
\pi(\rho\to\mathcal{S}_0) =
\pi_L(\widehat{f_R(1)(\rho)}\to\mathcal{S}_0) \, \pi_R(\rho\to\mathcal{S}_0),
\qquad
\forall\rho\in\mathcal{D}_{\mathcal{S}_0}(\mathcal{H}).
\eeq
This has the following implications: 
\begin{itemize}
\item[(i)] If $\mathcal{S}_0$ is recurrent for $\mathcal{E}_L$, then $\pi(\rho\to\mathcal{S}_0)=\pi_R(\rho\to\mathcal{S}_0)$ for all $\rho\in\mathcal{D}_{\mathcal{S}_0}(\mathcal{H})$.
\item[(ii)] If $\mathcal{S}_0$ is recurrent for $\mathcal{E}_R$, then $\pi(\rho\to\mathcal{S}_0)=\pi_L(f_R(1)(\rho)\to\mathcal{S}_0)$ for all $\rho\in\mathcal{D}_{\mathcal{S}_0}(\mathcal{H})$.
\item[(iii)] We have the relations $(1)\Leftarrow(2)\Leftarrow(3)\Leftrightarrow(4)$ among the following statements:
\begin{enumerate}
\item $\mathcal{S}_0$ is recurrent for $\mathcal{E}_R$.
\item $\mathcal{S}_0$ is recurrent for $\mathcal{E}$.
\item $\mathcal{S}_0$ is recurrent for $\mathcal{E}_L$ and  $\mathcal{E}_R$.
\item $f_{L,R}(1)$ are trace preserving.
\end{enumerate}
\end{itemize}
\item[b)] If the sum of sites $\mathcal{S}_0$ is recurrent for $\mathcal{E}_L$, the expected return time satisfies
\beq \label{eq:tfac}
\tau(\rho\to\mathcal{S}_0) = 
\tau_L(\widehat{f_R(1)(\rho)}\to\mathcal{S}_0) + \tau_R(\rho\to\mathcal{S}_0) - 1,
\qquad
\forall\rho\in\mathcal{D}_{\mathcal{S}_0}(\mathcal{H}).
\eeq
 
\noindent Therefore, we have the relations $(1)\Leftarrow(2)\Leftrightarrow(3)$ among the following statements:
\begin{enumerate}
\item $\mathcal{S}_0$ is positive recurrent for $\mathcal{E}$.
\item $\mathcal{S}_0$ is positive recurrent for $\mathcal{E}_L$ and  $\mathcal{E}_R$.
\item $f_{L,R}(1)$ are trace preserving and $f'_{L,R}(1)$ exist in the strong sense.
\end{enumerate}
\end{itemize}
\end{theorem}

{\bf Proof.}
The Schur function $f$ of the sum of sites $\mathcal{S}_0$ with respect to $\mathcal{E}$ factorizes as $f=f_Lf_R$. Therefore, Proposition~\ref{pro:rec-F}.(i) implies that, for every $\rho\in\mathcal{D}_{\mathcal{S}_0}(\mathcal{H})$,
$$
\pi(\rho\to\mathcal{S}_0) = 
\operatorname{Tr}(f_L(1)f_R(1)(\rho)) =
\operatorname{Tr}(f_L(1)(\widehat{f_R(1)(\rho)})) \, \operatorname{Tr}(f_R(1)(\rho)) =
\pi_L(\widehat{f_R(1)(\rho)}\to\mathcal{S}_0) \, \pi_R(\rho\to\mathcal{S}_0).
$$
Then, (i)-(iii) follow directly from this identity and Proposition~\ref{pro:rec-F}.(ii).

On the other hand, using \eqref{eq:ptf} we find that
$$
\tau(\rho\to\mathcal{S}_0) = 
1 + \lim_{x\uparrow1}\operatorname{Tr}(f'_L(x)f_R(x)(\rho)) 
+ \lim_{x\uparrow1}\operatorname{Tr}(f_L(x)f'_R(x)(\rho)). 
$$
The functions $\operatorname{Tr}(f'_L(x)f_R(y)(\rho))$ and $\operatorname{Tr}(f_L(y)f'_R(x)(\rho))$ are non-decreasing in $x$ and $y$ for $x,y\in(0,1)$. Therefore, applying Lemma~\ref{lem:lim} and \eqref{eq:ptf} to the previous identity yields
$$
\begin{aligned}
\tau(\rho\to\mathcal{S}_0) 
& = 1 + \lim_{x\uparrow1}\operatorname{Tr}(f'_L(x)f_R(1)(\rho)) 
+ \lim_{x\uparrow1}\operatorname{Tr}(f_L(1)f'_R(x)(\rho))
\\
& = \tau_L(f_R(1)(\rho)\to\mathcal{S}_0) + \tau_R(\rho\to\mathcal{S}_0) - 1,
\kern70pt
\text{ if } \pi_R(\rho\to\mathcal{S}_0)=1,
\end{aligned}
$$
where we have taken into account that $\pi_R(\rho\to\mathcal{S}_0)=1$ yields that $f_R(1)(\rho)$ is a density, while $f_L(1)$ is trace preserving due to Proposition~\ref{pro:rec-F}.(ii) because we assume the recurrence of the sum of sites $\mathcal{S}_0$ with respect to $\mathcal{E}_L$.
The relation \eqref{eq:tfac} extends the one above to the case $\pi_R(\rho\to\mathcal{S}_0)<1$ since, in that case, it gives $\tau(\rho\to\mathcal{S}_0)=\tau_R(\rho\to\mathcal{S}_0)=\infty$, which agrees with the fact that \eqref{eq:pfac} implies that $\pi(\rho\to\mathcal{S}_0)<1$. The implications in b) follow directly from \eqref{eq:tfac}, a) and Proposition~\eqref{pro:rec-F}.
\qed

\begin{remark} \label{rem:indep}
The above splitting rules have some striking consequences on the return probability in quantum Markov chains. Consider an overlapping splitting (decomposition or factorization) of a TOM $\mathcal{E}$ and suppose that the overlapping sum of sites $\mathcal{S}_0$ is recurrent, for instance, for the left TOM $\mathcal{E}_L$. Then, from (i) in Theorems~\ref{thm:dec} and \ref{thm:fac} we conclude that no change in $\mathcal{E}_L$ preserving the recurrent character of $\mathcal{S}_0$ can modify the return probability to $\mathcal{S}_0$ with respect to $\mathcal{E}$ for the states supported on $\mathcal{S}_0$. According to Theorem~\ref{thm:rec-fi} and \ref{thm:rec-ui}, for example, this is the case when the left TOM acts on a finite number of sites with a finite-dimensional space of internal degrees of freedom, and the modifications in $\mathcal{E}_L$ keep it irreducible or unital. From Theorem~\ref{thm:dec}.(ii), a similar independence result holds for modifications in the right TOM $\mathcal{E}_R$ in case of overlapping decompositions. We will illustrate this phenomenon in the examples below. In particular, an application of this principle to an OQW on an infinite lattice, where such a kind of qualitative result seems hard to be derived directly, will show the power of the splitting rules.
\end{remark}

\section{Characterizations of overlapping splittings}
\label{sec:CHAR}

The usefulness of the splitting rules depends on our ability to detect overlapping splittings. In this section we will search for practical characterizations of overlapping decompositions and factorizations of TOMs which help in their identification and, in consequence, in the application of the recurrence splitting rules given in Theorems~\ref{thm:dec} and \ref{thm:fac}.   

An overlapping decomposition \eqref{eq:dec} of a TOM has an obvious implication, which becomes manifest when looking at its block structure,
\beq\label{tom_dec1}
\mathcal{E} = 
\left[\begin{array}{ccc|c}
 & & & \\[-1pt]
 & \mathcal{E}_L & & \\
 & & & \\[-1pt]
 \hline
 & & & \\[-10pt]
 & & & 0_+ \\[-9pt]
 & & &
\end{array}\right] +
\left[\begin{array}{c|ccc}
 & & & \\[-9pt]
 0_- & & & \\
 & & & \\[-10pt]
 \hline
 & & & \\[-1pt]
 & & \mathcal{E}_R & \\
 & & & \\[-1pt]
\end{array}\right] -
\left[\begin{array}{c|c@{\kern2pt}c@{\kern2pt}c|c}
 & & & & \\[-9pt]
 0_- & & & & \\
 & & & & \\[-10pt]
 \hline
 & & & & \\[-10pt]
 & & \mathcal{E}_0 & & \\
 & & & & \\[-10pt]
 \hline
 & & & & \\[-10pt]
 & & & & 0_+ \\[-9pt]
 & & & & 
\end{array}\right] = 
\left[\begin{array}{c@{\kern7pt}|c@{\kern3pt}c@{\kern3pt}c|@{\kern7pt}c}
 & & & & \\[-9pt]
 \times & & \times & & \\
 & & & & \\[-10pt]
 \hline
 & & & & \\[-10pt]
 \times & & \times & & \times \\
 & & & & \\[-10pt]
 \hline
 & & & & \\[-10pt]
 & & \times & & \times \\[-9pt]
 & & & & 
\end{array}\right]. 
\eeq
The blocks of $\mathcal{E}$ in the upper-right and lower-left corner must vanish, i.e. $\mathcal{E}_i^j=0$ when $i\in V_-$, $j\in V_+$ or $i\in V_+$, $j\in V_-$. Actually, this simple condition is not only necessary, but also sufficient for the existence of an overlapping decomposition.

\begin{theorem} \label{thm:cardec}
A TOM $\mathcal{E}=[\mathcal{E}_i^j]_{i,j\in V}$ has a decomposition $\mathcal{E} = (\mathcal{E}_L \oplus 0_+) + (0_- \oplus \mathcal{E}_R) - (0_- \oplus \mathcal{E}_0 \oplus 0_+)$ which overlaps on the sum of sites labelled by the subset of vertices $V_0\subset V$ iff for some partition $V=V_-\cup V_0\cup V_+$, 
\beq\label{eq:condec}
\mathcal{E}_i^j=0 \text{ for } 
i\in V_-, j\in V_+ \text{ and } i\in V_+, j\in V_-.
\eeq
\end{theorem}

{\bf Proof.}
The condition \eqref{eq:condec} is necessary because any of the three terms in \eqref{eq:dec} has null $(i,j)$ coefficients for $i\in V_-$, $j\in V_+$ and $i\in V_+$, $j\in V_-$.

Let us prove the sufficiency of \eqref{eq:condec}, a condition which we assume in what follows. If $V_L=V_-\cup V_0$ and $V_R=V_0\cup V_+$, then an overlapping decomposition related to the partition $V=V_-\cup V_0\cup V_+$ determines $(\mathcal{E}_L)_i^j=\mathcal{E}_i^j$ for $(i,j)\in V_L^2\setminus V_0^2$ and $(\mathcal{E}_R)_i^j=\mathcal{E}_i^j$ for $(i,j)\in V_R^2\setminus V_0^2$. 

Consider first the case in which $V_0$ has a single vertex, denoted by 0. The truncations 
$$
\widetilde{\mathcal{E}}_L = 
\sum_{i,j\in V_L} \mathcal{E}_i^j\otimes 
|i\rangle\langle j|\cdot|j\rangle\langle i|,
\qquad
\widetilde{\mathcal{E}}_R = 
\sum_{i,j\in V_R} \mathcal{E}_i^j \otimes 
|i\rangle\langle j|\cdot|j\rangle\langle i|,
$$ 
fail to be TOMs only because their 0 column sum may not be trace preserving. We can take care of this problem  by adding CP maps to the $(0,0)$ coefficient in the following way
$$
\mathcal{E}_L = \widetilde{\mathcal{E}}_L 
+ \sum_{i\in V_+} \mathcal{E}_i^0 \otimes 
|0\rangle\langle0|\cdot|0\rangle\langle0|,
\qquad
\mathcal{E}_R = \widetilde{\mathcal{E}}_R 
+ \sum_{i\in V_-} \mathcal{E}_i^0 \otimes 
|0\rangle\langle0|\cdot|0\rangle\langle0|.
$$
Then, 
$$
\sum_{i\in V}(\mathcal{E}_L)_i^0 
=\sum_{i\in V}\mathcal{E}_i^0
=\sum_{i\in V}(\mathcal{E}_R)_i^0
$$
are trace preserving because $\mathcal{E}$ is a TOM, hence $\mathcal{E}_{L,R}$ are also TOMs. Besides,
$$
\mathcal{E}_0 = 
\sum_{i\in V}\mathcal{E}_i^0 \otimes 
|0\rangle\langle0|\cdot|0\rangle\langle0|
$$
is a TOM such that $\mathcal{E} = (\mathcal{E}_L \oplus 0_+) + (0_- \oplus \mathcal{E}_R) - (0_- \oplus \mathcal{E}_0 \oplus 0_+)$ because
$$
\mathcal{E}_0^0 = 
(\mathcal{E}_L)_0^0 + (\mathcal{E}_R)_0^0 - \sum_{i\in V}\mathcal{E}_i^0.
$$ 

If $V_0$ has more than one vertex, it suffices to split 
$$
\sum_{i\in V_+}\mathcal{E}_i^k = \sum_{j\in V_0}\mathcal{A}_j^k,
\qquad
\sum_{i\in V_-}\mathcal{E}_i^k = \sum_{j\in V_0}\mathcal{B}_j^k,
\qquad
k\in V_0,
$$ 
arbitrarily into CP maps $\mathcal{A}_j^k$ and $\mathcal{B}_j^k$. Then, 
$$
\begin{gathered}
\mathcal{E}_L = \widetilde{\mathcal{E}}_L 
+ \sum_{j,k\in V_0} \mathcal{A}_j^k \otimes 
|j\rangle\langle k|\cdot|k\rangle\langle j|,
\qquad
\mathcal{E}_R = \widetilde{\mathcal{E}}_R 
+ \sum_{j,k\in V_0} \mathcal{B}_j^k \otimes 
|j\rangle\langle k|\cdot|k\rangle\langle j|,
\\
\mathcal{E}_0 = 
\sum_{j,k\in V_0} (\mathcal{A}_j^k+\mathcal{B}_j^k+\mathcal{E}_j^k) \otimes 
|j\rangle\langle k|\cdot|k\rangle\langle j|,
\end{gathered}
$$
are TOMs giving an overlapping decomposition of $\mathcal{E}$.
\qed

Pictorially, the above theorem states that overlapping decompositions are linked to TOMs whose transition diagrams have the structure below.
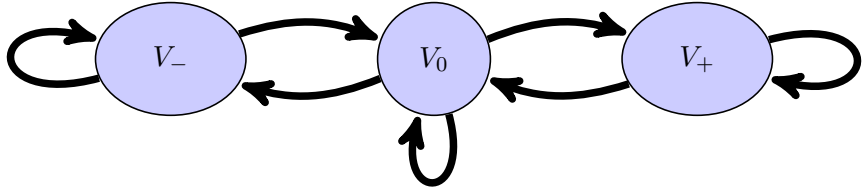
\begin{center}
\begin{tikzpicture}[->,>=stealth',shorten >=1pt,auto,node distance=1.7cm,
                    semithick]
  \tikzstyle{every state}=[ellipse,minimum height=1.5cm, 
    		minimum width=2cm,fill=blue!20,draw,inner sep=0pt,
			node distance=5cm]
  
  \node[state] (L)  			{$V_-$};
  \node[state] (R)  at (7,0) 	{$V_+$};
  \node[state,minimum height=1.5cm,minimum width=1.5cm] (0)  
  								at (3.5,0)  {$V_0$};
  
  \path
        (0)  edge   [loop below,line width=1mm]   	node {} (0)
        (L)  edge   [loop left,line width=1mm]   	node {} (L)
        (R)  edge   [loop right,line width=1mm]   	node {} (R)
        (0)  edge	[bend left=20,line width=1mm] 	node {} (L)
        (0)  edge   [bend left=20,line width=1mm] 	node {} (R)
        (L)  edge   [bend left=20,line width=1mm] 	node {} (0)
        (R)  edge   [bend left=20,line width=1mm] 	node {} (0);
        
\end{tikzpicture}
\end{center}
That is, a decomposition overlapping on a sum of sites corresponding to a certain subset of vertices occurs iff the remaining vertices split into two uncoupled subsets. 

Regarding overlapping factorizations \eqref{eq:fac}, they impose also a simple constraint on the shape of the whole TOM, namely,
$$
\mathcal{E} = 
\left[\begin{array}{ccc|c}
 & & & \\[-1pt]
 & \mathcal{E}_L & & \\
 & & & \\[-1pt]
\hline
 & & & \\[-10pt]
 & & & I_+ \\[-9pt]
 & & &
\end{array}\right] 
\left[\begin{array}{c|ccc}
 & & & \\[-9pt]
 I_- & & & \\
 & & & \\[-10pt]
 \hline
 & & & \\[-1pt]
 & & \mathcal{E}_R & \\
 & & & \\[-1pt]
\end{array}\right] = 
\left[\begin{array}{c@{\kern7pt}|c@{\kern3pt}c@{\kern3pt}c@{\kern7pt}c}
 & & & & \\[-9pt]
 \times & & \times & & \times \\
 & & & & \\[-10pt]
 & & & & \\[-10pt]
 \times & & \times & & \times \\
 & & & & \\[-10pt]
 \hline
 & & & & \\[-10pt]
 & & \times & & \times \\[-9pt]
 & & & & 
\end{array}\right]. 
$$
Now, only the block of $\mathcal{E}$ in the lower-left corner must vanish, i.e. $\mathcal{E}_i^j=0$ when $i\in V_+$, $j\in V_-$. Nevertheless, this condition alone is not sufficient for the existence of an overlapping factorization, which requires an additional constraint on the upper-right block. To formulate this constraint we introduce the notion of {\bf CP vector}, $\mathcal{V}=[\mathcal{V}_i]$, as a vector constituted by CP maps $\mathcal{V}_i$ on a common space $\mathcal{I}(\mathcal{H})$, and such that $\sum_i\mathcal{V}_i$ converges in the strong sense. When $\sum_i\mathcal{V}_i$ is trace preserving, $\mathcal{V}$ will be called a {\bf CPTP vector}. With this terminology, TOMs are square matrices whose columns are CPTP vectors. Any CP vector $\mathcal{V}$ may be composed with a CP map $\Phi$ on the same space $\mathcal{I}(\mathcal{H})$,
$$
\mathcal{V}\Phi := [\mathcal{V}_i\Phi].
$$
Given CP vectors $\mathcal{V}^k$ and CP maps $\Phi_k$, a CP vector given by a strongly convergent sum
$$
\sum_k \mathcal{V}^k \Phi_k 
$$
will be called a {\bf combination} of $\mathcal{V}^k$ with coefficients $\Phi_k$. We refer to the set of combinations of a given set of CP vectors as the {\bf cone} spanned by such CP vectors, which is closed under sums and products by CP maps, thus also under multiplications by non-negative numbers, but not so by arbitrary complex or real numbers. The {\bf rank} of a set of CP vectors $\mathcal{V}^k$ is the minimal number of CP vectors spanning a cone which contains all $\mathcal{V}^k$.  

A combination of CP vectors may be always rewritten as a combination of the same number of CPTP vectors, a fact that is a direct consequence of the following result.

\begin{lemma} \label{lem:CPTPcomb}
Any CP vector $\mathcal{V}$ may be expressed as $\mathcal{V}=\mathcal{U}\Phi$, with $\mathcal{U}$ a CPTP vector and $\Phi$ a CP map.
\end{lemma}

{\bf Proof.}
Let $\mathcal{V}_i=\sum_kB_{i,k}\cdot B_{i,k}^*$ be the Kraus decompositions of the components of $\mathcal{V}$, which we can assume without loss that run over a common set of indices by enlarging such decompositions with null terms if necessary. Consider the CP map $\boldsymbol{\mathcal{V}}=\sum_i\mathcal{V}_i$ and the positive square root $X=\sqrt{\boldsymbol{\mathcal{V}}^*(I)}$ of the positive operator $\boldsymbol{\mathcal{V}}^*(I)=\sum_{i,k}B_{i,k}^*B_{i,k}$. The identity $X^2=\sum_{i,k}B_{i,k}^*B_{i,k}$ yields $\|X\psi\|^2=\sum_{i,k}\|B_{i,k}\psi\|^2$, which implies that 
\beq\label{eq:kerX}
\ker(X)=\bigcap_{i,k}\ker(B_{i,k}).
\eeq 

Assume first that $\boldsymbol{\mathcal{V}}^*(I)>0$, i.e. $\ker(X)=\{0\}$. Then, there exists $X^{-1}\colon\operatorname{ran}(X)\to\mathcal{H}$, which is densely defined on $\mathcal{H}$ because $\ker(X)=\operatorname{ran}(X)^\bot$. While $X^{-1}$ could be unbounded, the operator $A_{i,k}=B_{i,k}X^{-1}$ is bounded because $\|A_{i,k}X\psi\|^2=\|B_{i,k}\psi\|^2\le\|X\psi\|^2$ for all $\psi\in\mathcal{H}$, so that $\|A_{i,k}\|\le1$. In consequence, $A_{i,k}$ has a bounded extension to $\mathcal{H}$ with the same norm, which we also denote by $A_{i,k}$. Since $B_{i,k}=A_{i,k}X$, we have that $\sum_{i,k\le n}A_{i,k}^*A_{i,k} \le I$ due to the inequality
\beq\label{eq:AB}
\sum_{i,k\le n}\langle X\psi|A_{i,k}^*A_{i,k}X\psi\rangle 
= \sum_{i,k\le n}\|B_{i,k}\psi\|^2 \le \|X\psi\|^2 = \langle X\psi|X\psi\rangle, 
\qquad
\psi\in\mathcal{H},
\eeq
which implies that $\sum_{i,k\le n}\langle\psi|A_{i,k}^*A_{i,k}\psi\rangle \le \langle\psi|\psi\rangle$ for every $\psi\in\operatorname{ran}(X)$ and, by continuity, for  every $\psi\in\mathcal{H}$. Bearing in mind Lemma~\ref{lem:monconv}, this proves that $\mathcal{U}_i = \sum_kA_{i,k}\cdot A_{i,k}^*$ and $\sum_i\mathcal{U}_i$ define CP maps, so that $\mathcal{U}=[\mathcal{U}_i]$ is a CP vector. Indeed, $\mathcal{U}$ is a CPTP vector, as follows from the equality
$$
\sum_{i,k}\langle X\psi|A_{i,k}^*A_{i,k}X\psi\rangle 
= \sum_{i,k}\|B_{i,k}\psi\|^2 = \|X\psi\|^2 = \langle X\psi|X\psi\rangle,
\qquad
\psi\in\mathcal{H}.
$$
This CPTP vector satisfies $\mathcal{V} = \mathcal{U}\Phi$ with $\Phi = X \cdot X$.

Taking into account the previous result, the proposition is proved if we show that in the general case $\mathcal{V}=\widetilde{\mathcal{V}}\widetilde\Phi$ for some CP map $\widetilde\Phi$ and some CP vector $\widetilde{\mathcal{V}}$ with $\widetilde{\boldsymbol{\mathcal{V}}}=\sum_i\widetilde{\mathcal{V}}_i$ satisfying $\widetilde{\boldsymbol{\mathcal{V}}}^*(I)>0$. 
Suppose now that $\boldsymbol{\mathcal{V}}^*(I)$ is not strictly positive, i.e. $\ker(X)$ is non-trivial. If $P$ is the orthogonal projection of $\mathcal{H}$ onto $\ker(X)$ and $Q=I-P$, we have that $XP=0=PX$ and $XQ=X=QX$. Define a CP vector $\widetilde{\mathcal{V}}$ by 
$$
\widetilde{\mathcal{V}}_i = \mathcal{V}_i + \lambda_i \mathbb{P},
\qquad
\mathbb{P} = P \cdot P,
\qquad
\lambda_i\ge0,
\qquad
\sum_i\lambda_i=\lambda>0.
$$
Then, $\widetilde{\boldsymbol{\mathcal{V}}}^*(I) = \boldsymbol{\mathcal{V}}^*(I) + \lambda P = X^2 + \lambda P$ is such that $\ker(\widetilde{\boldsymbol{\mathcal{V}}}^*(I)) = \ker(X) \cap \ker(P) = \{0\}$, i.e. $\widetilde{\boldsymbol{\mathcal{V}}}^*(I)>0$. Also, 
$$
\widetilde{\mathcal{V}}\mathbb{Q} = \mathcal{V}\mathbb{Q},
\qquad
\mathbb{Q}=Q\cdot Q.
$$
From \eqref{eq:kerX} we get $B_{i,k}P=0=PB_{i,k}^*$, $B_{i,k}Q=B_{i,k}$ and $QB_{i,k}^*=B_{i,k}^*$, which yield $\mathcal{V}\mathbb{Q}=\mathcal{V}$. We conclude that $\mathcal{V}=\widetilde{\mathcal{V}}\widetilde\Phi$ with $\widetilde\Phi=\mathbb{Q}$.
\qed

With this result and the terminology above in mind, we can state the characterization of overlapping factorizations of TOMs.

\begin{theorem} \label{thm:carfac}
A TOM $\mathcal{E}=[\mathcal{E}_i^j]_{i,j\in V}$ has a factorization $\mathcal{E} = (\mathcal{E}_L \oplus I_+) (I_- \oplus \mathcal{E}_R)$ which overlaps on the sum of sites labelled by the subset of vertices $V_0\subset V$ iff for some partition $V=V_-\cup V_0\cup V_+$, denoting $V_L=V_-\cup V_0$ and $V_R=V_0\cup V_+$: 
\begin{itemize}
\item[(i)] $\mathcal{E}_i^j=0$ for $i\in V_+$, $j\in V_-$. 
\item[(ii)] The columns of $\mathcal{E}_{LR}:=[\mathcal{E}_i^j]_{i\in V_L,j\in V_R}$ have at most rank $|V_0|$.
\end{itemize}
\end{theorem}

{\bf Proof.}
The factorization \eqref{eq:fac} implies (i) because the $(i,k)$ coefficients of the left factor vanish for $i\in V_+$, $k\in V_L$, while the $(k,j)$ coefficients of the right factor are null for $k\in V_R$, $j\in V_-$. On the other hand, (ii) follows from the fact that the $(i,k)$ coefficients of the left factor vanish for $i\in V_L$, $k\in V_+$, while the $(k,j)$ coefficients of the right factor are null for $k\in V_-$, $j\in V_R$, leading to 
$$
\mathcal{E}_i^j=\sum_{k\in V_0}(\mathcal{E}_L)_i^k(\mathcal{E}_R)_k^j,
\qquad
i\in V_L, \quad j\in V_R.
$$
Denoting by $\mathcal{U}^k:=[(\mathcal{E}_L)_i^k]_{i\in V_L}$, $k\in V_L$, and $\mathcal{V}^j:=[\mathcal{E}_i^j]_{i\in V_L}$, $j\in V_R$, the columns of $\mathcal{E}_L$ and $\mathcal{E}_{LR}$ respectively, the above identity reads as
$$
\mathcal{V}^j = \sum_{k\in V_0} \mathcal{U}^k (\mathcal{E}_R)_k^j,
$$
which proves (ii).

To see the converse, assume (i) and (ii). The second condition means that the columns $\mathcal{V}^j$ of $\mathcal{E}_{LR}$ may be expressed as
$$
\mathcal{V}^j = \sum_{k\in V_0} \mathcal{U}^k \Phi_k^j,
\qquad
j\in V_R,
$$
for some CP vectors $\mathcal{U}^k$ and CP maps $\Phi_k^j$. Bearing in mind Lemma~\ref{lem:CPTPcomb}, we can assume without loss that $\mathcal{U}^k$ are CPTP maps. Let us introduce the left and right operators
$$
\begin{aligned}
& \mathcal{E}_L = 
\sum_{i\in V_L,k\in V_-} 
\mathcal{E}_i^k \otimes |i\rangle\langle k|\cdot|k\rangle\langle i| +
\sum_{i\in V_L,k\in V_0} 
\mathcal{U}^k_i \otimes |i\rangle\langle k|\cdot|k\rangle\langle i|,
\\
& \mathcal{E}_R = 
\sum_{k\in V_0,j\in V_R} 
\Phi_k^j \otimes |k\rangle\langle j|\cdot|j\rangle\langle k| +
\sum_{k\in V_+,j\in V_R} 
\mathcal{E}_k^j \otimes |k\rangle\langle j|\cdot|j\rangle\langle k|.
\end{aligned}
$$
Both are TOMs because $\mathcal{U}^k$ are CPTP vectors, $\sum_{i\in V_L}\mathcal{E}_i^j=\sum_{i\in V}\mathcal{E}_i^j$ for $j\in V_-$, while 
$\sum_{k\in V_0}\Phi_k^j + \sum_{k\in V_+}\mathcal{E}_k^j$ is trace preserving as a consequence of the same property of $\sum_{i\in V_L} \mathcal{U}^k_i$ and
$$
\sum_{i\in V} \mathcal{E}_i^j = 
\sum_{i\in V_L} \mathcal{V}^j_i + \sum_{i\in V_+} \mathcal{E}_i^j = 
\sum_{k\in V_0} \bigg(\sum_{i\in V_L} \mathcal{U}^k_i\bigg) \Phi_k^j +
\sum_{i\in V_+} \mathcal{E}_ij,
\qquad
j \in V_R.
$$
The TOMs $\mathcal{E}_{L,R}$ yield an overlapping factorization \eqref{eq:fac} of $\mathcal{E}$, as follows from
$$
\begin{aligned}
(\mathcal{E}_L \oplus I_+)(I_- \oplus \mathcal{E}_R) 
& = 
\sum_{i\in V_L,k\in V_-} 
\mathcal{E}_i^k \otimes |i\rangle\langle k|\cdot|k\rangle\langle i| +
\sum_{k\in V_+,j\in V_R} 
\mathcal{E}_k^j \otimes |k\rangle\langle j|\cdot|j\rangle\langle k| 
\\ 
& +
\sum_{i\in V_L,k\in V_0,j\in V_R} 
\mathcal{U}^k_i\Phi_k^j \otimes |i\rangle\langle j|\cdot|j\rangle\langle i|,
\end{aligned}
$$
and the fact that $\sum_{k\in V_0}\mathcal{U}^k_i\Phi_k^j=\mathcal{V}^j_i=\mathcal{E}_i^j$ for $i\in V_L$ and $j\in V_R$.
\qed

The condition (i) in the above theorem has a simple diagrammatic depiction summarized by the following transition structure.    
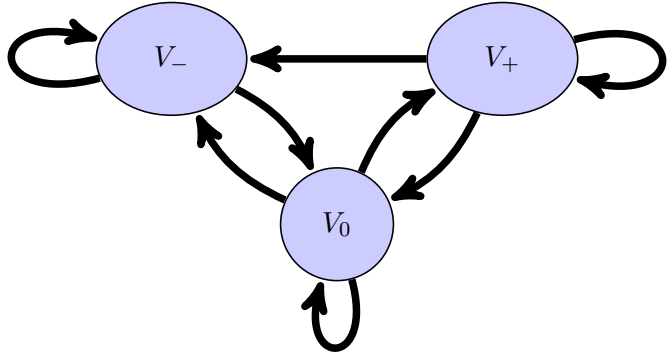
\begin{center}
\begin{tikzpicture}[->,>=stealth',shorten >=1pt,auto,node distance=1.7cm,
                    semithick]
  \tikzstyle{every state}=[ellipse,minimum height=1.5cm, 
    		minimum width=2cm,fill=blue!20,draw,inner sep=0pt,
			node distance=5cm]
  
  \node[state] (L)  			{$V_-$};
  \node[state] (R)  at (4.4,0) 	{$V_+$};
  \node[state,minimum height=1.5cm,minimum width=1.5cm] (0)  
  								at (2.2,-2.2)  {$V_0$};
  
  \path
        (0)  edge   [loop below,line width=1mm]   	node {} (0)
        (L)  edge   [loop left,line width=1mm]   	node {} (L)
        (R)  edge   [loop right,line width=1mm]   	node {} (R)
        (0)  edge	[bend left=20,line width=1mm] 	node {} (L)
        (0)  edge   [bend left=20,line width=1mm] 	node {} (R)
        (L)  edge   [bend left=20,line width=1mm] 	node {} (0)
        (R)  edge   [bend left=20,line width=1mm] 	node {} (0)
        (R)  edge   [line width=1mm] 				node {} (L);
        
\end{tikzpicture}
\end{center}
In other words, (i) states that, for a certain subset of vertices, the remaining ones split into two subsets with no transitions from one of them to the other. On the other hand, the CP vectors affected by the rank condition (ii) are those alluding to the transitions represented in the above diagram by the {\color{black}{four arrows}} which go from the right to the left (arrows from $V_+$ to $V_0$ and to $V_-$, from $V_0$ to $V_-$ and the loop on $V_0$).

\bex
Consider the TOM
$$
\begin{gathered}
\mathcal{E} = 
\begin{bmatrix} 
\mathcal{E}_1^1 & \mathcal{E}_1^2 & \mathcal{E}_1^3 \\[2pt] 
\mathcal{E}_2^1 & \mathcal{E}_2^2 & \mathcal{E}_2^3 \\[2pt] 
0 & \mathcal{E}_3^2 & \mathcal{E}_3^3
\end{bmatrix},
\qquad
\mathcal{E}_i^j = 
\begin{cases}
B^1_{i,1}\cdot B^{1*}_{i,1} + B^1_{i,2}\cdot B^{1*}_{i,2}, 
\quad & i=1,2,
\\[2pt]
B_i^j\cdot B_i^{j*}, 
& j=2,3,
\end{cases} 
\\
B_{1,1}^1=\sqrt{\frac{5}{8}}I_2,
\quad 
B_{1,2}^1=\sqrt{\frac{1}{8}}\begin{bmatrix}0&1\\1&0\end{bmatrix},
\quad
B_{2,1}^1=\sqrt{\frac{1}{8}}\begin{bmatrix}0&-i\\i&0\end{bmatrix},
\quad
B_{2,2}^1=\sqrt{\frac{1}{8}}\begin{bmatrix}1&0\\0&-1\end{bmatrix},
\\
B_2^1=\frac{1}{\sqrt{6}}\begin{bmatrix}0&0\\1&0\end{bmatrix},
\quad
B_2^2=\begin{bmatrix}\sqrt{\frac{1}{6}}&\sqrt{\frac{2}{3}}\\0&0\end{bmatrix},
\quad
B_2^3=\frac{1}{\sqrt{3}}\begin{bmatrix}1&0\\-1&1\end{bmatrix},
\\
B_3^1=\begin{bmatrix}0&0\\\sqrt{\frac{2}{3}}&-\sqrt{\frac{1}{6}}\end{bmatrix},
\quad
B_3^2=\frac{1}{\sqrt{6}}\begin{bmatrix}0&1\\0&0\end{bmatrix},
\quad
B_3^3=\frac{1}{\sqrt{3}}\begin{bmatrix}1&1\\0&1\end{bmatrix}.
\end{gathered}
$$
The fact that $\mathcal{E}_3^1=0$ points to the possible existence of a factorization of $\mathcal{E}$ overlapping on site 2, corresponding to the partition of the vertices $V=\{1,2,3\}$ given by $V_-=\{1\}$, $V_0=\{2\}$, $V_+=\{3\}$. The left and right TOMs must be associated with graphs with vertices $V_L=\{1,2\}$ and $V_R=\{2,3\}$ respectively, so that the overlapping factorization should look like
\beq \label{eq:ex3-of}
\mathcal{E} = 
\left[\begin{array}{ccc|c}
 & & & \\[-4pt]
 & \mathcal{E}_L & & \\
 & & & \\[-4pt]
\hline
 & & & \\[-10pt]
 & & & I_+ \\[-9pt]
 & & &
\end{array}\right] 
\left[\begin{array}{c|ccc}
 & & & \\[-9pt]
 I_- & & & \\
 & & & \\[-10pt]
 \hline
 & & & \\[-2pt]
 & & \mathcal{E}_R & \\[8pt]
\end{array}\right] = 
\left[\begin{array}{cc|c}
 & & \\[-11pt]
 \mathcal{E}_1^1 & (\mathcal{E}_L)_1^2 & 0 \\[4pt]
 \mathcal{E}_2^1 & (\mathcal{E}_L)_2^2 & 0 \\[4pt]
\hline
 & & \\[-10pt]
 0 & 0 & I_+ \\[-10pt] 
 & & 
\end{array}\right] 
\left[\begin{array}{c|cc}
 & & \\[-10pt]
 I_- & 0 & 0 \\[2pt]
 \hline
 & & \\[-9pt]
 0 & (\mathcal{E}_R)_2^2 & (\mathcal{E}_R)_2^3 \\[4pt]
 0 & \mathcal{E}_3^2 & \mathcal{E}_3^3 \\[-12pt]
 & &
\end{array}\right]. 
\eeq
According to Theorem \ref{thm:carfac}, to analyze the existence of such a factorization only requires to inspect the rank of the columns of the upper right block
$$
\mathcal{E}_{LR}=\begin{bmatrix} \mathcal{E}_1^2 & \mathcal{E}_1^3 
\\[2pt] 
\mathcal{E}_2^2 & \mathcal{E}_2^3
\end{bmatrix}.
$$
Let us show that such rank equals 1 by applying the reasoning of Lemma \ref{lem:CPTPcomb}. Denote the columns of $\mathcal{E}_{LR}$ by
$$
\mathcal{V}_\alpha = 
\begin{bmatrix} 
\mathcal{E}_1^2 \\[2pt] \mathcal{E}_2^2 
\end{bmatrix} =
\begin{bmatrix} 
B_1^2\cdot {B_1^2}^* \\[2pt] B_2^2\cdot {B_2^2}^* \end{bmatrix},
\qquad
\mathcal{V}_\beta = 
\begin{bmatrix} 
\mathcal{E}_1^3 \\[2pt] \mathcal{E}_2^3 \end{bmatrix} =
\begin{bmatrix} B_1^3\cdot {B_1^3}^* \\[2pt] B_2^3\cdot {B_2^3}^* \end{bmatrix},
$$
and consider the CP maps given by the sum of their components,
$\boldsymbol{\mathcal{V}}_\alpha=\mathcal{E}_1^2+\mathcal{E}_2^2$ and
$\boldsymbol{\mathcal{V}}_\beta=\mathcal{E}_1^3+\mathcal{E}_2^3$.
The positive square roots
$$
X_\alpha = \sqrt{\boldsymbol{\mathcal{V}}_\alpha^*(I_2)} =
\sqrt{{B_1^2}^*B_1^2+{B_2^2}^*B_2^2} = 
\frac{1}{\sqrt{15}} \begin{bmatrix} 2 & 1 \\ 1 & 3 \end{bmatrix},
\quad
X_\beta = \sqrt{\boldsymbol{\mathcal{V}}_\beta^*(I_2)} ={\color{black}
\sqrt{{B_1^3}^*B_1^3+{B_2^3}^*B_2^3} }= 
\frac{1}{\sqrt{15}} \begin{bmatrix} 3 & \kern-5pt -1 \\ -1 & \kern-5pt 2 \end{bmatrix},
$$
allow us to define $A_i^2 = B_i^2 X_\alpha^{-1}$ and $A_i^3 = B_i^3 X_\beta^{-1}$ for $i=1,2$. It turns out that
$$
A_1^2 = A_1^3 =
\frac{1}{\sqrt{10}} \begin{bmatrix} 0 & 0 \\ 3 & -1 \end{bmatrix}
=: A_1,
\qquad
A_2^2 = {\color{black}A_2^3} =
\frac{1}{\sqrt{10}} \begin{bmatrix} 1 & 3 \\ 0 & 0 \end{bmatrix}
=: A_2.
$$
Therefore, $\mathcal{V}_\alpha=\mathcal{U}\Phi_\alpha$ and $\mathcal{V}_\beta=\mathcal{U}\Phi_\beta$ with
$$
\mathcal{U} = 
\begin{bmatrix} \mathcal{U}_1 \\ \mathcal{U}_2 \end{bmatrix} =
\begin{bmatrix} A_1\cdot A_1^* \\ A_2\cdot A_2^* \end{bmatrix},
\qquad 
\Phi_\alpha = X_\alpha\cdot X_\alpha,
\qquad
\Phi_\beta = X_\beta\cdot X_\beta,
$$
and $A_1^*A_1+A_2^*A_2=I_2$, so that $\mathcal{U}$ is a CPTP vector. Hence, the columns of $\mathcal{E}_{LR}$ have rank 1, which in view of Theorem~\ref{thm:carfac} proves that the overlapping factorization \eqref{eq:ex3-of} holds. Theorem~\ref{thm:carfac} also shows that such a factorization follows by setting $(\mathcal{E}_L)_1^2=\mathcal{U}_1$, $(\mathcal{E}_L)_2^2=\mathcal{U}_2$, $(\mathcal{E}_R)_2^2=\Phi_\alpha$ and $(\mathcal{E}_R)_2^3=\Phi_\beta$ in \eqref{eq:ex3-of}, i.e.
$$
\mathcal{E}_L = 
\begin{bmatrix} 
\mathcal{E}_1^1 & \mathcal{U}_1 
\\[2pt]
\mathcal{E}_2^1 & \mathcal{U}_2 
\end{bmatrix},
\qquad
\mathcal{E}_R = 
\begin{bmatrix} 
\Phi_\alpha & \Phi_\beta 
\\[2pt]
\mathcal{E}_3^2 & \mathcal{E}_3^3
\end{bmatrix}.
$$
In terms of block matrix representations, this overlapping factorization reads as
$$
\begin{aligned}
\widehat{\mathcal{E}} & = 
\left[\begin{array}{cccc|cccc|cccc}
\frac{5}{8} & 0 & 0 & \frac{1}{8} & 0 & 0 & 0 & 0 & 0 & 0 & 0 & 0 \\[2pt]
0 & \frac{5}{8} & \frac{1}{8} & 0 & 0 & 0 & 0 & 0 & 0 & 0 & 0 & 0 \\[2pt]
0 & \frac{1}{8} & \frac{5}{8} & 0 & 0 & 0 & 0 & 0 & 0 & 0 & 0 & 0 \\[2pt]
\frac{1}{8} & 0 & 0 & \frac{5}{8} & \frac{1}{6} & 0 & 0 & 0 & 
\frac{2}{3} & \frac{-1}{3} & \frac{-1}{3} & \frac{1}{6} \\[2pt]
\hline 
& & & & & & & & & & & \\[-11pt]
\frac{1}{8} & 0 & 0 & \frac{1}{8} & 
\frac{1}{6} & \frac{1}{3} & \frac{1}{3} & \frac{2}{3} & 
0 & 0 & 0 & \frac{1}{6} \\[2pt]
0 & \frac{-1}{8} & \frac{-1}{8} & 0 & 0 & 0 & 0 & 0 & 0 & 0 & 0 & 0 \\[2pt]
0 & \frac{-1}{8} & \frac{-1}{8} & 0 & 0 & 0 & 0 & 0 & 0 & 0 & 0 & 0 \\[2pt]
\frac{1}{8} & 0 & 0 & \frac{1}{8} & 0 & 0 & 0 & 0 & 0 & 0 & 0 & 0 \\[2pt]
\hline
& & & & & & & & & & & \\[-11pt]
& & & & \frac{1}{3} & 0 & 0 & 0 & 
\frac{1}{3} & \frac{1}{3} & \frac{1}{3} & \frac{1}{3} \\[2pt]
& & & & \frac{-1}{3} & \frac{1}{3} & 0 & 0 & 
0 & \frac{1}{3} & 0 & \frac{1}{3} \\[2pt]
& & & & \frac{-1}{3} & 0 & \frac{1}{3} & 0 & 
0 & 0 & \frac{1}{3} & \frac{1}{3} \\[2pt]
& & & & \frac{1}{3} & \frac{-1}{3} & \frac{-1}{3} & \frac{1}{3} & 
0 & 0 & 0 & \frac{1}{3}
\end{array}\right]
\\
& =
\left[\begin{array}{cccc|cccc|cccc}
\frac{5}{8} & 0 & 0 & \frac{1}{8} & 0 & 0 & 0 & 0 \\[2pt]
0 & \frac{5}{8} & \frac{1}{8} & 0 & 0 & 0 & 0 & 0 \\[2pt]
0 & \frac{1}{8} & \frac{5}{8} & 0 & 0 & 0 & 0 & 0 \\[2pt]
\frac{1}{8} & 0 & 0 & \frac{5}{8} & \frac{9}{10} & \frac{-3}{10} & 
\frac{-3}{10} & \frac{1}{10} \\[2pt]
\hline 
& & & & & & & & & & & \\[-11pt]
\frac{1}{8} & 0 & 0 & \frac{1}{8} & \frac{1}{10} & \frac{3}{10} & 
\frac{3}{10} & \frac{9}{10} \\[2pt]
0 & \frac{-1}{8} & \frac{-1}{8} & 0 & 0 & 0 & 0 & 0 \\[2pt]
0 & \frac{-1}{8} & \frac{-1}{8} & 0 & 0 & 0 & 0 & 0 \\[2pt]
\frac{1}{8} & 0 & 0 & \frac{1}{8} & 0 & 0 & 0 & 0 \\[2pt]
\hline
& & & & & & & & & & & \\[-11pt]
& & & & & & & & 1 & & & \\[2pt]
& & & & & & & & & 1 & & \\[2pt]
& & & & & & & & & & 1 & \\[2pt]
& & & & & & & & & & & 1
\end{array}\right]
\kern-2pt
\left[\begin{array}{cccc|cccc|cccc}
1 & & & & & & & & & & & \\[2pt]
& 1 & & & & & & & & & & \\[2pt]
& & 1 & & & & & & & & & \\[2pt]
& & & 1 & & & & & & & & \\[2pt]
\hline 
& & & & & & & & & & & \\[-11pt]
& & & & \frac{4}{15} & \frac{2}{15} & \frac{2}{15} & \frac{1}{15} & 
\frac{3}{5} & \frac{-1}{5} & \frac{-1}{5} & \frac{1}{15} \\[2pt]
& & & & \frac{2}{15} & \frac{2}{5} & \frac{1}{15} & \frac{1}{5} & 
\frac{-1}{5} & \frac{2}{5} & \frac{1}{15} & \frac{-2}{15} \\[2pt]
& & & & \frac{2}{15} & \frac{1}{15} & \frac{2}{5} & \frac{1}{5} & 
\frac{-1}{5} & \frac{1}{15} & \frac{2}{5} & \frac{-2}{15} \\[2pt]
& & & & \frac{1}{15} & \frac{1}{5} & \frac{1}{5} & \frac{3}{5} & 
\frac{1}{15} & \frac{-2}{15} & \frac{-2}{15} & \frac{4}{15} \\[2pt]
\hline
& & & & & & & & & & & \\[-11pt]
& & & & \frac{1}{3} & 0 & 0 & 0 & 
\frac{1}{3} & \frac{1}{3} & \frac{1}{3} & \frac{1}{3} \\[2pt]
& & & & \frac{-1}{3} & \frac{1}{3} & 0 & 0 & 
0 & \frac{1}{3} & 0 & \frac{1}{3} \\[2pt]
& & & & \frac{-1}{3} & 0 & \frac{1}{3} & 0 & 
0 & 0 & \frac{1}{3} & \frac{1}{3} \\[2pt]
& & & & \frac{1}{3} & \frac{-1}{3} & \frac{-1}{3} & \frac{1}{3} & 
0 & 0 & 0 & \frac{1}{3}
\end{array}\right]\kern-1pt,
\end{aligned}
$$
giving the following matrix representations of the Schur functions $f$, $f_{L/R}$ for site 2 with respect to $\mathcal{E}$, $\mathcal{E}_{L/R}$,
$$
\begin{aligned}
& \widehat{f}(z) = 
\begin{pmatrix}
\frac{8z^4-29z^3+62z^2-93z+72}{4(z-3)^3(3z-4)} & 
\frac{z^4-46z^3+181z^2-276z+144}{4(z-3)^3(3z-4)} & 
\frac{z^4-46z^3+181z^2-276z+144}{4(z-3)^3(3z-4)} & 
\frac{9z^4-70z^3+281z^2-480z+288}{4(z-3)^3(3z-4)}
\\ 
0 & 0 & 0 & 0
\\ 
0 & 0 & 0 & 0
\\ 
\frac{z(8z^3-23z^2+18z+9)}{4(z-3)^3(3z-4)} & 
-\frac{z^2(5z^2-10z+9)}{4(z-3)^3(3z-4)} & 
-\frac{z^2(5z^2-10z+9)}{4(z-3)^3(3z-4)} & 
\frac{z^2(3z^2-2z+3)}{4(z-3)^3(3z-4)}
\end{pmatrix},
\\
& \widehat{f}_L(z) = 
\begin{pmatrix}
-\frac{3z+8}{20(3z-4)} & \frac{3(7z-8)}{20(3z-4)} & 
\frac{3(7z-8)}{20(3z-4)} & \frac{53z-72}{20(3z-4)}
\\ 
0 & 0 & 0 & 0
\\ 
0 & 0 & 0 & 0
\\ 
-\frac{9z}{20(3z-4)} & \frac{3z}{20(3z-4)} & 
\frac{3z}{20(3z-4)} & -\frac{z}{20(3z-4)}
\end{pmatrix},
\\
& \widehat{f}_R(z) =
\begin{pmatrix}
-\frac{15z^3-35z^2+12z+36}{5(z-3)^3} & \frac{2(5z^3-10z^2+15z-9)}{5(z-3)^3} & 
\frac{2(5z^3-10z^2+15z-9)}{5(z-3)^3} & -\frac{5z^3+4z^2-6z+9}{5(z-3)^3}
\\[3pt] 
\frac{2(5z^3-20z^2+27z-9)}{5(z-3)^3} & -\frac{5z^3-5z^2-30z+54}{5(z-3)^3} & 
-\frac{5z^3-10z^2+9}{5(z-3)^3} & \frac{5z^3-17z^2+33z-27}{5(z-3)^3}
\\[3pt] 
\frac{2(5z^3-20z^2+27z-9)}{5(z-3)^3} & -\frac{5z^3-10z^2+9}{5(z-3)^3} & 
-\frac{5z^3-5z^2-30z+54}{5(z-3)^3} & \frac{5z^3-17z^2+33z-27}{5(z-3)^3}
\\[3pt] 
-\frac{5z^3-20z^2+18z+9}{5(z-3)^3} & \frac{5z^3-25z^2+45z-27}{5(z-3)^3} & 
\frac{5z^3-25z^2+45z-27}{5(z-3)^3} & -\frac{16z^2-69z+81}{5(z-3)^3}
\end{pmatrix}.
\end{aligned}
$$
The factorization $\widehat{f}=\widehat{f}_L\widehat{f}_R$, which follows from \eqref{eq:facdec}, may be patiently checked.

All the TOMs involved are irreducible since they have a unique invariant state which is faithful. Hence, site 2 is positive recurrent for all of them. Concerning the TOM $\mathcal{E}$, the expected return time to site 2 for an arbitrary state $\rho$ supported on it is given by
\beq \label{eq:ex3-t}
\tau(\rho\to|2\rangle) = 1 + \operatorname{Tr}(f'(1)(\rho)) = 
1+({\widehat{f}}{\,'}(1)\,vec(\rho))_1+({\widehat{f}}{\,'}(1)\,vec(\rho))_4 =
3+2\rho_{11}-\frac{7}{2}\operatorname{Re}(\rho_{12}).
\eeq
A similar calculation for the left/right TOMs $\mathcal{E}_{L/R}$ yields
\beq \label{eq:ex3-tLR}
\tau_L(\rho\to|2\rangle) =
\frac{1}{5}(7+16\rho_{11}-12\operatorname{Re}(\rho_{12})),
\qquad
\tau_R(\rho\to|2\rangle) =
2-\frac{3}{2}\operatorname{Re}(\rho_{12}).
\eeq
The expected return time for the whole TOM may be recovered from those of the left and right ones via \eqref{eq:tfac}. To check this, note that the trace preserving map $f_R(1)$ transforms the state $\rho$ into a state $\sigma$ concentrated on site 2 given by $vec(\sigma) = \widehat{f}_R(1) \, vec(\rho)$, thus determined by
$$
\sigma_{11} = \frac{1}{10}(3+4\rho_{11}-\operatorname{Re}(\rho_{12})),
\qquad
\sigma_{12} = \frac{1}{20}(3-6\rho_{11}+12\rho_{12}+2\overline{\rho_{12}}).
$$
Combined with \eqref{eq:ex3-tLR}, this yields
$\tau_L(\sigma\to|2\rangle) = 2(1+\rho_{11}-\operatorname{Re}(\rho_{12}))$. Inserting these results into the identity $\tau(\rho\to|2\rangle)=\tau_L(\sigma\to|2\rangle)+\tau_R(\rho\to|2\rangle)-1$ arising from \eqref{eq:tfac} leads again to \eqref{eq:ex3-t}.
\eex
\qee

\section{Nearest neighbour TOMs in 1D}
\label{sec:1D}

Although the previous theoretical results are valid even for infinite-dimensional CPTP maps and TOMs, the examples illustrating them have been {\color{black}limited, until now, to} the finite-dimensional setting. In this section we will discuss a special class of infinite-dimensional TOMs where the Schur techniques are particularly useful and reveal striking results: TOMs on the line ($V=\mathbb{Z}$) and the half-line ($V=\{i\in\mathbb{Z}:i\ge0\}$) with only nearest neighbour transitions, i.e. given by doubly infinite or semi-infinite block tridiagonal matrix representations. 

In what follows we will consider TOMs $\mathcal{E}$ with a Hilbert space $\mathcal{H}$ of internal degrees of freedom and an infinite-dimensional space of sites which will be either $\mathcal{S}=\operatorname{span}\{|i\rangle\}_{i\in\mathbb{Z}}$ (line) or $\mathcal{S}=\operatorname{span}\{|i\rangle\}_{i\ge0}$ (half-line). For convenience, we will distinguish the identity $I$ on $\mathcal{I}(\mathcal{H})$ from the direct sum of infinitely many copies of $I$, denoted by $\mathbb{I}$, while $\mathbb{I}_n$ will stand for the sum of a finite number $n$ of copies of $I$.

\subsection{Nearest neighbour TOMs on the half-line}
\label{ssec:HALF-LINE}

Consider a TOM $\mathcal{E}$ with a Hilbert space $\mathcal{H}$ of internal degrees of freedom and a site space $\mathcal{S}=\operatorname{span}\{|i\rangle\}_{i\ge0}$. We will assume that $\mathcal{E}$ has the tridiagonal shape
\beq \label{eq:3d-HL}
\mathcal{E} = 
\begin{bmatrix} 
\mathcal{E}_0^0 & \mathcal{E}_0^1 & \\[2pt] 
\mathcal{E}_1^0 & \mathcal{E}_1^1 & \mathcal{E}_1^2 \\[2pt]  
& \mathcal{E}_2^1 & \mathcal{E}_2^2 & \mathcal{E}_2^3 \\[-2pt]
& & \ddots & \ddots & \ddots
\end{bmatrix},
\qquad
\mathcal{E}^{i\pm1}_i \; \text{invertible}.
\eeq
The recurrence properties of site $i$ are codified by the Schur function
\beq \label{eq:HL-SCHUR}
f^{(i)}(z) = \mathbb{P}_i \mathcal{E} 
(\mathbb{I}-z\mathbb{Q}_i\mathcal{E})^{-1} \mathbb{P}_i,
\eeq
where $\mathbb{Q}_i=\mathbb{I}-\mathbb{P}_i$ and $\mathbb{P}_i=P_i\cdot P_i$ is the projection of $\mathcal{I}_\mathcal{S}(\mathcal{H})$ onto $\mathcal{I}(\mathcal{H})\otimes|i\rangle\langle i|$ given by the orthogonal projection $P_i$ of $\mathcal{H}\otimes\mathcal{S}$ onto $\mathcal{H}\otimes|i\rangle$. A general relation for FR-functions --which generalizes the so-called renewal equation for random walks-- implies that \cite[Theorem~2.5]{gvfr}
\beq \label{eq:ST-SC}
\mathbb{P}_i(\mathbb{I}-z\mathcal{E})^{-1}\mathbb{P}_i = (I-zf^{(i)}(z))^{-1},
\qquad |z|<1,
\eeq
{\color{black} a relation which will} be useful to identify $f^{(i)}$.

{\color{black}We note that} the matrix calculations carried out in the finite-dimensional case cannot be extended to this situation. Nevertheless, other useful techniques are available. Let us look for the moment at the Schur function \eqref{eq:HL-SCHUR} for site 0, i.e. $f=f^{(0)}$. Consider the block decompositions
$$
\mathcal{E} = 
\left[\begin{array}{c|cccc} 
\mathcal{E}_0^0 & \mathcal{E}_0^1 & \\[2pt]
\hline 
\\[-10pt]
\mathcal{E}_1^0 & \mathcal{E}_1^1 & \mathcal{E}_1^2 \\[2pt]  
& \mathcal{E}_2^1 & \mathcal{E}_2^2 & \mathcal{E}_2^3 \\[-2pt]
& & \ddots & \ddots & \ddots
\end{array}\right] =
\left[\begin{array}{c|cccc} 
\mathcal{E}_0^0 & \mathcal{E}_0^1 & \\[2pt]
\hline 
\\[-10pt]
\mathcal{E}_1^0 & & \\  
& & \mathcal{E}_1 & \\[-2pt]
& & & &
\end{array}\right],
\qquad 
\mathbb{Q}_0 = 
\left[\begin{array}{c|ccc}
\kern5pt & & & \\[2pt]
\hline 
& & & \\[-11pt]
& & & \\  
& & \kern3pt \mathbb{I} & \\[-2pt]
& & &
\end{array}\right],
$$
where $\mathcal{E}_1=\mathbb{Q}_0\mathcal{E}\mathbb{Q}_0$ is trace non-increasing on $\mathcal{I}_{\mathcal{S}_1}(\mathcal{H})$, with 
$\mathcal{S}_1=\operatorname{span}\{|i\rangle\}_{i\ge1}$. 
They lead to the block structures
$$
\mathbb{I}-z\mathbb{Q}_0\mathcal{E} = 
\left[\begin{array}{c|cccc}
I & & \\[2pt]
\hline 
\\[-10pt]
-z\mathcal{E}_1^0 & & \\  
& & \mathbb{I}-z\mathcal{E}_1 \kern-14pt & \\[-2pt]
& & & &
\end{array}\right],
\qquad
(\mathbb{I}-z\mathbb{Q}_0\mathcal{E})^{-1} = 
\left[\begin{array}{c|cccc}
I \kern2pt & & \\[2pt]
\hline 
\\[-11pt]
& & \\  
& & \kern-5pt (\mathbb{I}-z\mathcal{E}_1)^{-1} \kern-20pt & \\[-2pt]
& & & &
\end{array}\right]
\left[\begin{array}{c|cccc}
I & & \\[2pt]
\hline 
\\[-10pt]
z\mathcal{E}_1^0 & & \\  
& & \mathbb{I} \kern-12pt & \\[-2pt]
& & & &
\end{array}\right].
$$
Since $f(z)$ is precisely the upper-left block of $\mathcal{E} (\mathbb{I}-z\mathbb{Q}_0\mathcal{E})^{-1}$, the previous calculations yield 
$$
f(z) = \mathcal{E}_0^0 + z \mathcal{E}_0^1 
\mathbb{P}_1 (\mathbb{I}-z\mathcal{E}_1)^{-1} \mathbb{P}_1 \mathcal{E}_1^0.
$$
A relation similar to \eqref{eq:ST-SC} shows that 
$\mathbb{P}_1 (\mathbb{I}-z\mathcal{E}_1)^{-1} \mathbb{P}_1 = (I-zf_1(z))^{-1}$ for $|z|<1$, where $f_1$ is the Schur function of site 1 with respect to the trace non-increasing map $\mathcal{E}_1$ (see Remark~\ref{SCHUR-GEN}). Therefore,
\beq \label{eq:f-f1}
f(z) = \mathcal{E}_0^0 + 
z \mathcal{E}_0^1 (I-zf_1(z))^{-1} \mathcal{E}_1^0,
\eeq
which may be solved {\color{black}for} $f_1$ to give
\beq \label{eq:f1-f}
f_1(z) = z^{-1}I - 
\mathcal{E}_1^0(f(z)-\mathcal{E}_0^0)^{-1}\mathcal{E}_0^1.
\eeq
The iteration of the relation $f_0=f \mapsto f_1$ extends it to $f_i \mapsto f_{i+1}$, where $f_i$ is the Schur function for site $i$ with respect to the operator $\mathcal{E}_i$ obtained by removing the first $i$ block rows and columns in $\mathcal{E}$. We will refer to the Schur functions $f_i$ as the iterates of $f$ with respect to $\mathcal{E}$. In some situations, the relations among the iterates $f_i$ may help in the analysis of the Schur function $f$ (see the example below).

Regarding the Schur function $f^{(i)}$ of site $i$ for the original operator $\mathcal{E}$, we can resort to using the splitting rules \eqref{eq:facdec} to relate them to the iterates $f_i$ of $f=f^{(0)}$. Consider the overlapping decomposition given by
$$
\mathcal{E} = 
\left[\begin{array}{c|c}
& \\
\kern10pt \mathcal{E}_i^- \kern10pt & \\
& \\
\hline
& \\
& \kern10pt 0 \kern7pt \\
& 
\end{array}\right] +
\left[\begin{array}{c|c}
& \\[-5pt]
\kern7pt 0 \kern9pt & \\[-5pt]
& \\
\hline
& \\[5pt]
& \kern18pt \mathcal{E}_i \kern14pt \\[5pt]
&  
\end{array}\right],
\qquad
\mathcal{E}_i^- =
\begin{bmatrix}
\mathcal{E}_0^0 & \kern7pt \mathcal{E}_0^1 & & & 
\\[2pt]
\mathcal{E}_1^0 & \kern7pt \mathcal{E}_1^1 & \kern7pt \mathcal{E}_1^2 & & 
\\[-2pt]
& \kern-15pt \ddots & \kern-10pt \ddots & \ddots & 
\\
& & \kern-25pt \mathcal{E}_{i-1}^{i-2} & \kern-10pt \mathcal{E}_{i-1}^{i-1} & \mathcal{E}_{i-1}^{i} 
\\[2pt]
& & & \kern-10pt \mathcal{E}_{i}^{i-1} & 0
\end{bmatrix}, 
$$
which has site $i$ as overlapping subspace. According to Remark~\ref{rem:SR-GEN}, 
$$
f^{(i)} = f_i^- + f_i, 
$$
where $f_i$ is the $i$-th iterate of $f$ and $f_i^-$ is the Schur function of site $i$ with respect to the trace non-increasing map $\mathcal{E}_i^-$. Since $\mathcal{E}_i^-$ lives on a finite graph, $f_i^-$ may be obtained by standard matrix calculations when $\dim\mathcal{H}<\infty$.

The above overlapping decomposition is not of the type analyzed in Theorem~\ref{thm:dec} since neither $\mathcal{E}_i$ nor $\mathcal{E}_i^-$ are  TOMs. If we wish to use the results of that theorem we should refer instead to a decomposition such as $\mathcal{E} = (\mathcal{E}_L^{(i)}\oplus\mathbb{I}) +(\mathbb{I}_i\oplus\mathcal{E}_R^{(i)})-(\mathbb{I}_i\oplus\mathcal{E}_0^{(i)}\oplus\mathbb{I})$, with $\mathcal{E}_0^{(i)}=\mathcal{E}_{i-1}^{i}+\mathcal{E}_{i}^{i}+\mathcal{E}_{i+1}^{i}$ and
\beq \label{eq:od-nn}
\mathcal{E}_L^{(i)} = 
\begin{bmatrix}
\mathcal{E}_0^0 & \kern10pt \mathcal{E}_0^1 & & & 
\\[2pt]
\mathcal{E}_1^0 & \kern10pt \mathcal{E}_1^1 & \kern10pt \mathcal{E}_1^2 & & 
\\[-2pt]
& \kern-15pt \ddots & \kern-5pt \ddots & \kern5pt \ddots & 
\\
& & \kern-35pt \mathcal{E}_{i-1}^{i-2} & \kern-15pt \mathcal{E}_{i-1}^{i-1} &  \kern-12pt \mathcal{E}_{i-1}^{i} 
\\[2pt]
& & & \kern-15pt \mathcal{E}_{i}^{i-1} & \kern-7pt \mathcal{E}_{i}^{i}\!+\mathcal{E}_{i+1}^{i}
\end{bmatrix},
\qquad
\mathcal{E}_R^{(i)} = 
\begin{bmatrix}
\mathcal{E}_{i-1}^{i}\!+\mathcal{E}_{i}^{i} & \mathcal{E}_{i}^{i+1}
\\[2pt]
\mathcal{E}_{i+1}^{i} & \mathcal{E}_{i+1}^{i+1} & \kern5pt \mathcal{E}_{i+1}^{i+2}
\\[2pt]
& \mathcal{E}_{i+2}^{i+1} & \kern5pt \mathcal{E}_{i+2}^{i+2} & \kern5pt \mathcal{E}_{i+2}^{i+3}
\\[-2pt]
& & \kern-20pt \ddots & \kern-20pt \ddots & \kern-5pt \ddots
\end{bmatrix}.
\eeq
The splitting rules \eqref{eq:facdec} imply that the corresponding Schur functions for site $i$ satisfy
\beq \label{eq:fi}
f^{(i)} = f^{(i)}_L + f^{(i)}_R -\mathcal{E}^{(i)}_0.
\eeq
This alternative approach {\color{black}gives rise to} the following result, which reveals a surprising independence of the return probability to a site with respect to the details of the transition CP maps connecting to the previous sites. 

\begin{theorem}
Suppose that a nearest neighbour TOM $\mathcal{E}$ on the half-line, with a finite-dimensional Hilbert space of internal degrees of freedom, is symmetric --i.e. $\mathcal{E}_i^j=\mathcal{E}_j^i$ for every $i,j$-- and unital. Then, for any state $\rho$ supported on site $i$, the return probability $\pi(\rho\to|i\rangle)$ is invariant under any change in the first $i$ columns of $\mathcal{E}$ which keeps it as a unital symmetric nearest neighbour TOM.   
\end{theorem}

\begin{proof}
According to Theorem~\ref{thm:dec}, 
\beq \label{eq:sym-od}
\pi(\rho\to|i\rangle) = 
\pi_L^{(i)}(\rho\to|i\rangle)+\pi_R^{(i)}(\rho\to|i\rangle)-1,
\eeq
where the subscripts $L/R$ and the superscripts $(i)$ refer to the TOMs in \eqref{eq:od-nn}. If $\mathcal{E}$ is unital and symmetric,
$$
\mathcal{E}_k^{k-1}(I) + \mathcal{E}_k^k(I) + \mathcal{E}_{k+1}^k(I) =
\mathcal{E}_k^{k-1}(I) + \mathcal{E}_k^k(I) + \mathcal{E}^{k+1}_k(I) = I.
$$ 
This implies that the finite-dimensional TOM $\mathcal{E}_L^{(i)}$ in \eqref{eq:od-nn} is unital. Then, from Theorem~\ref{thm:rec-ui} we know that any subspace is positive recurrent for $\mathcal{E}_L^{(i)}$, so that \eqref{eq:sym-od} becomes the identity $\pi(\rho\to|i\rangle)=\pi_R^{(i)}(\rho\to|i\rangle)$, which proves the result because $\pi_R^{(i)}(\rho\to|i\rangle)$ only depends on the CP maps in $\mathcal{E}_R^{(i)}$.
\end{proof} 

The following examples illustrate the use of these techniques to study site recurrence for a nearest neighbour TOM on the half-line.

\bex \label{ex:HL}
Suppose that a TOM on the half-line with a finite-dimensional Hilbert space $\mathcal{H}$ of internal degrees of freedom has the structure
\beq \label{eq:ex4-3dqc}
\mathcal{E} = 
\begin{bmatrix}
\Phi_0 & \lambda\Phi_+
\\
\Phi_- & 0 & \kern-5pt \lambda\Phi
\\
& (1-\lambda)\Phi & \kern-5pt 0 & \kern5pt \lambda\Phi
\\
& & \kern-5pt (1-\lambda)\Phi & \kern5pt 0 & \kern15pt \lambda\Phi
\\
& & & \kern5pt \ddots & \kern15pt \ddots & \kern15pt \ddots
\end{bmatrix},
\qquad 
\begin{aligned}
& \lambda\in(0,1),
\\
& \Phi, \Phi_\pm \; \text{invertible},
\end{aligned}
\eeq
where $\Phi_+$ and $\Phi$ are CPTP maps, while $\Phi_0$ and $\Phi_-$ are CP maps such that $\Phi_0+\Phi_-$ is trace preserving. In view of \eqref{eq:f-f1}, the Schur function $f=f^{(0)}$ for site 0 is given by
\beq \label{eq:ex4-eqf}
f(z) = \Phi_0 + \lambda z\Phi_+(I-zf_1(z))^{-1}\Phi_-,
\eeq
where $f_1$ is the Schur function of site 1 with respect to the TOM on $\mathcal{I}_{\mathcal{S}_1}(\mathcal{H})$ given by
$$
\mathcal{E}_1 = 
\begin{bmatrix}
0 & \kern-5pt \lambda\Phi \\
(1-\lambda)\Phi & \kern-5pt 0 & \kern5pt \lambda\Phi \\
& \kern-5pt (1-\lambda)\Phi & \kern5pt 0 & \kern15pt \lambda\Phi \\
& & \kern5pt \ddots & \kern15pt \ddots & \kern15pt \ddots
\end{bmatrix}.
$$ 
Therefore, the iterates $f_1$ and $f_2$ coincide, which, according to \eqref{eq:f1-f}, implies that
\beq \label{eq:ex4-eqf1}
f_1(z) = z^{-1}I-\lambda(1-\lambda)\Phi f_1(z)^{-1}\Phi.
\eeq
This may be rewritten as
$$
\Phi^{-1} = z\left(f_1(z)\Phi^{-1} + 
\lambda(1-\lambda)\left(f_1(z)\Phi^{-1}\right)^{-1}\right),
$$
showing that $\Phi^{-1}$ commutes with $f_1(z)\Phi^{-1}$, which means that $\Phi$ commutes with $f_1(z)$.
Therefore, \eqref{eq:ex4-eqf1} is equivalent to the quadratic equation
$$
f_1(z)^2 - z^{-1}f_1(z) + \lambda(1-\lambda)\Phi^2 = 0,
$$
whose solution has the form
\beq \label{eq:ex4-f1}
f_1(z) = \frac{1}{2z} \left(I-\sqrt{I-4\lambda(1-\lambda)z^2\Phi^2}\right).
\eeq
Since $f_1$ must be analytic at the origin, in the above expression the square root corresponds to the analytic branch determined by $\sqrt{I-\lambda(1-\lambda)z^2\Phi^2} \xrightarrow{z\to0} I$, so that
\beq \label{eq:ex4-f1pot}
f_1(z) = -\frac{1}{2} \sum_{n\ge1} \binom{1/2}{n} z^{2n-1} \Psi^n,
\qquad
\Psi = 4\lambda(1-\lambda)\Phi^2.
\eeq
This square root is indeed analytic on the open unit disk because $\lambda(1-\lambda)\le1/4$ and $\|\Phi\|=1$ because $\Phi$ is CPTP. This is in agreement with the analyticity domain for the Schur function $f_1$. Inserting into \eqref{eq:ex4-eqf} the relation \eqref{eq:ex4-eqf1}, rewritten as $ \lambda(1-\lambda)z(I-zf_1(z))^{-1} = \Phi^{-1}f_1(z)\Phi^{-1} $, we finally find that    
\beq \label{eq:ex4-f}
f(z) = \Phi_0 + 
\frac{1}{1-\lambda} \Phi_+\Phi^{-1} f_1(z) \Phi^{-1}\Phi_- = 
\Phi_0 + \frac{1}{2(1-\lambda)z} \Phi_+\Phi^{-1} 
\left(I-\sqrt{I-4\lambda(1-\lambda)z^2\Phi^2}\right) \Phi^{-1}\Phi_-.
\eeq
This yields the Schur function $f=f^{(0)}$ of site 0. From \eqref{eq:ex4-f1pot}, \eqref{eq:ex4-f} and the fact that $\Phi$ and $\Phi_+$ are trace preserving, we conclude that, for any state $\rho$ supported on site 0, 
\beq \label{eq:ex4-f1rho}
\begin{aligned}
\operatorname{Tr}(f_1(z)(\rho)) 
& = \frac{1}{2z} \left(1-\sqrt{1-4\lambda(1-\lambda)z^2}\right),
\\
\operatorname{Tr}(f(z)(\rho)) 
& = \operatorname{Tr}(\Phi_0(\rho)) + 
\frac{1-\sqrt{1-4\lambda(1-\lambda)z^2}}{2(1-\lambda)z}
\operatorname{Tr}(\Phi_-(\rho)) 
\\
& = 1 - 
\left(1 - \frac{1-\sqrt{1-4\lambda(1-\lambda)z^2}}{2(1-\lambda)z} \right)
\operatorname{Tr}(\Phi_-(\rho)).
\end{aligned}
\eeq
This result leads to the return probability
$$
\begin{aligned}
\pi(\rho\to|0\rangle) 
& = \operatorname{Tr}(f(1)(\rho)) 
= 1 - \left[1 - \frac{1-|1-2\lambda|}{2(1-\lambda)} \right]
\operatorname{Tr}(\Phi_-(\rho)) 
= 1 - \frac{1-2\lambda+|1-2\lambda|}{2(1-\lambda)} 
\operatorname{Tr}(\Phi_-(\rho))
\\
& = \begin{cases}
1 - \frac{1-2\lambda}{1-\lambda} \operatorname{Tr}(\Phi_-(\rho)), 
& \lambda<\frac{1}{2},
\\[2pt]
1, & \lambda\ge\frac{1}{2}.
\end{cases}
\end{aligned}
$$
Since $\Phi_-$ is invertible, $\operatorname{Tr}(\Phi_-(\rho))\ne0$, so we find that $\pi(\rho\to|0\rangle)=1$ iff $\lambda\ge1/2$, in which case,
$$
\tau(\rho\to|0\rangle) = 1 + \lim_{x\uparrow1}\operatorname{Tr}(f'(x)(\rho)) 
= 1 + \frac{1}{2(1-\lambda)} \left(\frac{1}{|1-2\lambda|}-1\right) \operatorname{Tr}(\Phi_-(\rho)) = 
1 + \frac{\operatorname{Tr}(\Phi_-(\rho))}{2\lambda-1}.
$$
Hence, site 0 is recurrent iff $\lambda\ge1/2$, otherwise no state returns with probability one to site 0. Also, site 0 is positive recurrent iff $\lambda>1/2$, the expected return time being infinite for every state otherwise.   

Bearing in mind that $f_i=f_1$ for $i\ge1$, we conclude that the Schur function of a site $i\ge1$ is given by $f^{(i)}=f_i^-+f_1$, where $f_i^-$ is the Schur function of site $i$ for the map $\mathcal{E}_i^-$ given by the principal $(i+1)\times(i+1)$ submatrix of $\mathcal{E}$. To illustrate this result consider the case of site 1, for which we find that
$$
\mathcal{E}_1^- = 
\begin{bmatrix}
\Phi_0 & \lambda\Phi_+ \\[2pt] \Phi_- & 0
\end{bmatrix},
\qquad
f_1^-(z) = \lambda z\Phi_-(1-z\Phi_0)^{-1}\Phi_+.
$$
The CP maps $\Phi_0$ and $\Phi_-$ must be trace non-increasing because $\Phi_0+\Phi_-$ is trace preserving, thus $\|\Phi_0\|,\|\Phi_-\|\le1$ according to \eqref{eq:CPnorm}. Since we assume $\dim\mathcal{H}<\infty$, the relation \eqref{eq:CPnorm} also shows that $\|\Phi_0\|=1$ would imply that $\operatorname{Tr}(\Phi_0(\rho_*)) = 1$ for some state $\rho_*$. Combined with the fact that $\Phi_0+\Phi_-$ is trace preserving, this yields $\operatorname{Tr}(\Phi_-(\rho_*)) = 0$, so that $\Phi_-(\rho_*)=0$ in contradiction with the invertibility of $\Phi_-$. We conclude that $\|\Phi_0\|<1$, hence $I-\Phi_0$ is invertible and $\rho \mapsto \sigma=(I-\Phi_0)(\rho)$ maps $\mathcal{I}(\mathcal{H})$ one to one onto itself. Therefore, the trace preserving condition
$\operatorname{Tr}(\Phi_-(\rho)) = \operatorname{Tr}((I-\Phi_0)(\rho))$ 
also reads as
$\operatorname{Tr}(\Phi_-(I-\Phi_0)^{-1}(\sigma)) = \operatorname{Tr}(\sigma)$, i.e. $\Phi_-(I-\Phi_0)^{-1}$ is trace preserving. Using this result and the first equality in \eqref{eq:ex4-f1rho}, bearing in mind that $f^{(1)}=f_1^-+f_1$ and  $\Phi_+$ is trace preserving, we get for any state $\rho$ supported on site 1,
\beq \label{eq:ex4-piHL}
\pi(\rho\to|1\rangle) = 
\operatorname{Tr}(f_1^-(1)(\rho)) + \operatorname{Tr}(f_1(1)(\rho))
= \lambda + \frac{1}{2} \left(1-|1-2\lambda|\right) =
\begin{cases}
2\lambda, & \lambda<\frac{1}{2},
\\[2pt]
1, & \lambda\ge\frac{1}{2}.
\end{cases}
\eeq
Therefore, $\pi(\rho\to|1\rangle)=1$ iff $\lambda\ge1/2$. Then,
$$
\begin{aligned}
\tau(\rho\to|1\rangle)  
& = 1 + \operatorname{Tr}((f_1^-)'(1)(\rho)) + \operatorname{Tr}(f'_1(1)(\rho))
= 1 + \lambda \operatorname{Tr}((I-\Phi_0)^{-1}\Phi_+(\rho)) + 
\frac{1}{2} \left(\frac{1}{|1-2\lambda|}-1\right)
\\ 
& = \lambda\left(\frac{1}{2\lambda-1} + 
\operatorname{Tr}((I-\Phi_0)^{-1}\Phi_+(\rho))\right),
\end{aligned}
$$
so that $\tau(\rho\to|1\rangle)<\infty$ when $\lambda>1/2$. We conclude that site 1 is recurrent iff $\lambda\ge1/2$, the return probability being independent of the sate in any case. In contrast, the mean return time depends on the sate when $\lambda>1/2$, the only case which makes site 1 to be positive recurrent  since otherwise the expected return time is infinite for every state.

A difference between the recurrence properties of sites 0 and 1 is that in the later case the return probabilities do not depend on the first column of $\mathcal{E}$, i.e. any choice of $\Phi_0$ and $\Phi_-$ --with $\Phi_-$ invertible and $\Phi_0+\Phi_-$ trace preserving-- yields the same probability $\pi(\rho\to|1\rangle)$ for each state $\rho$. The splitting rules shed light on this invariance (see Remark~\ref{rem:indep}). In the present case, the overlapping decomposition \eqref{eq:od-nn} is given by
$$
\mathcal{E}_L^{(1)} = 
\begin{bmatrix}
\Phi_0 & \lambda\Phi_+ \\[2pt] \Phi_- & (1-\lambda)\Phi
\end{bmatrix},
\kern15pt
\mathcal{E}_R^{(1)} = 
\begin{bmatrix}
\lambda\Phi_+ & \kern-5pt \lambda\Phi \\
(1-\lambda)\Phi & \kern-5pt 0 & \kern5pt \lambda\Phi \\
& \kern-5pt (1-\lambda)\Phi & \kern5pt 0 & \kern15pt \lambda\Phi \\
& & \kern5pt \ddots & \kern15pt \ddots & \kern15pt \ddots
\end{bmatrix},
\kern15pt
\mathcal{E}_0^{(1)} = \lambda\Phi_+ + (1-\lambda)\Phi.
$$
From Theorem~\ref{thm:dec} we know that the corresponding return probabilities are related by 
$$
\pi(\rho\to|1\rangle) = 
\pi_L^{(1)}(\rho\to|1\rangle) + \pi_R^{(1)}(\rho\to|1\rangle) - 1.
$$
Let us have a look at the return probabilities for the left and right TOMs $\mathcal{E}_{L/R}^{(1)}$. From their definitions, \eqref{eq:f-f1} and \eqref{eq:ex4-eqf1} it follows that the Schur functions of site 1 for the left and right TOMs are given by 
$$
f_L^{(1)}(z) = (1-\lambda)\Phi + \lambda z\Phi_-(I-z\Phi_0)^{-1}\Phi_+,
\qquad
f_R^{(1)}(z) = \lambda\Phi_+ + \lambda(1-\lambda) z\Phi(I-zf_1(z))^{-1}\Phi =
\lambda\Phi_+ + f_1(z).
$$
Note that $f^{(1)}=f_L^{(1)}+f_R^{(1)}-\mathcal{E}_0^{(1)}$, in agreement with \eqref{eq:fi}. Since $\Phi_-(I-\Phi_0)^{-1}$ and $\Phi_+$ are trace preserving,  so is $f_L^{(1)}(1)$, hence site 1 is recurrent with respect to $\mathcal{E}_L^{(1)}$. In consequence, 
$$
\pi(\rho\to|1\rangle) = \pi_R^{(1)}(\rho\to|1\rangle) = 
\operatorname{Tr}(f_R^{(1)}(1)(\rho)), 
$$
which not only reproduces \eqref{eq:ex4-piHL}, but also explains that $\pi(\rho\to|1\rangle)$ is independent of $\Phi_0$ and $\Phi_-$ because of the recurrence of site 1 with respect to the left TOM (see Remark~ \ref{rem:indep}). Note however that this independence does not hold for the expected return time, for which the identity 
$$
\tau(\rho\to|1\rangle) = 
\tau_L^{(1)}(\rho\to|1\rangle) + \tau_R^{(1)}(\rho\to|1\rangle) - 1,
$$ 
predicted by Theorem~\ref{thm:dec}, can be easily checked starting from the previous expressions of $f_{L/R}^{(1)}$.

\eex
\qee

\subsection{Nearest neighbour TOMs on the line}
\label{ssec:LINE}
The study of TOMs on the line with the structure
\beq \label{eq:3d-L}
\mathcal{E} = 
\begin{bmatrix}
\ddots& \ddots & \ddots \\  
& \mathcal{E}_{-1}^{-2} & \mathcal{E}_{-1}^{-1} & \mathcal{E}_{-1}^0 \\[2pt] 
& & \mathcal{E}_0^{-1} & \mathcal{E}_0^0 & \kern2pt \mathcal{E}_0^1 \\[2pt]  
& & & \mathcal{E}_1^0 & \kern2pt \mathcal{E}_1^1 & \kern2pt \mathcal{E}_1^2 \\
& & & & \kern2pt \ddots & \kern2pt \ddots & \kern2pt \ddots
\end{bmatrix},
\qquad
\mathcal{E}^{i\pm1}_i \; \text{invertible},
\eeq
may be reduced to the half-line case \eqref{eq:3d-HL} via a classical folding trick \cite{berez} which combines sites $i$ and $-i-1$ into a single one by doubling the corresponding site space $\mathcal{I}(\mathcal{H})$ to $\mathcal{I}(\mathcal{H})\oplus\mathcal{I}(\mathcal{H})$. This amounts to rewriting the TOM \eqref{eq:3d-L} as map 
$$
\boldsymbol{\mathcal{E}} = 
\begin{bmatrix} 
\boldsymbol{\mathcal{E}}_0^0 & \boldsymbol{\mathcal{E}}_0^1 & \\[2pt] 
\boldsymbol{\mathcal{E}}_1^0 & \boldsymbol{\mathcal{E}}_1^1 & \boldsymbol{\mathcal{E}}_1^2 \\[2pt]  
& \boldsymbol{\mathcal{E}}_2^1 & \boldsymbol{\mathcal{E}}_2^2 & \boldsymbol{\mathcal{E}}_2^3 \\[-2pt]
& & \ddots & \ddots & \ddots
\end{bmatrix},
\qquad
\boldsymbol{\mathcal{E}}_i^j = 
\begin{bmatrix}
\mathcal{E}_{-i-1}^{-j-1} & \mathcal{E}_{-i-1}^{j}
\\[3pt]
\mathcal{E}_{i}^{-j-1} & \mathcal{E}_{i}^{j}
\end{bmatrix},
$$
acting on column vectors
$$
\boldsymbol{\rho} = 
\begin{bmatrix}
\boldsymbol{\rho}_0 \\ \boldsymbol{\rho}_1 \\ \vdots
\end{bmatrix},
\qquad \boldsymbol{\rho}_i =
\begin{bmatrix}
\rho_{-i-1} \\ \rho_i
\end{bmatrix}.
$$

Instead of following this folding approach, we will see that the half-line reduction may be simplified for recurrence problems by using splitting techniques. Consider the decomposition which overlaps at site 0 given by
$\mathcal{E} = (\mathcal{E}^-\oplus\mathbb{I}) + (\mathbb{I}\oplus\mathcal{E}^+)$, where
\beq \label{eq:EpmL}
\mathcal{E}^- = 
\begin{bmatrix}
\ddots & \ddots & \ddots & 
\\
& \mathcal{E}_{-1}^{-2} & \mathcal{E}_{-1}^{-1} & \mathcal{E}_{-1}^{0} 
\\[2pt]
& & \mathcal{E}_{0}^{-1} & 0
\end{bmatrix},
\qquad
\mathcal{E}^+ = 
\begin{bmatrix}
\mathcal{E}_{0}^{0}  & \mathcal{E}_{0}^{1}
\\[2pt]
\mathcal{E}_{1}^{0} & \mathcal{E}_{1}^{1} & \mathcal{E}_{1}^{2}
\\[-2pt]
& \ddots & \ddots & \ddots
\end{bmatrix}.
\eeq
The splitting rules \eqref{eq:facdec} imply that the corresponding Schur functions for site 0 are related by 
$$
f=f^-+f^+.
$$
Under the relabelling $i\to-i$ of the sites, the map $\mathcal{E}^-$ reads as 
\beq \label{eq:reord}
\begin{bmatrix}
0  & \mathcal{E}_{0}^{-1}
\\[2pt]
\mathcal{E}_{-1}^{0} & \mathcal{E}_{-1}^{-1} & \mathcal{E}_{-1}^{-2}
\\[-2pt]
& \ddots & \ddots & \ddots
\end{bmatrix},
\eeq
hence $f^\pm$ may be analyzed using the tricks developed for the half-line.

Alternatively, one can use an overlapping decomposition into true TOMs on the half-line, namely, $\mathcal{E} = (\mathcal{E}_L\oplus\mathbb{I}) + (\mathbb{I}\oplus\mathcal{E}_R)-(\mathbb{I}\oplus\mathcal{E}_0\oplus\mathbb{I})$, with $\mathcal{E}_0=\mathcal{E}_{-1}^{0}+\mathcal{E}_{0}^{0}+\mathcal{E}_{1}^{0}$ and
\beq \label{eq:od-nn-L}
\mathcal{E}_L = 
\begin{bmatrix}
\ddots & \ddots & \ddots & 
\\
& \mathcal{E}_{-1}^{-2} & \mathcal{E}_{-1}^{-1} & \kern-5pt \mathcal{E}_{-1}^{0} 
\\[2pt]
& & \mathcal{E}_{0}^{-1} & \kern-5pt \mathcal{E}_{0}^{0}\!+\mathcal{E}_{1}^{0}
\end{bmatrix},
\qquad
\mathcal{E}_R = 
\begin{bmatrix}
\mathcal{E}_{-1}^{0}\!+\mathcal{E}_{0}^{0} & \kern-5pt \mathcal{E}_{0}^{1}
\\[2pt]
\mathcal{E}_{1}^{0} & \kern-5pt \mathcal{E}_{1}^{1} & \kern7pt \mathcal{E}_{1}^{2}
\\[-2pt]
& \kern-3pt \ddots & \kern7pt \ddots & \kern7pt \ddots
\end{bmatrix}.
\eeq
This not only permits us to express the Schur function $f$ for site 0 on the line in terms of those $f_{L/R}$ of the left and right TOMs on the half-line,  
$$
f = f_L + f_R - \mathcal{E}_0,
$$
but also allows us to use directly the results of Theorem~\ref{thm:dec}.

\bex \label{ex:L}
Assume that a TOM on the line with a finite-dimensional Hilbert space $\mathcal{H}$ of internal degrees of freedom has the shape
\beq \label{eq:ex4-3dqc-L}
\mathcal{E} = 
\begin{bmatrix}
\kern-20pt \ddots & \kern-50pt \ddots & \kern-80pt \ddots
\\
(1-\lambda)\Phi & \kern-5pt 0 & \kern-5pt \lambda\Phi
\\
& \kern-5pt (1-\lambda)\Phi & \kern-5pt 0 & \kern-5pt \lambda\Phi_-
\\
& & \kern-5pt (1-\lambda)\Phi_+ & \kern-5pt \Phi_0 & \kern-5pt \lambda\Phi_+
\\
& & & \kern-5pt (1-\lambda)\Phi_- & \kern-5pt 0 & \lambda\Phi
\\
& & & & \kern-5pt (1-\lambda)\Phi & \kern5pt 0 & \kern15pt \lambda\Phi
\\
& & & & & \ddots & \ddots & \ddots
\end{bmatrix},
\qquad 
\begin{aligned}
& \lambda\in(0,1),
\\
& \Phi, \Phi_\pm \; \text{invertible},
\end{aligned}
\eeq
with $\Phi$, $\Phi_+$ CPTP maps and $\Phi_0$, $\Phi_-$ CP maps such that $\Phi_0+\Phi_-$ is trace preserving. Labelling the sites so that $\Phi_0$ is the self-transition CP map for site 0, \eqref{eq:EpmL} becomes
$$
\mathcal{E}^- = 
\begin{bmatrix}
\kern-20pt \ddots & \kern-50pt \ddots & \kern-80pt \ddots
\\
(1-\lambda)\Phi & \kern-5pt 0 & \kern-5pt \lambda\Phi
\\
& \kern-5pt (1-\lambda)\Phi & \kern-5pt 0 & \lambda\Phi_-
\\
& & \kern-5pt (1-\lambda)\Phi_+ & 0
\end{bmatrix},
\qquad
\mathcal{E}^+ = 
\begin{bmatrix}
\Phi_0 & \kern-5pt \lambda\Phi_+
\\
(1-\lambda)\Phi_- & \kern-5pt 0 & \lambda\Phi
\\
& \kern-5pt (1-\lambda)\Phi & \kern5pt 0 & \kern15pt \lambda\Phi
\\
& & \ddots & \ddots & \ddots
\end{bmatrix}.
$$
Resorting to the techniques developed for the half-line and performing the reordering which transforms $\mathcal{E}^-$ into \eqref{eq:reord}, we find the following Schur functions for site 0,
$$
f^-(z) =  \lambda(1-\lambda)z \Phi_+(I-zf_1(z))^{-1}\Phi_- =
\Phi_+\Phi^{-1}f_1(z)\Phi^{-1}\Phi_-,
\qquad
f^+ =  \Phi_0 + f^{-},
\qquad
f = \Phi_0 + 2f^-,
$$
with $f_1$ given by \eqref{eq:ex4-f1}. Therefore, for any state $\rho$ at site 0,
$$
\operatorname{Tr}(f(z)(\rho)) = \operatorname{Tr}(\Phi_0(\rho)) + 
\frac{1-\sqrt{1-4\lambda(1-\lambda)z^2}}{z} \operatorname{Tr}(\Phi_-(\rho)) = 
1 - \left(1-\frac{1-\sqrt{1-4\lambda(1-\lambda)z^2}}{z}\right) 
\operatorname{Tr}(\Phi_-(\rho)) 
$$
which leads to the return probability
\beq \label{eq:ex4-pi-L}
\pi(\rho\to|0\rangle) = \operatorname{Tr}(f(1)(\rho)) =
1 - |1-2\lambda| \operatorname{Tr}(\Phi_-(\rho)).
\eeq
Now, site 0 is recurrent only for $\lambda=1/2$, otherwise no state returns to site 0 with probability one due to the invertibility hypothesis for $\Phi_-$. The mean return time for $\lambda=1/2$ is
$$
\pi(\rho\to|0\rangle) = 1 + \lim_{x\uparrow1}\operatorname{Tr}(f'(x)(\rho)) =
\infty,
$$
thus no state has a finite expected return time to site 0.

Decomposing the TOM on the line into true TOMs on the half-line overlapping at site 0, leads to the left and right TOMs
$$
\mathcal{E}_L = 
\begin{bmatrix}
\kern-20pt \ddots & \kern-50pt \ddots & \kern-80pt \ddots
\\
(1-\lambda)\Phi & \kern-5pt 0 & \kern-5pt \lambda\Phi
\\
& \kern-5pt (1-\lambda)\Phi & \kern-5pt 0 & \kern-5pt \lambda\Phi_-
\\
& & \kern-5pt (1-\lambda)\Phi_+ & \kern-5pt \Phi_0+(1-\lambda)\Phi_-
\end{bmatrix},
\qquad
\mathcal{E}_R = 
\begin{bmatrix}
\Phi_0+\lambda\Phi_- & \kern-5pt \lambda\Phi_+
\\
(1-\lambda)\Phi_- & \kern-5pt 0 & \lambda\Phi
\\
& \kern-5pt (1-\lambda)\Phi & \kern5pt 0 & \kern15pt \lambda\Phi
\\
& & \ddots & \ddots & \ddots
\end{bmatrix},
$$
with Schur functions for site 0 given by
$$
f_L = \Phi_0 + (1-\lambda)\Phi_- + f^{-},
\qquad
f_R = \Phi_0 + \lambda\Phi_- + f^{-},
\qquad
f = f_L + f_R - (\Phi_0+\Phi_-).
$$
This yields for any state $\rho$ at site 0,
$$
\pi_L(\rho\to|0\rangle) = 1 - \lambda\operatorname{Tr}(\Phi_-(\rho)) + \operatorname{Tr}(f^-(1)(\rho)),
\qquad
\pi_R(\rho\to|0\rangle) = 1 - (1-\lambda)\operatorname{Tr}(\Phi_-(\rho)) + \operatorname{Tr}(f^-(1)(\rho)),
$$
which reproduces \eqref{eq:ex4-pi-L} via the identity $\pi(\rho\to|0\rangle)=\pi_L(\rho\to|0\rangle)+\pi_R(\rho\to|0\rangle)-1$ from Theorem~\ref{thm:dec}. The expected return time is also recovered from the relation $\tau(\rho\to|0\rangle)=\tau_L(\rho\to|0\rangle)+\tau_R(\rho\to|0\rangle)-1$ in this theorem since $\tau_L(\rho\to|0\rangle)=\tau_R(\rho\to|0\rangle)=\infty$ for any state $\rho$ at site 0.

\eex
\qee

{\bf Acknowledgments.} CFL acknowledges financial support by PROPESQ/UFRGS 16535-2017 and Programa de Apoio \`a P\'os-Gradua\c c\~ao CAPES/PROAP 2018 (PPGMAT/UFRGS).  LVC has been supported in part by the research project MTM2017-89941-P from Ministerio de Econom\'{\i}a, Industria y Competitividad of Spain and the European Regional Development Fund (ERDF), and by Project E26\_17R of Diputaci\'{o}n General de Arag\'{o}n (Spain) and the ERDF 2014-2020 ``Construyendo Europa desde Arag\'on". The authors are grateful to J. Zubelli for his hospitality and financial support regarding our visit to IMPA, to the organizers of the Quantum Walks, Open Quantum Walks, Quantum Computation Session at the MCA 2017 Montr\'eal and to Departamento de Matem\'atica Aplicada and IUMA from Universidad de Zaragoza, where part of this work was carried out.

\end{document}